\documentclass{article}

\usepackage[a4paper]{geometry}
\usepackage{algorithm}
\usepackage{algorithmic}
\usepackage{amsfonts}
\usepackage{amsmath}
\usepackage{amsthm}
\usepackage{graphicx}
\usepackage{epstopdf}

\title{Nonparametric Bayesian label prediction on a large graph using truncated Laplacian 
regularization}
\author{Jarno Hartog and Harry van Zanten}
\date{April 10, 2018}

\newcommand{\given}{\,|\,}
\newcommand{\RR}{\mathbb{R}}

\theoremstyle{plain}

\newtheorem{lem}{Lemma}

\begin{document}

\maketitle

\begin{abstract}
This article describes an implementation of a nonparametric Bayesian approach to solving binary classification problems on graphs. We consider a hierarchical Bayesian approach with a prior that is 
constructed by truncating a series expansion of the soft label function using the graph Laplacian eigenfunctions as basisfunctions. 
We compare our truncated prior to the untruncated Laplacian based prior in simulated and real data examples to illustrate the improved scalability in terms of size of the underlying graph. 
\end{abstract}

\section{Introduction}

Classification problems on graphs arise in various contexts, including for instance prediction in protein-protein interaction graphs and graph-based semi-supervised learning. 
In this paper we consider problems in which the graph is considered as given and 
binary labels, say $0$ or $1$, 
are given at some of the vertices. The goal is to predict the missing labels. A statistical approach is to view this as a binary regression problem where from the available data, first the binary regression function, or soft label function $\ell$ is estimated which gives for every vertex $i$ the probability $\ell(i)$ that the vertex has label $1$. Subsequently, the estimated soft label function can be  used for prediction by thresholding it.

In many cases it is not natural, or desirable, to postulate a certain parametric form 
for the soft label function. Instead it is common to explicitly or implicitly 
assume some form of  `smoothness' of the function.  
The main idea behind nonparametric methods for this problem is to employ some form of 
regularization that exploits the fact that  the graph geometry should somehow help 
to predict the correct labels, in the sense that vertices that are `close' should 
have `similar' labels. 
Various methods  have
been considered in the literature,  including penalized least squares regression using a Laplacian-based penalty
(e.g.\ \cite{ando2007, belkin2004, kolaczyk2009, smola2003, zhu2005}), penalization using the total variation norm (e.g.\ \cite{ryan}) 
and Bayesian regularization (e.g.\ \cite{jarno}, \cite{andrew}, \cite{me}).

In this article we consider an extension to the method proposed in \cite{jarno} in the context of binary classification problems on graphs. We have noisy observations of the labels 
of part of the vertices of a large given graph and the goal is to classify all vertices correctly, including those for which there is no observation available. In \cite{jarno}, an implementation is provided of nonparametric Bayesian prediction on graphs using Gaussian priors based on the Laplacian on the graph. Using the eigendecomposition of the Laplacian, we can view this prior as a Gaussian series prior
\begin{equation}\label{eq: series}
f = \sum_{i=1}^{n} g_i u^{(i)},
\end{equation}
where $n$ is the number of vertices of the gaph,  $u^{(i)}$ are the eigenvectors of the graph Laplacian and $g_i$ Gaussian random variables for $i=1, \ldots, n$. 
As indicated in \cite{jarno}, using the full Laplacian, i.e.\ all $n$ eigenvectors, 
is computationally demanding and limits the applicability of a Bayes procedure
with this natural prior for very large graphs. 
In the present paper, we address this issue by truncating the series at a random point for computational efficiency. This leads to a number of practical issues regarding 
prior choices etcetera, which we address in a simulation study. We illustrate the 
improved scalability by considering an example involving a graph with $90,000$ nodes.

Another advantage of truncating the series \eqref{eq: series}
at a random point is that it yields a more flexible prior in terms of adaptation to smoothness. 
Theoretical results for random inverse-gamma scaling of series priors with Gaussian coefficients 
and random truncation are given in \cite{vanwaaij2016} in the context of signal in white noise and estimating the drift function of a diffusion process. In these contexts it was shown that the truncated series prior with a geometric or Poisson prior on the truncation level achieves the optimal posterior contraction rate. Although in this work we are in a different setup where the results of \cite{vanwaaij2016} do not directly apply, we will also use a geometric prior and our proposed method will be a reversible jump Markov chain Monte Carlo algorithm similar to the method in \cite{vandermeulen2014} in the context of diffusion processes.

%

The rest of this paper is organized as follows. In the next section a more precise description of the problem setting and of the priors we consider are given. A sampling scheme to draw from the posterior distribution is given in Section 3 and some computational aspects are discussed in Section 4. In Section 5 we present numerical experiments. We first apply our method on a simple example on the path graph to illustrate the impact of the prior on the truncation level. In the next example we use 
data from the MNIST dataset to illustrate how the truncated prior is  more attractive than an untruncated prior in terms of computation time at a similar level of prediction accuracy. We also illustrate the impact of implicit regularization during the construction of the graph on the posterior draws and robustness of the prediction performance to changes in the construction of the graph. As a final example we apply our algorithm to a simple object tracking problem in a noisy environment, to further illustrate the improved scalability achieved by truncation. 
In Section 6 some concluding remarks are given.

\section{Observation model and priors}

\subsection{Observation model, latent variables and missing labels}

The context of our problem setup is the same as in \cite{jarno}. We have  
a given connected, simple graph $G = (V, E)$, with $\# V = n$ vertices,  denoted 
for simplicity by $V = \{1, 2, \ldots, n\}$. 
Associated to every vertex $i$ is a random, `hard label' $y_i \in \{0,1\}$. 
We assume that the variables $y_i$ are independent, so that their joint 
distribution is determined by the  
unobserved `soft label function' $\ell:V \to (0, 1)$ given by 
\[
\ell(i) = P(y_i = 1) = 1 - P(y_i = 0).
\]
The observed data is  $D = \{(i, y_i): i \in I^\text{obs}\}$, where $I^\text{obs}\subset V$ is drawn from an arbitrary distribution $\mu$ on the collection $2^V$ of subsets of vertices. The exact sampling mechanism $\mu$ is not important for the algorithm we propose, only that the subset is
independent of the labels.

Throughout we use the well-known latent variable perspective on this model
(cf. \cite{albert1993}). This is simply the observation that we can sample Bernoulli 
variables $y_1, \ldots,  y_n$ with succes probabilities $\ell(1), \ldots, \ell(n)$
using an intermediate layer of latent Gaussian variables. Indeed, 
let $\Phi$ be the probit link, i.e.\ the cdf of the standard normal distribution.
Then if $f: V\to \RR$ is given, sampling independent Bernoulli variables
$y_i$ with succes probabilities $\ell(i) = \Phi(f(i))$ can be achieved by susequently
  sampling independent Gaussian variables $z_i$ with mean $f(i)$ and variance $1$
and then setting $y_i = 1_{z_i > 0}$ for $i = 1, \ldots, n$.

\subsection{Prior on $f$}

The idea proposed in \cite{jarno} is essentially to 
 achieve a form of Bayesian  Laplacian regularization  in this problem by 
 putting a Gaussian prior on the function $f$ that determines the distribution 
 of the hard labels, with a precision matrix (inverse covariance) given 
 by a power of the graph Laplacian $L$. 
 The Laplacian is given by $L = A-D$, with $A$ the adjacency matrix of the graph 
 and $D$ the diagonal matrix of vertex degrees. It is a symmetric, non-negative 
 definite matrix. Since it always has eigenvalue $0$ however, it is not invertible, 
 so it has to be slight adapted before it can serve as a precision matrix. 
In \cite{jarno}  we made the matrix $L$ invertible by adding a small number $1/n^2$ 
 to the diagonal,  motivated by the result that the smallest nonzero eigenvalue of the 
 Laplacian is at least $4/n^2$ (Theorem 4.2 of \cite{mohar1991}). Adding a multiplicative scale parameter $c > 0$ and a hyperparameter $q \ge 0$ as well, \cite{jarno} proposed to employ the prior
\[
f \given c \sim N(0, (c(L + n^{-2}I)^q)^{-1}).
\]
Using the eigendecomposition of the Laplacian matrix $L = U \Lambda U^T$, with $\Lambda$ 
the matrix of Laplacian eigenvalues and $U$ the orthogonal matrix containing the corresponding 
eigenvectors, we can write $f = Ug$ for some vector $g$ and write the prior proposed 
in \cite{jarno} in series form as 
\[
f\given c  \sim \sum_{i=1}^{n} g_i u^{(i)},
\]
where $u^{(i)}$ is the $i$th eigenvector of $L$ and 
\[
g \given c \sim N(0, (c(\Lambda + n^{-2}I)^q)^{-1}).
\]

In the present paper we propose a prior that is 
more flexible and that improves scalability with the graph size $n$.  
We  truncate the series above at a random point $k$ that we will endow with an appropriate 
prior. Specifically, the prior we use in this paper can then be written as
\[
f\given k,c \sim \sum_{i=1}^k g_i u^{(i)},
\]
which depends on the random truncation level $k$ and random scale parameter $c$ via $g$. The prior on $g$ given $c$ and $k$ is in this case
\[
g \given k, c \sim N(0, (c(\Lambda_k + n^{-2}I)^q)^{-1}),
\]
where $\Lambda_k$ denotes the left upper block matrix given by the first $k$ rows and columns of $\Lambda$.

\subsection{Prior on $k$}

As we wish to express some preference for small models, i.e.\ low values for $k$, we use an exponential prior with rate $\gamma$ with probability mass function
\begin{equation}\label{eq: k}
P(k = l) \propto e^{-\gamma l}, \quad l = 1, \ldots, n.
\end{equation}
The rate $\gamma$ controls how strongly we prefer small models over large models, with 
the limiting case $\gamma \to 0$ giving uniform mass $1/n$ to all possible values $l=1, \ldots, n$. It can be seen that for every $l \in \{1, \ldots, n\}$, the prior \eqref{eq: k} on $k$ assigns 
mass
\[
(1-e^{-\gamma l})\frac{e^{\gamma n}}{e^{\gamma n} - 1}
\]
 to $\{1, \ldots, l\}$. For large graphs this is approximately $1-e^{-\gamma l}$ and can be used to set $\gamma$ in such a way that the prior is mostly concentrated on the first $l$ eigenvectors, possibly relieving the computational burden of having to compute all the eigenvalues. In some cases this might result in oversmoothing, but for large graphs it might simply be prohibitive to compute all the eigenvectors. 
 
 In our numerical experiments ahead we use the rule-of-thumb
 of setting $\gamma = 20/n$, unless otherwise stated.
 This corresponds to concentrating the  prior mass on the first eigenvectors. Specifically, 
 for this choice it holds that approximately $63\%$ of the prior mass is on the first $5\%$ of the eigenvectors, $86\%$ is on the first $10\%$ and $98\%$ is on the first $20\%$. 
Simulations indicate that this is an appropriate choice in many situations.

\subsection{Prior on $c$}

We use the natural choice of prior for $c$, which is a gamma prior with density
\[
p(c) \propto c^{a-1}e^{-bc}, \quad c>0
\]
for certain $a, b > 0$. This choice is motivated by the normal-inverse gamma partial conjugacy (see e.g. \cite{choudhuri2007, liang2007} in the context of our setting) and the positive results in the numerical experiments in \cite{jarno}. We can even choose the improper prior 
corresponding to $a=b=0$, in which case $p(c)\propto 1/c$.

\subsection{Full hierarchical model}

All in all, the full hierarchical scheme we work with is the following:
\begin{equation}\label{eq:fullmodel}
\begin{split}
D &= \{(i, y_i): i\in I^\text{obs}\},\\
I^\text{obs} &\sim \mu,\\
y_i &= 1_{z_i > 0}, \quad i = 1, \ldots, n,\\
z \given f &\sim N(f, I),\\
f &= \sum_{i=1}^k g_i u^{(i)},\\
g \given k, c &\sim N(0, (c(\Lambda_k + n^{-2}I)^q)^{-1}),\\
p(k) &\propto e^{-\gamma k},\\
p(c) &\propto c^{a-1}e^{-bc}.
\end{split}
\end{equation}
Our goal is to compute $f \given D$ and use it to predict the unobserved labels.

\section{Sampling scheme}

We will use a reversible jump Markov chain Monte Carlo algorithm (\cite{green1995}) to sample from $f \given D$ in the setup (\ref{eq:fullmodel}). This involves sampling repeatedly from the conditionals $p(z \given D, g, k, c)$,  $p(g, k \given D, z, c)$, and $p(c \given D, z, g, k)$. The joint move in $g$ and $k$ is the reversible jump step as $k$ is the dimension of $g$. We detail these three steps in the following subsections.

\subsection{Sampling from $p(z \given D, g, k, c)$}

As we identify $f = \sum_{i=1}^k g_i u^{(i)}$, we see that $z$ has the same full conditional as in the setup in \cite{jarno}. Given $D$, $f$ and $c$, the $z_i$'s are independent and 
\[
z_i \given D, f, c \sim
\begin{cases}
N(f_i, 1), & \text{if $i \not \in I^\text{obs}$}, \\
N_+(f_i, 1), & \text{if $i  \in I^\text{obs}$ and $y_i = 1$},\\
N_-(f_i, 1), & \text{if $i  \in I^\text{obs}$ and $y_i = 0$}.
\end{cases}
\]
Where $N_+$ and $N_-$ denote the normal distribution, conditioned to be positive or negative, respectively. Generating variables from these distribution can for example be done by a simple rejection algorithm or inversion (e.g. \cite{devroye1986}, see \cite{chopin} 
for a more refined analysis).

\subsection{Sampling from $p(g, k \given D, z, c)$}\label{sec:samplegk}

Since given $z$ we know all the $y_i$'s and $I^\text{obs}$ is independent of all other elements of the model, we have $p(g, k \given D, z, c) = p(g, k \given z, c)$. Due to the role of $k$ in the model we do not have conjugacy to draw from the exact conditional. Instead we use a reversible jump step. To this end we choose a proposal density $s(k' \given k)$. To generate a new draw for $g, k$ we propose the following steps:

\begin{itemize}
\item draw a proposal $k' \sim s(\cdot \given k)$;
\item draw an independent uniform random variable $v$ on $(0,1)$;
\item if 
\[
v \leq \frac{p(z \given k', c)p(k')s(k \given k')}{p(z \given k, c)p(k)s(k' \given k)},
\]
then accept the new proposal $k'$ and for $i = 1, \ldots k'$ draw
\[
g_i \sim N\left(\frac{z^Tu^{(i)}}{1 + c(\lambda_i + 1/n^2)^q}, \frac{1}{1 + c(\lambda_i + 1/n^2)^q}\right),
\]
otherwise retain the old draws $g$ and $k$.
\end{itemize}

We may choose a symmetric proposal distribution $s$, where, for example, the dimension can move a few steps up or down from the current level in a uniform, triangular or  binomial way. This is similar to a random walk proposal. In that case the ratio $s(k \given k')/s(k' \given k) = 1$. We may integrate to see that 
\begin{align*}
p(z|k, c) & = \int p(z|g, k, c)p(g|k, c)dg \\
& = (2\pi)^{-n/2} \left(\prod_{i=1}^k \frac{c (\lambda_i + 1/n^2)^q}{1 + c(\lambda_i + 1/n^2)^q} \right)^{1/2}e^{-\frac{1}{2}z^Tz + \frac{1}{2}\sum_{i=1}^k \frac{(z^Tu^{(i)})^2}{1 + c(\lambda_i + 1/n^2)^q}},
\end{align*}
resulting in the following three cases:

\begin{equation*}
\frac{p(z \given k', c)}{p(z \given k, c)} = 
\begin{cases}
\left(\prod_{i = k' + 1}^k \frac{1 + c(\lambda_i + 1/n^2)^q}{c(\lambda_i + 1/n^2)^q}\right)^{1/2} e^{-\frac{1}{2}\sum_{i = k'+1}^k \frac{z^Tu^{(i)}}{1 + c(\lambda_i + 1/n^2)^q}} & \text{if $k' < k$}, \\
1 &  \text{if $k' = k$}, \\
\left(\prod_{i = k + 1}^{k'} \frac{c(\lambda_i + 1/n^2)^q}{1 + c(\lambda_i + 1/n^2)^q}\right)^{1/2} e^{\frac{1}{2}\sum_{i = k+1}^{k'} \frac{z^Tu^{(i)}}{1 + c(\lambda_i + 1/n^2)^q}} & \text{if $k' > k$}.
\end{cases}
\end{equation*}

In our numerical experiments, we use $s(k' \given k) = k - 2 + s$, where $s \sim \text{Binom}(4, 0.5)$. In the following lemma we show detailed balance for this move, this implies that our proposed Markov chain has the correct stationary distribution (see e.g. \cite{brooks2011}).

\begin{lem}
The above proposed steps satisfy the relation
\[
p(g, k \given z, c)p((g, k) \to (g', k')) = p(g', k' \given z, c)p((g', k') \to (g, k)),
\]
where $p(A \to B)$ denotes the transition density from state $A$ to state $B$.
\end{lem}

\begin{proof}
The transition density from $(g, k)$ to $(g', k')$ is 
\[
p((g, k) \to (g', k')) = \min \left\{1, \frac{p(z \given k', c)p(k')s(k \given k')}{p(z \given k, c)p(k)s(k' \given k)}\right\} s(k' \given k) p(g' \given z, k', c).
\]
Note that if the minimum is less than $1$, the opposite move has a minimum larger than one. Using 
\[
p(g, k \given z, c) = p(g \given z, k, c) p(k \given z, c),
\]
and that the priors for $k$ and $c$ are independent, the assertion is verified. In case the minimum is greater than $1$ can be dealt with in a similar way.
\end{proof}

\subsection{Sampling from $p(c \given D, z, g, k)$}

We see that given $g$, $c$ is independent of the rest of the variables. In this case we have the usual normal-inverse gamma conjugacy giving
\[
c \given g, k \sim \Gamma\left(a + \frac{k}{2}, b + \frac{1}{2}\sum_{i = 1}^k (\lambda_i + 1/n^2)^q g_i^2\right).
\]

\subsection{Overview of sampling scheme}

For convenience we summarize our sampling scheme.

\begin{algorithm}[H]
\caption{Sampling scheme.}
\label{alg:mcmc}
\begin{algorithmic}[1]
\REQUIRE Data $D= \{(i, y_i) : i \in I^\text{obs}\}$, initial values $g = g^{(0)}$, $k = k^{(0)}$ and $c = c^{(0)}$.
\ENSURE MCMC sample from the joint posterior $p(z, g, k, c| D)$.
\REPEAT
\STATE 
Compute $f = \sum_{i=1}^k g_i u^{(i)}$ and for $i =1, \ldots, n$, draw independent
\[
z_i \sim \begin{cases}
N(f_i, 1), & \text{if $i \not \in I^\text{obs}$}, \\
N_+(f_i, 1), & \text{if $i  \in I^\text{obs}$ and $y_i = 1$},\\
N_-(f_i, 1), & \text{if $i  \in I^\text{obs}$ and $y_i = 0$}.
\end{cases}
\]
\STATE 
Draw a proposal $k' \sim s(\cdot \given k)$ and a uniform $v$ on $(0, 1)$.
\STATE 
\IF{ 
\[
v \leq e^{-\gamma(k'-k)}\frac{s(k \given k')}{s(k' \given k)}
\begin{cases}
\left(\prod_{i = k' + 1}^k \frac{1 + c(\lambda_i + 1/n^2)^q}{c(\lambda_i + 1/n^2)^q}\right)^{1/2} e^{-\frac{1}{2}\sum_{i = k'+1}^k \frac{z^Tu^{(i)}}{1 + c(\lambda_i + 1/n^2)^q}} & \text{if $k' < k$}, \\
1 &  \text{if $k' = k$}, \\
\left(\prod_{i = k + 1}^{k'} \frac{c(\lambda_i + 1/n^2)^q}{1 + c(\lambda_i + 1/n^2)^q}\right)^{1/2} e^{\frac{1}{2}\sum_{i = k+1}^{k'} \frac{z^Tu^{(i)}}{1 + c(\lambda_i + 1/n^2)^q}} & \text{if $k' > k$}.
\end{cases}
\]
}
\STATE
Set $k=k'$ and for $i = 1, \ldots k$ draw
\[
g_i \sim N\left(\frac{z^Tu^{(i)}}{1 + c(\lambda_i + 1/n^2)^q}, \frac{1}{1 + c(\lambda_i + 1/n^2)^q}\right),
\]
\ELSE
\STATE
Retain $g$ and $k$.
\ENDIF
\STATE 
Draw
\[
c \sim \Gamma\left(a + \frac{k}{2}, b + \frac{1}{2}\sum_{i = 1}^k (\lambda_i + 1/n^2)^q g_i^2\right). 
\]
\UNTIL{You have a large enough sample.}
\end{algorithmic}
\end{algorithm}

\section{Computational aspects}

If the underlying function $f$ is smooth enough that we can approximate it with only a few $k \ll n$ eigenvectors, then the proposed algorithm needs an initial investment of $O(kn^2)$ to compute the first $k$ eigenvalues and eigenvectors, in case 
these are not explicitly know for the graph under consideration. 
Step 6 in Algorithm \ref{alg:mcmc} has complexity $O(kn)$ and is the most expensive step. In principle, it could be that $k=n$ and our method would be as complex as the algorithm proposed in \cite{jarno}. However, for very large graphs it could be prohibitive to calculate the full eigendecomposition. One could compute a fixed number of eigenvalues and eigenvectors and if the Markov chain is about to step beyond this number one could either compute the next eigenvalue-eigenvector pair on the fly or reject the proposed $k$.

\section{Numerical results}

In this section we numerically assess scalability of the method and 
the sensitivity to the choice of the truncation level.

\subsection{Impact of the truncation level}

To assess the impact of the truncation level $\gamma$ we first consider a basic example of simulated data on the path graph with $n = 500$ vertices. In this case, the eigenvalues of the Laplacian matrix are $\lambda_k = 4 \sin^2(\pi (k-1) / (2n))$ with corresponding eigenvectors given by
\begin{equation}\label{eq:eigenvector}
u_i^{(k)} = \begin{cases}
\frac{\sqrt{2}}{\sqrt{n}} \cos \left(\frac{\pi(i-\frac{1}{2})k}{n}\right) & k = 2,\ldots, n,\\
\frac{1}{\sqrt{n}} & k=1,
\end{cases}
\end{equation}
for $i=1, \ldots, n$. We construct a function $f_0$ on the graph representing the ground truth 
by setting
\[
f_0 = \sum_{k=1}^n w_k u^{(k)},
\]
where we choose $w_k = \sqrt{n}(k-1)^{-1.5}\sin (k-1)$ for $k > 1$ and $w_1=0$. We simulate noisy labels $Y_i$ on the graph vertices satisfying $P(Y_i=1) = \ell_0(i) = \Phi(f_0(i))$, where $\Phi$ is the cdf of the standard normal distribution. Finally, we remove at random $20\%$ of the labels to generate the set of observed labels $Y^\text{obs}$. Figure \ref{fig:path} shows the resulting soft label function $\ell_0$ and the simulated noisy labels $Y_i$.

\begin{figure}[H]
\centering
\begin{minipage}{0.5\textwidth}
\includegraphics[width=\textwidth]{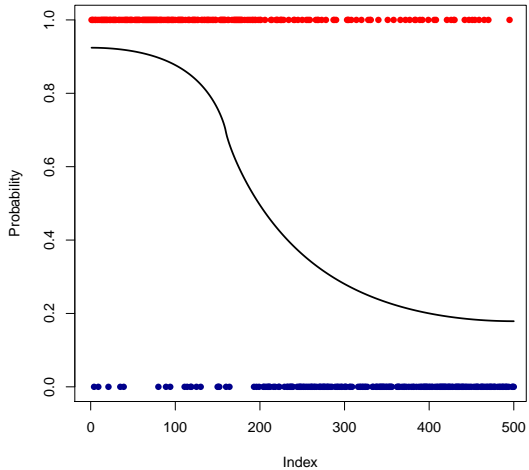}
\end{minipage}
\caption{Soft label function $\ell_0 = \Phi(f_0)$ and simulated noisy label on a path graph with $n=500$ nodes.}\label{fig:path}
\end{figure}

From the construction of our prior, we would like to spread out the mass in the prior on $c$ and perhaps favor low values in the prior on $k$ for computational efficiency. We can for example choose $a=b=0$, corresponding to an improper prior $p(c) \propto 1/c$ (as in \cite{choudhuri2007} and \cite{jarno}) or $a=1$ and $b=0$ so that $p(c) \propto 1$. From the construction of the prior on $k$ we see that high values for the parameter $\gamma$ correspond to more prior mass on low values of $k$ and low values spread out the prior mass over all possible values of $k$ with limiting case $\gamma = 0$ corresponding to $p(k)\propto 1$. In Figure \ref{fig:gamma1} we visualize the posterior for the soft label function $\ell$ for various values for $\gamma$. We have used the $a=b=0$ and proposal probabilities $(0.0625, 0.25, 0.375, 0.25, 0.0625)$ for $k-2, \ldots, k+2$, this corresponds to a $\text{Binom}(4, 0.5)$ proposal as mentioned in Section \ref{sec:samplegk}. 

The blue line is the posterior mean and the gray area depicts point-wise $95\%$ credible intervals. The bottom plots are the posterior draws for $k$. We observe that a high $\gamma$ results in low values for $k$, as expected. However, if we choose $\gamma$ too high, we might be oversmoothing as a result of taking too few eigenvectors. If we compare the cases $\gamma=0$ and $\gamma=0.1$ we observe only a little difference in the estimation performance, whereas the number of eigenvectors used in case of $\gamma=0.1$ is only a fraction of the number of eigenvectors used in case of $\gamma=0$.

\begin{figure}[H]
\begin{minipage}{0.33\textwidth}
\includegraphics[width=\textwidth]{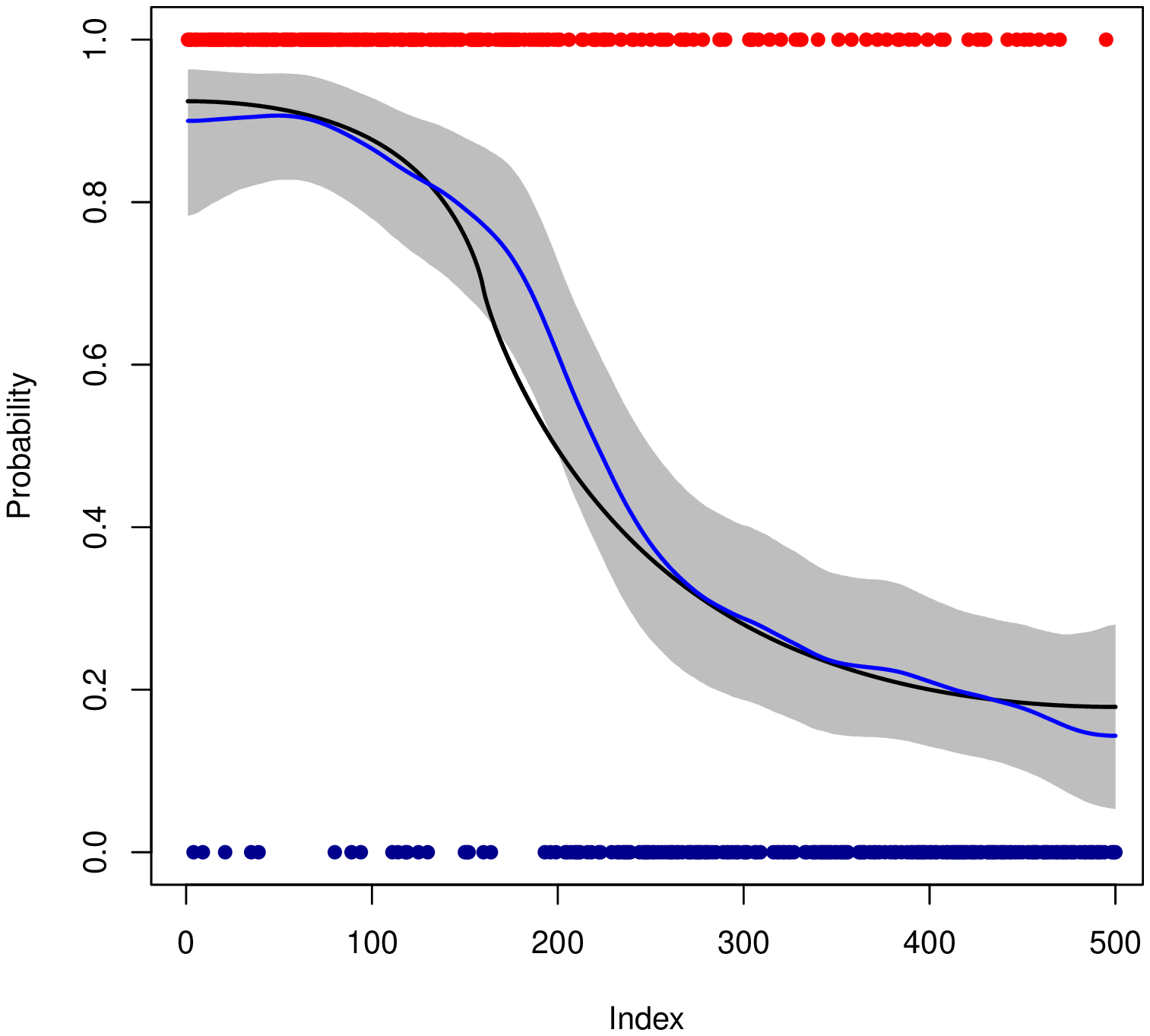}
\includegraphics[width=\textwidth]{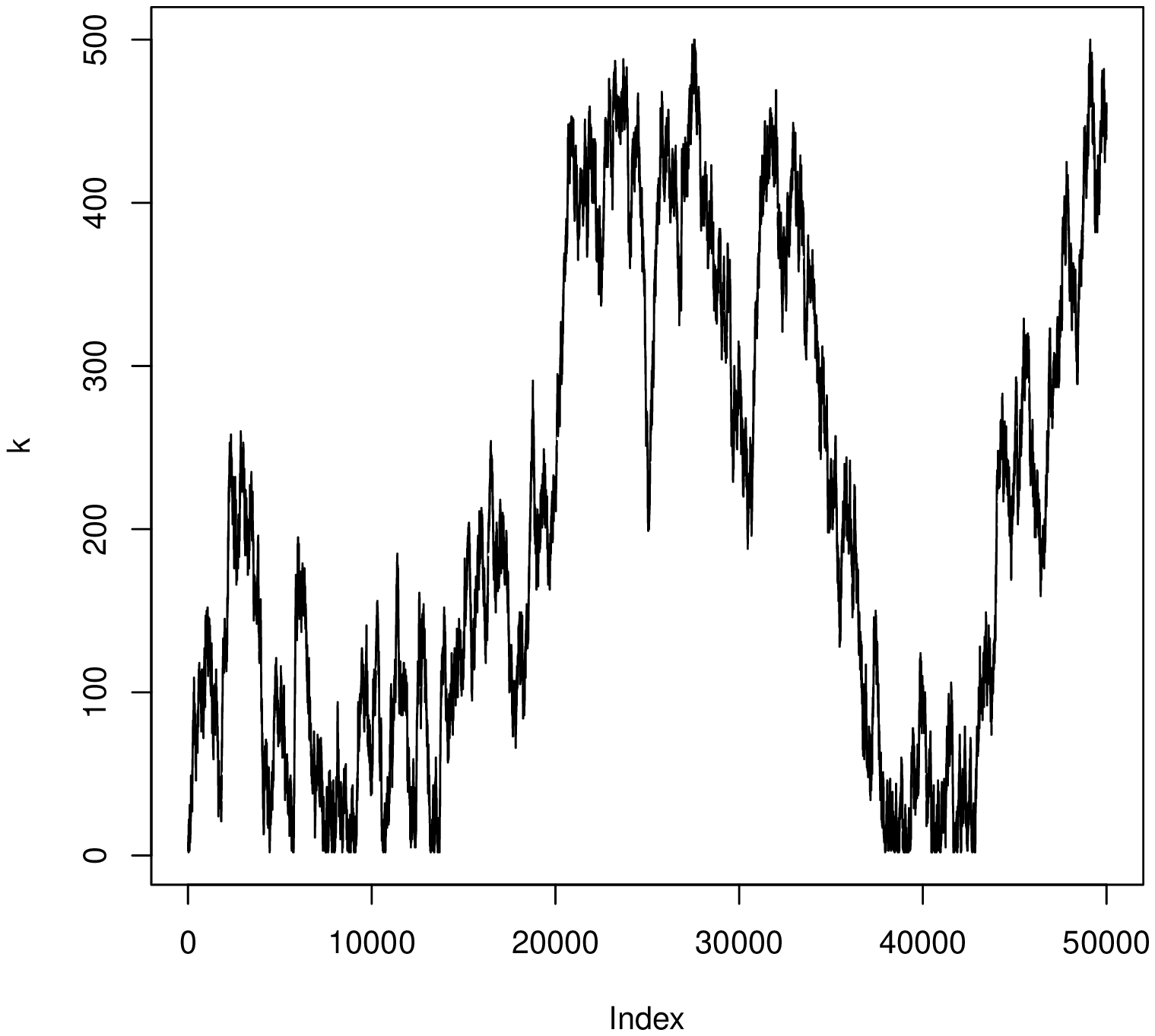}
\end{minipage}
\begin{minipage}{0.33\textwidth}
\includegraphics[width=\textwidth]{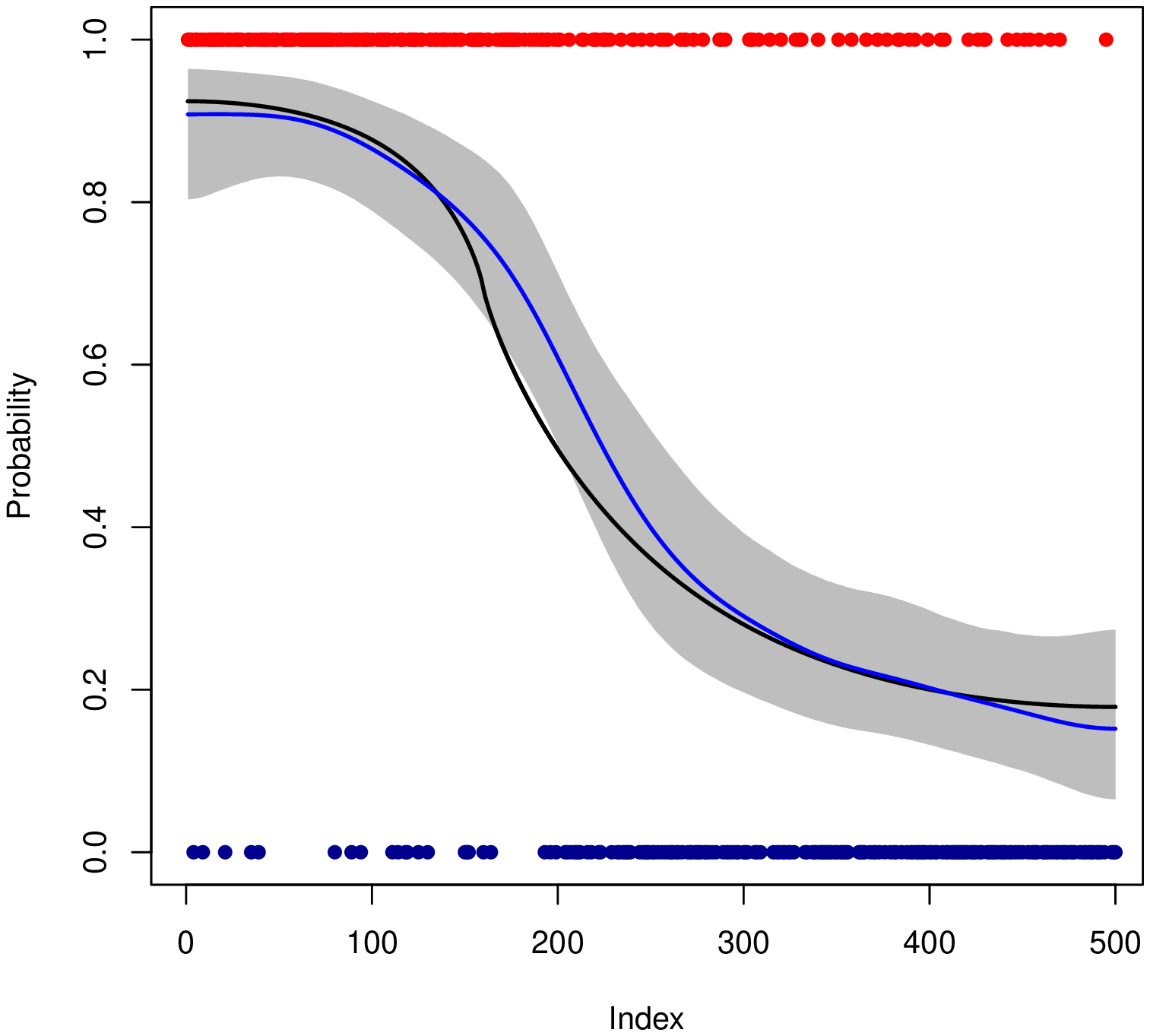}
\includegraphics[width=\textwidth]{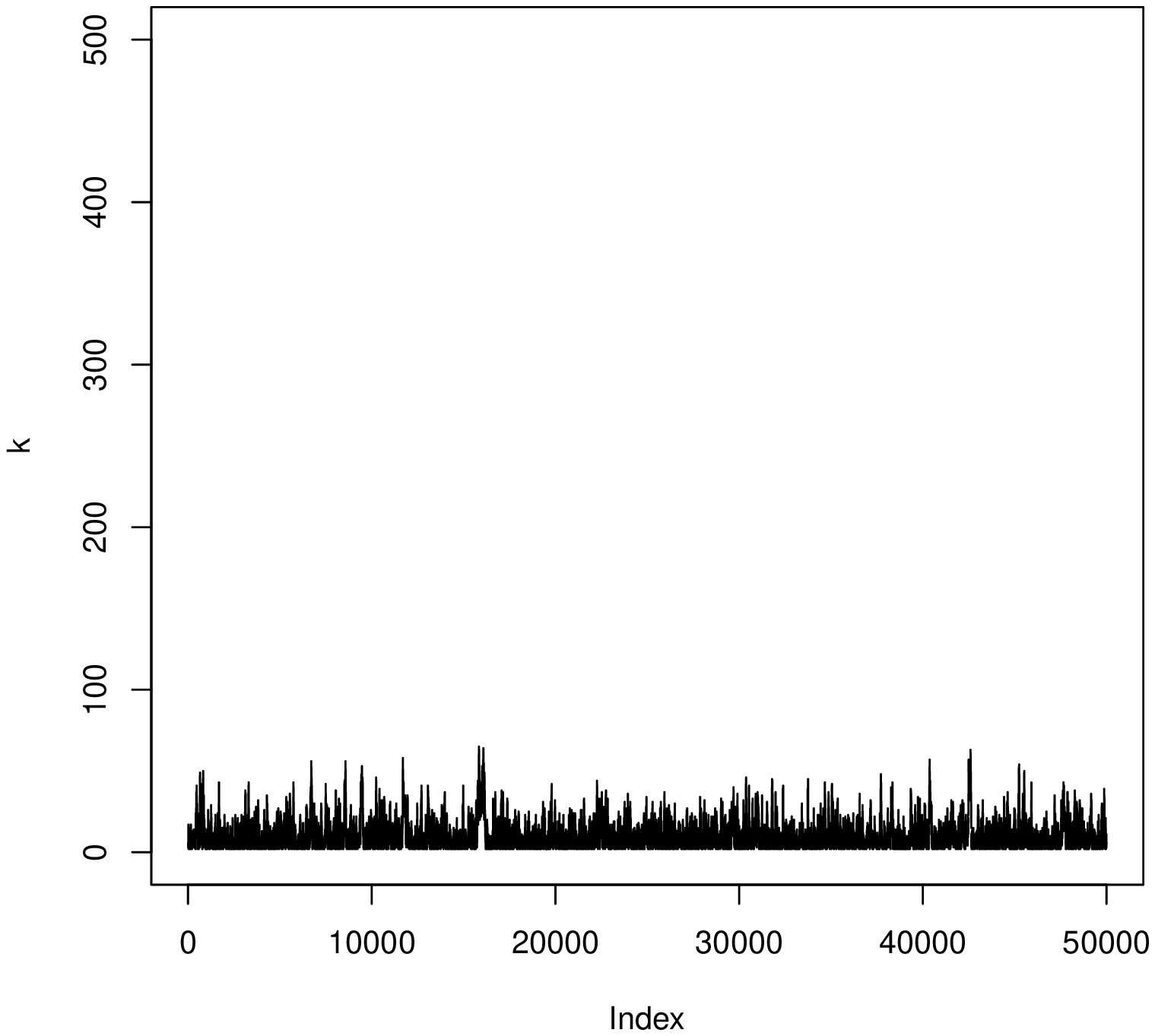}
\end{minipage}
\begin{minipage}{0.33\textwidth}
\includegraphics[width=\textwidth]{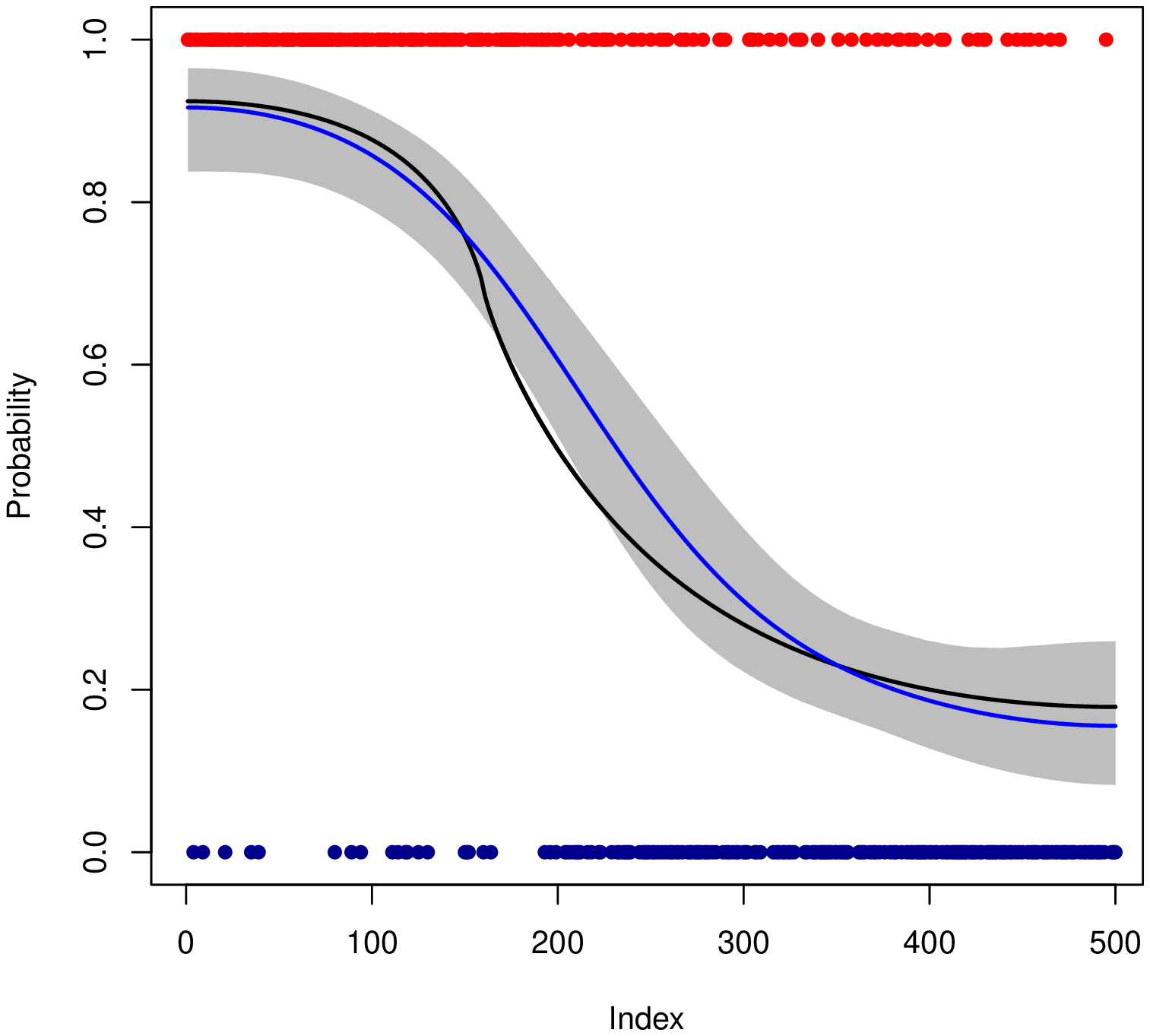}
\includegraphics[width=\textwidth]{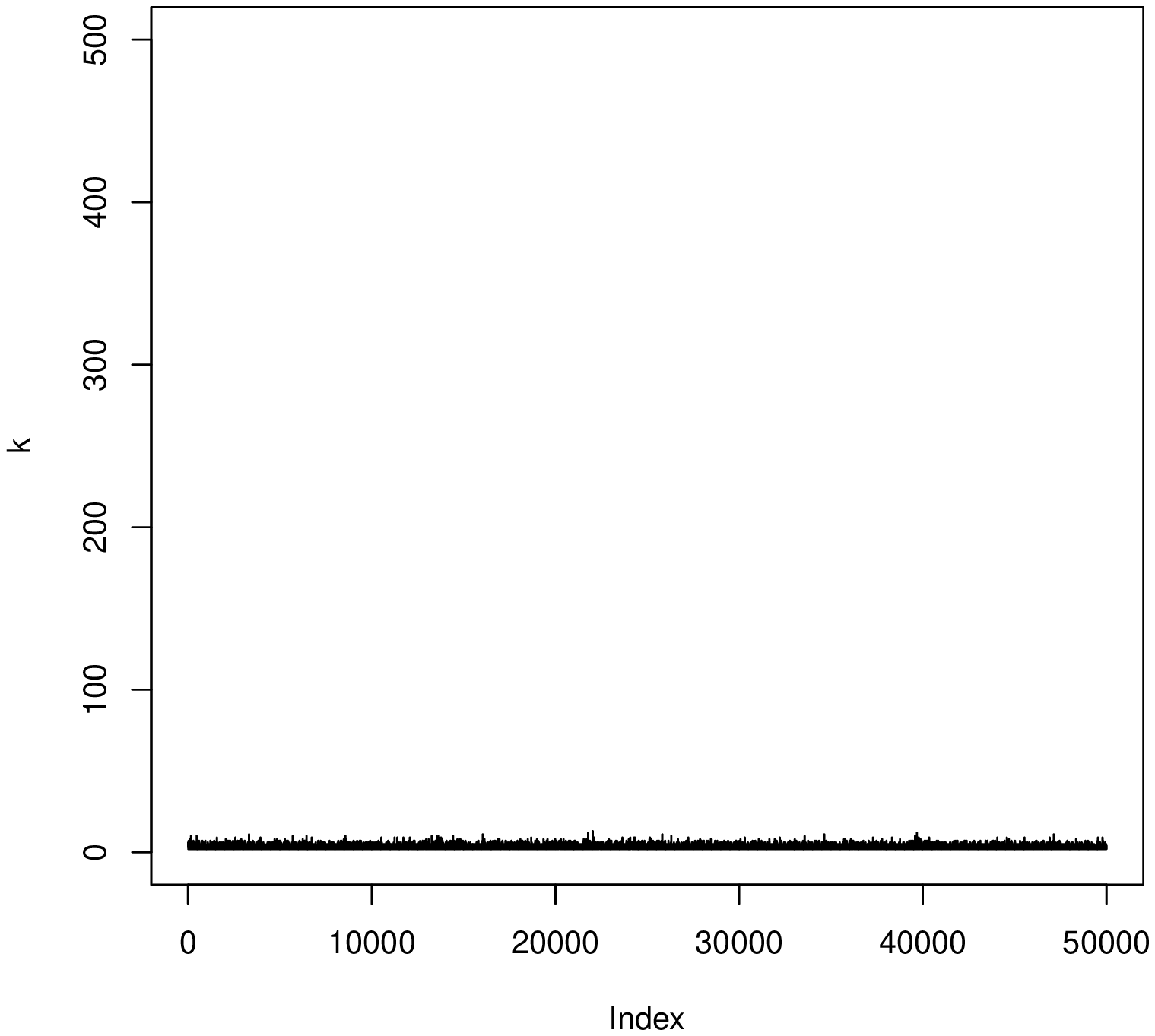}
\end{minipage}
\caption{Top: Posteriors for the soft label function for $\gamma = 0, 0.1, 1$. Bottom: The corresponding draws from the posterior of $k$.}\label{fig:gamma1}
\end{figure}

We also consider a simulated example on a small-world graph obtained as a realization of the Watts-Strogatz model (\cite{watts1998}). The graph is obtained by first considering a ring graph of $1000$ nodes. Then we loop through the nodes and uniformly rewire each edge with probability $0.25$. We keep the largest connected component and delete multiple edges and loops, resulting in a graph with $848$ nodes as shown in Figure 
\ref{fig:ws}. We use the same construction of the observed data on the graph as in the previous example on the path graph. Our suggested rule-of-thumb of setting $\gamma=20/n$ corresponds in the previous examples to $\gamma=0.04$ and $\gamma=0.024$ which in both cases ends up in using only a small fraction of the total number of eigenvectors, but it does not oversmooth too much.

\begin{figure}[H]
\centering
\begin{minipage}{0.5\textwidth}
\includegraphics[width=\textwidth]{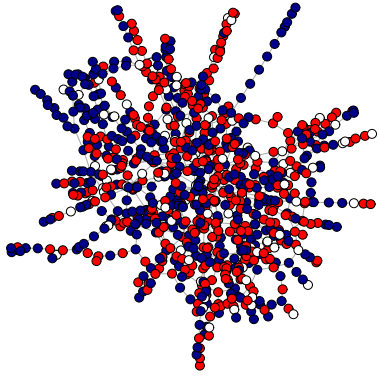}
\end{minipage}
\caption{Small world graph with two types of labels. White nodes represent unobserved labels.}
\label{fig:ws}
\end{figure}

\begin{figure}[H]
\begin{minipage}{0.33\textwidth}
\includegraphics[width=\textwidth]{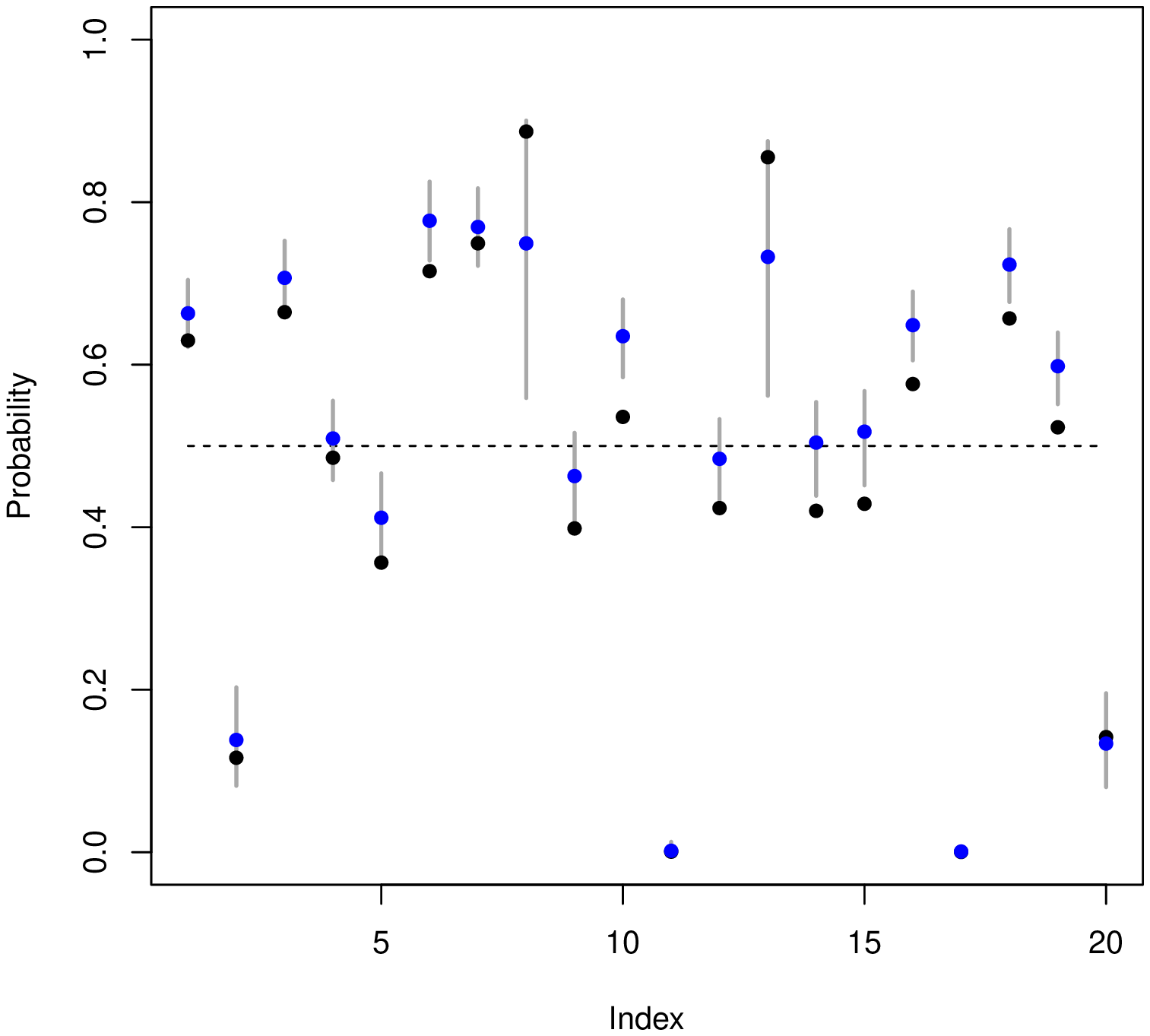}
\includegraphics[width=\textwidth]{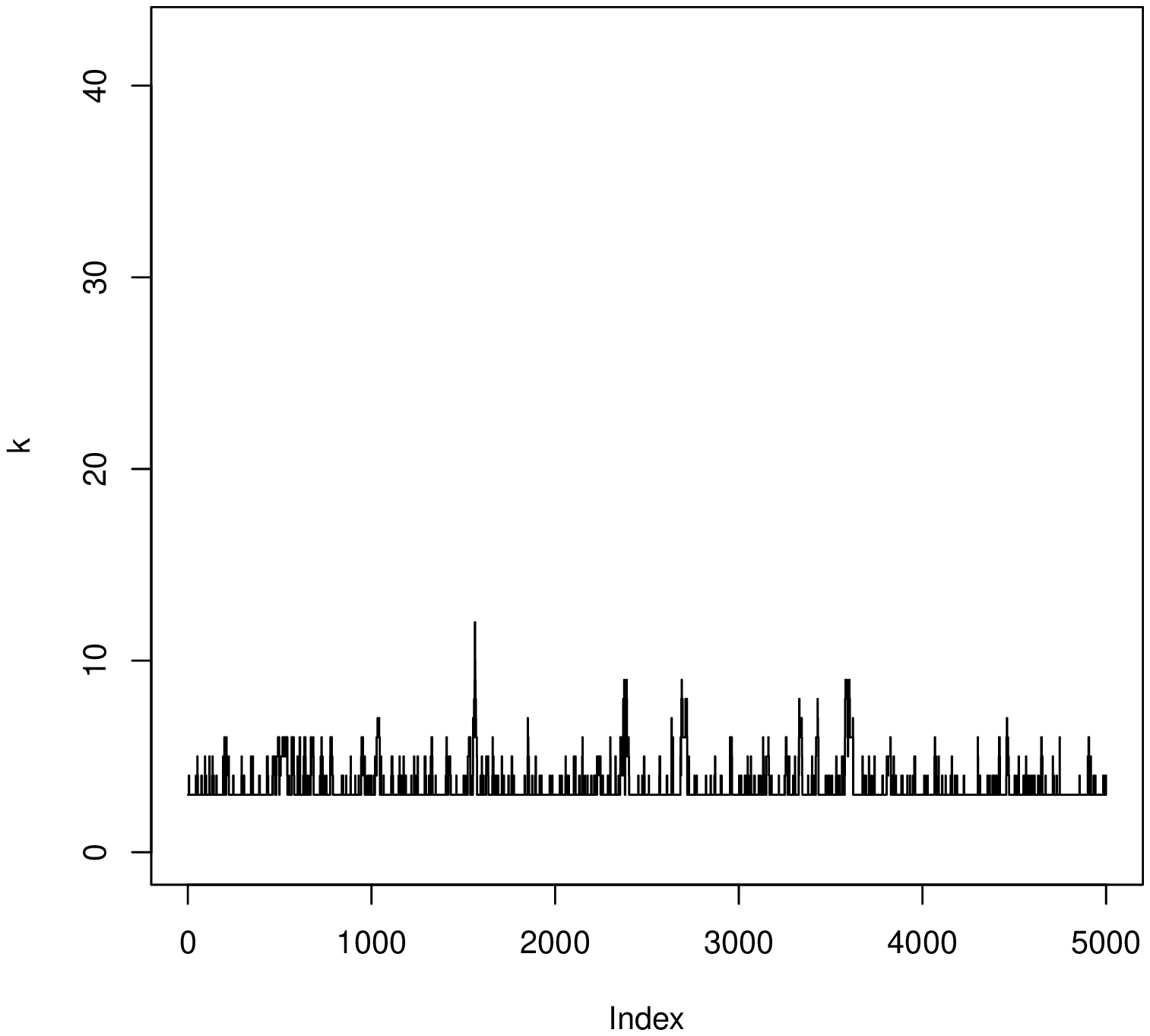}
\end{minipage}
\begin{minipage}{0.33\textwidth}
\includegraphics[width=\textwidth]{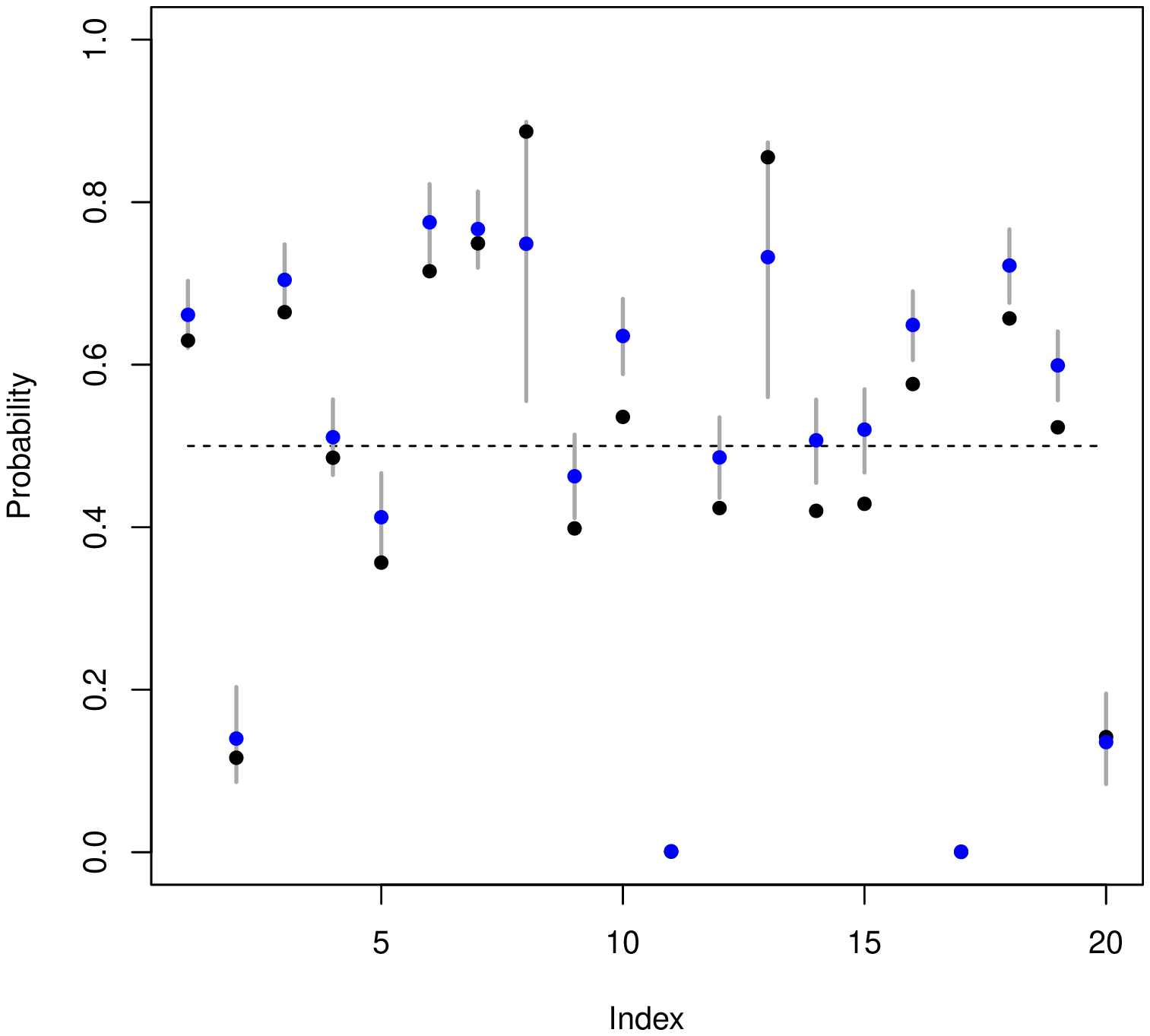}
\includegraphics[width=\textwidth]{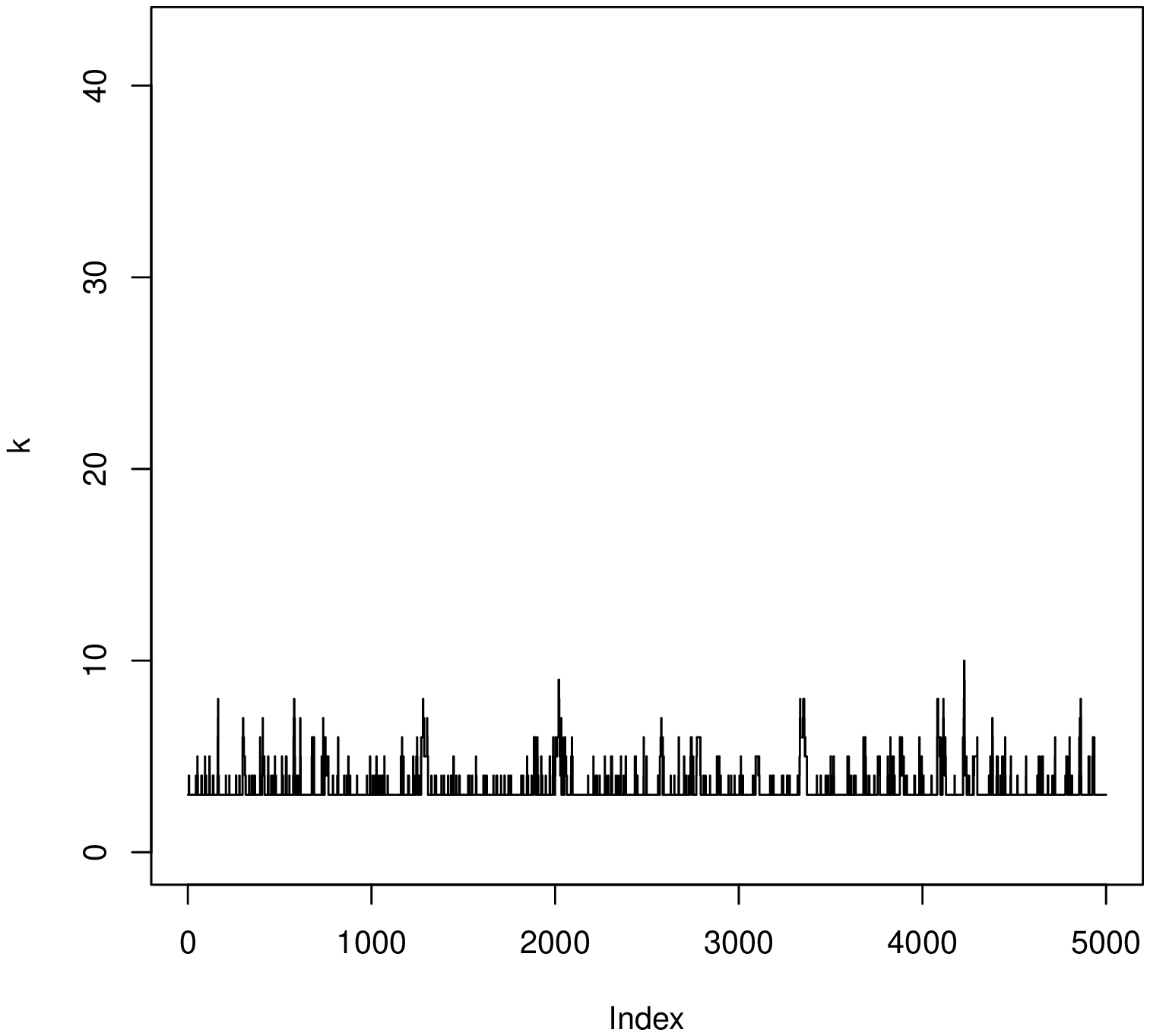}
\end{minipage}
\begin{minipage}{0.33\textwidth}
\includegraphics[width=\textwidth]{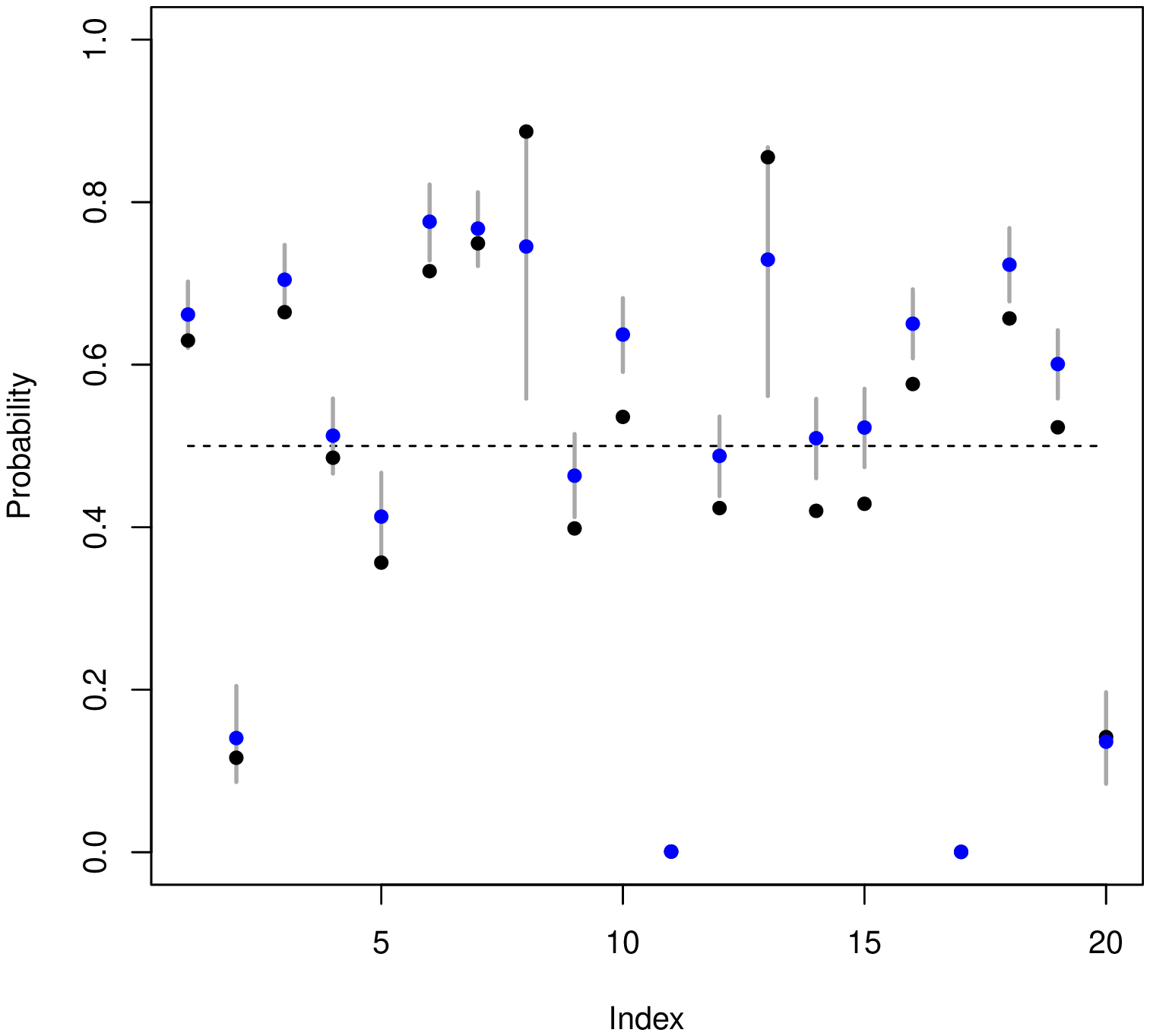}
\includegraphics[width=\textwidth]{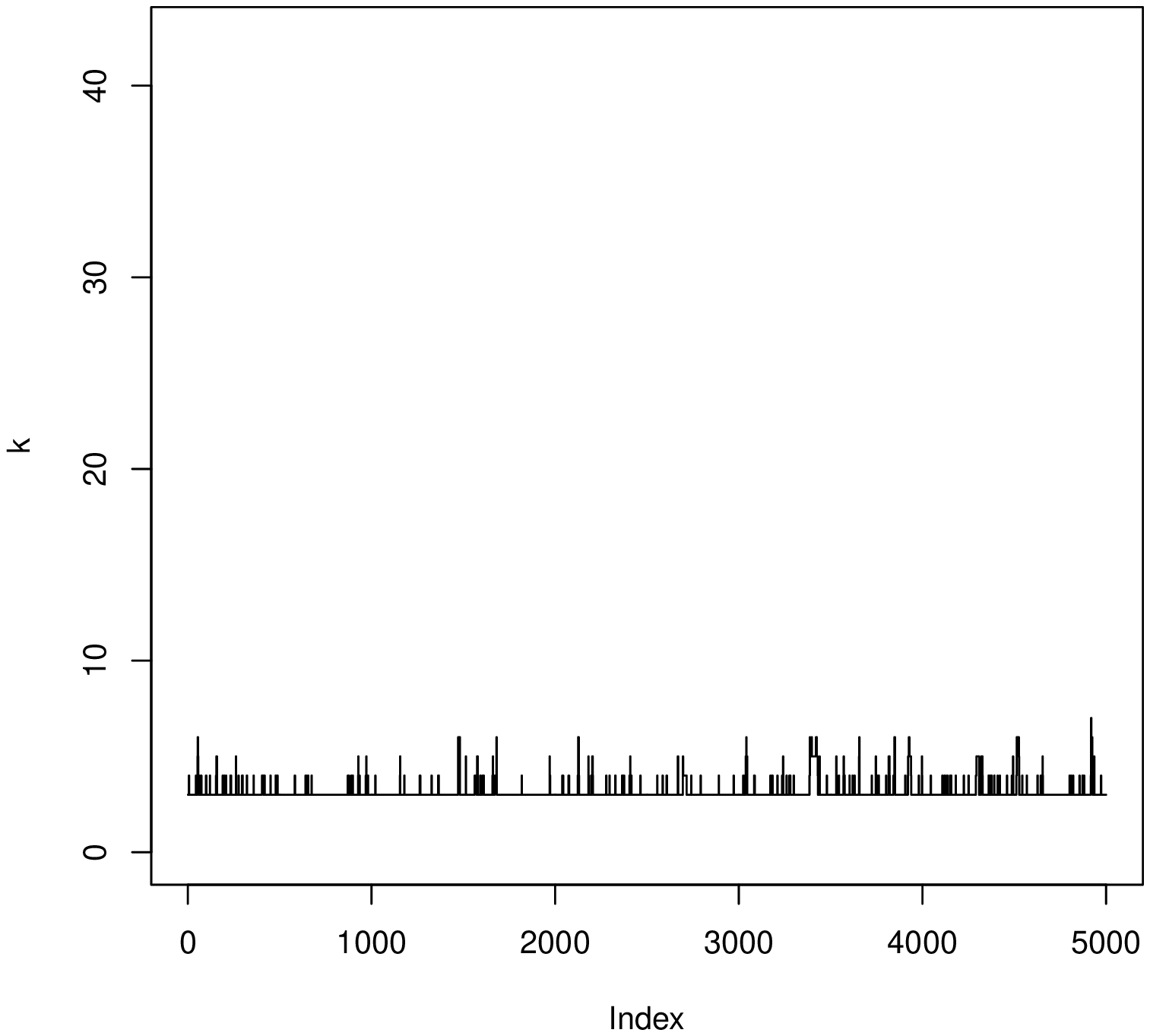}
\end{minipage}
\caption{Top: Posteriors for the soft label function for $\gamma = 0, 0.1, 1$. Bottom: The corresponding draws from the posterior of $k$.}\label{fig:gamma}
\end{figure}

\subsection{Computational gains: MNIST data}

The MNIST dataset consists of images of handwritten digits. The images are size-normalized and centered. The dataset is publicly available at {\tt http://yann.lecun.com/exbd/mnist}. We have selected the images of only the digits $4$ and $9$ from both the test set ($1991$ images) and the training set ($11791$ images). Our goal is to classify the images from the test set using the images from the training set. To turn this into a label prediction problem on a graph we construct a graph with $11791+1991=13782$ nodes representing the images. For each image we determine the $15$ closest images in Euclidean distance between the projections on the first $50$ principal components, similar to \cite{bertozzi2017}, \cite{liang2007} and \cite{belkin2006}. 

We use this example to explore the relative speedup of our proposed method with respect to the method proposed in \cite{jarno} without truncation, where we also compare the difference in prediction accuracy. To this end, we take random subsamples of different sizes of the graph to illustrate what happens when the size of the graph grows. The ratio of $4$'s and $9$'s in the test and train set is kept constant and equal to that of the entire dataset. In Figure \ref{fig:time} we observe a 
dramatic difference in the computational time for the algorithm without truncation versus the algorithm with truncation, while the prediction performance is comparable.

\begin{figure}[H]
\begin{minipage}{0.33\textwidth}
\includegraphics[width=\textwidth]{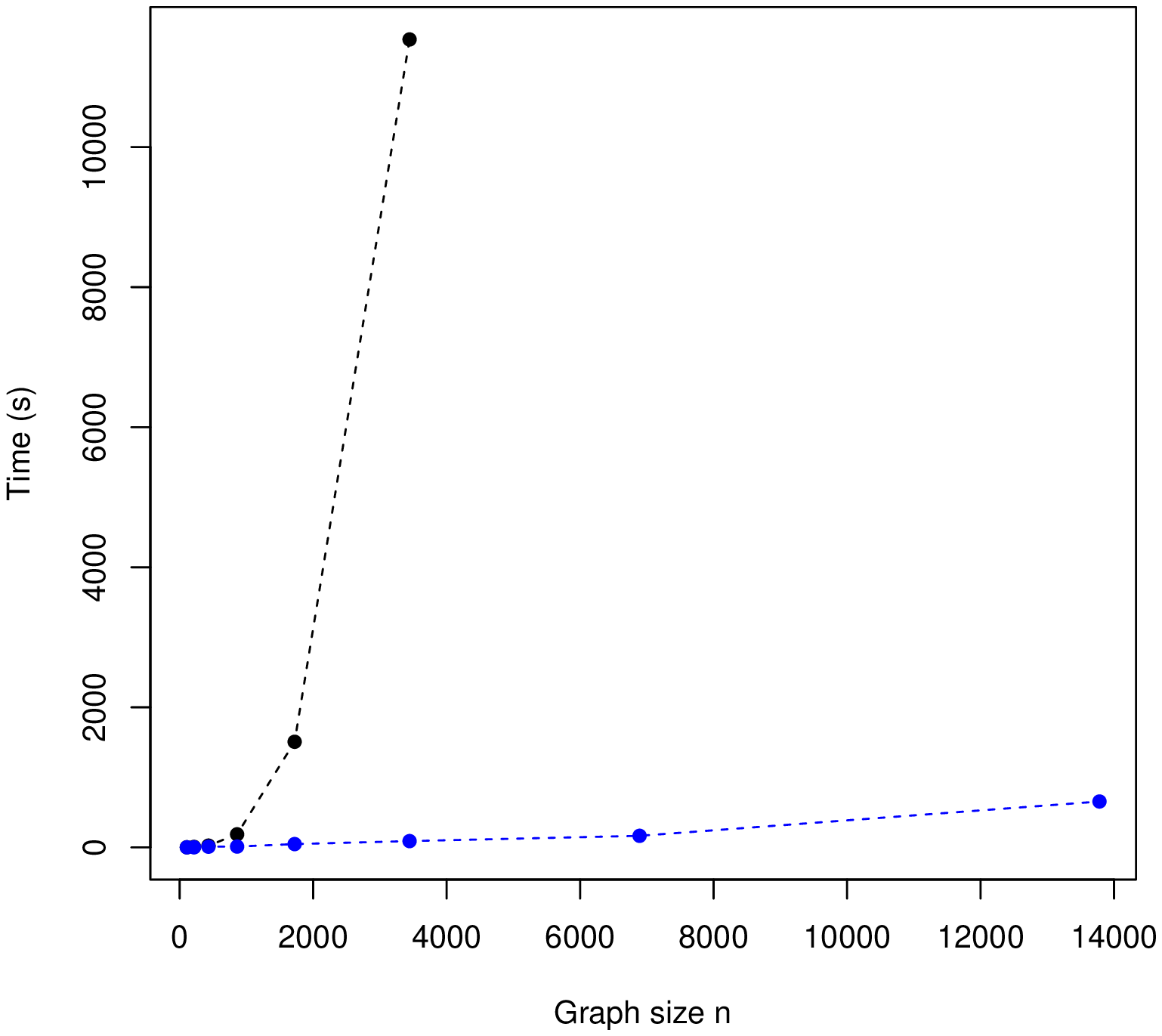}
\end{minipage}
\begin{minipage}{0.33\textwidth}
\includegraphics[width=\textwidth]{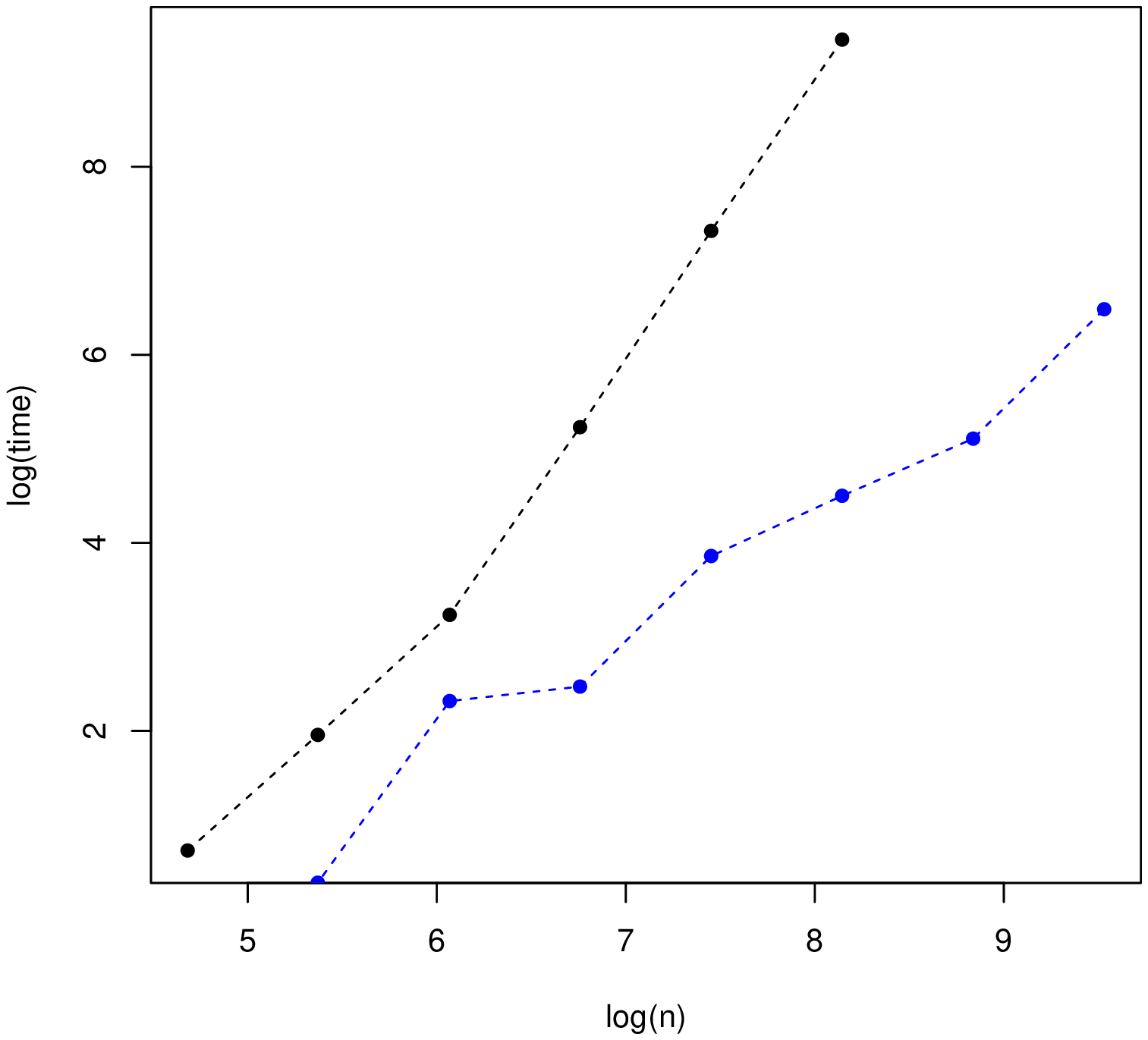}
\end{minipage}
\begin{minipage}{0.33\textwidth}
\includegraphics[width=\textwidth]{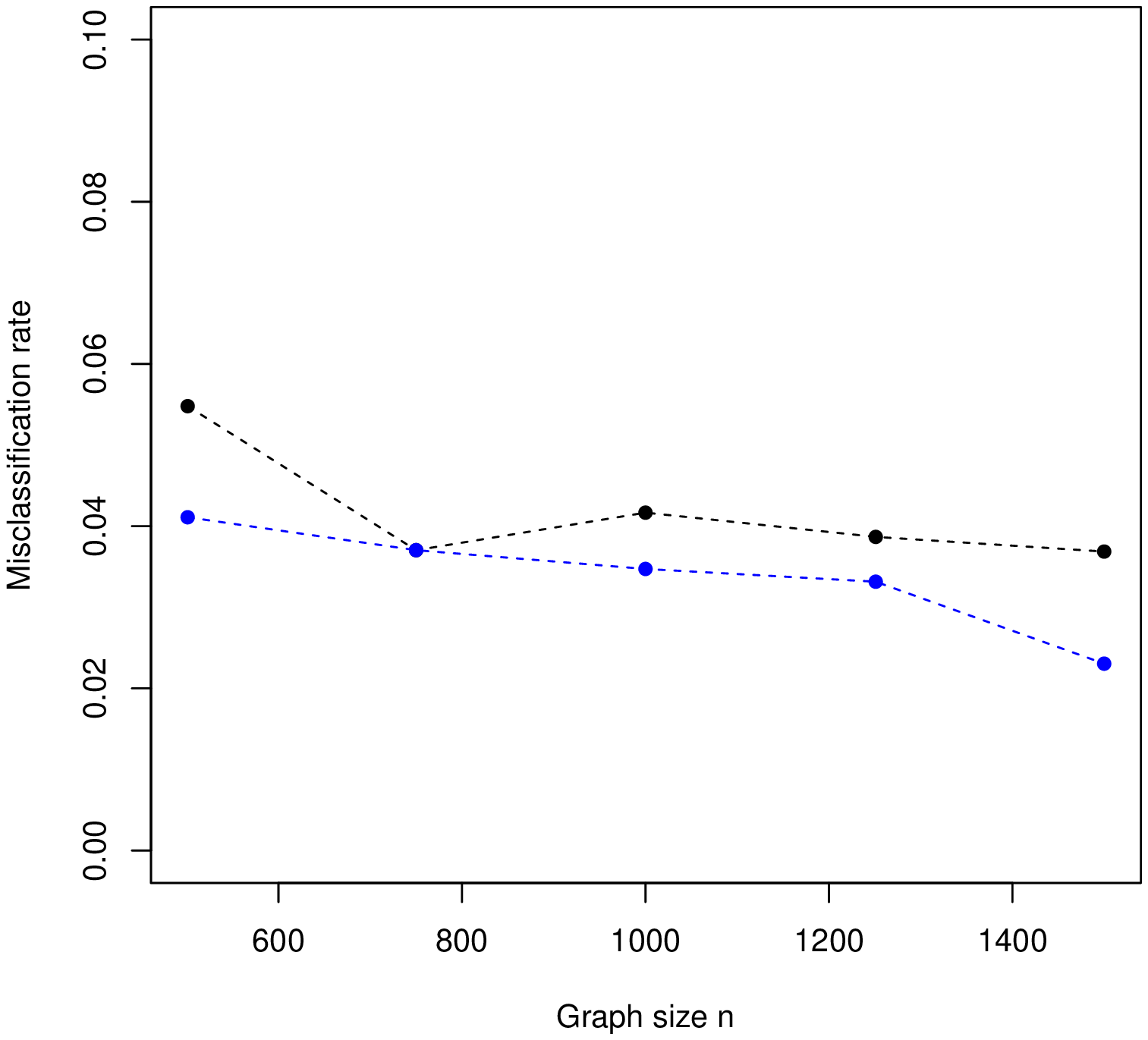}
\end{minipage}
\caption{Left: Computation time versus graph size. Middle: Computation time versus graph size on a log-log scale. Right: The prediction performance on the different subgraphs. The blue line is our proposed truncation method, the black line is the method from \cite{jarno}. In both methods we set $a=b=0$ and we choose $\gamma=20/n$.}\label{fig:time}
\end{figure}

\subsection{Large scale example: object tracking}

To demonstrate the applicability of our proposed method in a large graph, where for example the untruncated method from \cite{jarno} is prohibitive, we use a simulated object tracking application. As ground truth, we use the animation given by the following frames. 

\begin{figure}[H]
\begin{minipage}{0.33\textwidth}
\includegraphics[width=\textwidth]{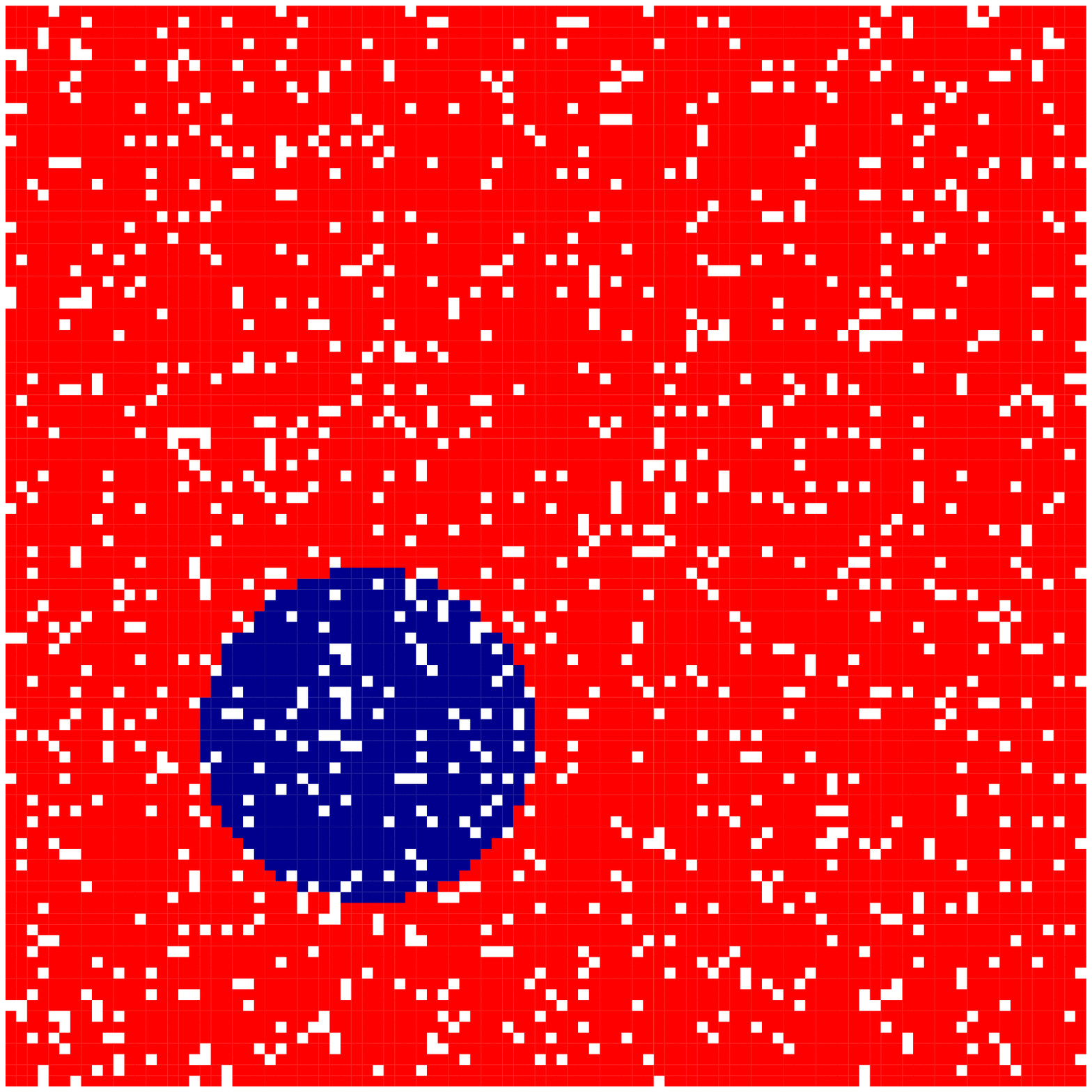}
\includegraphics[width=\textwidth]{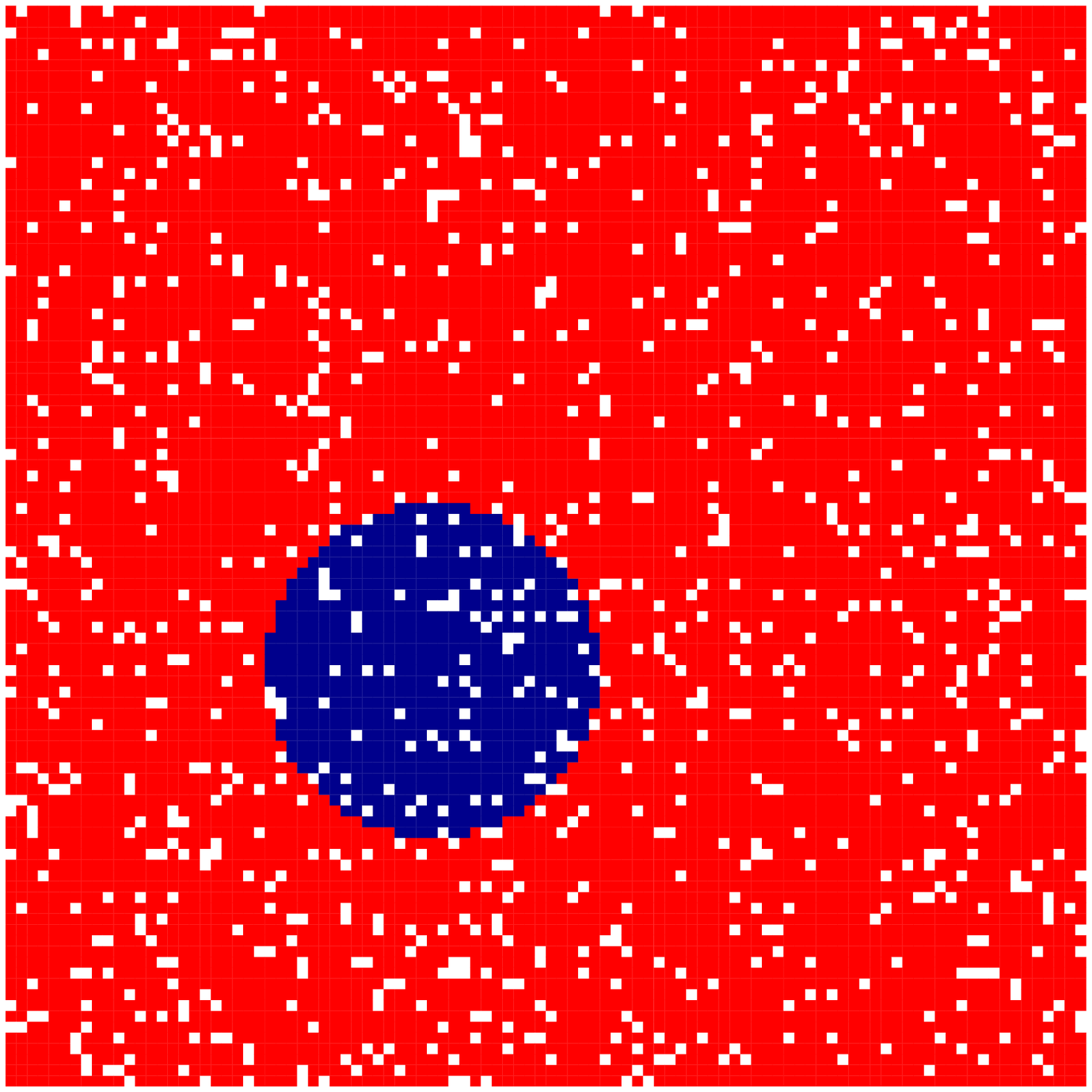}
\includegraphics[width=\textwidth]{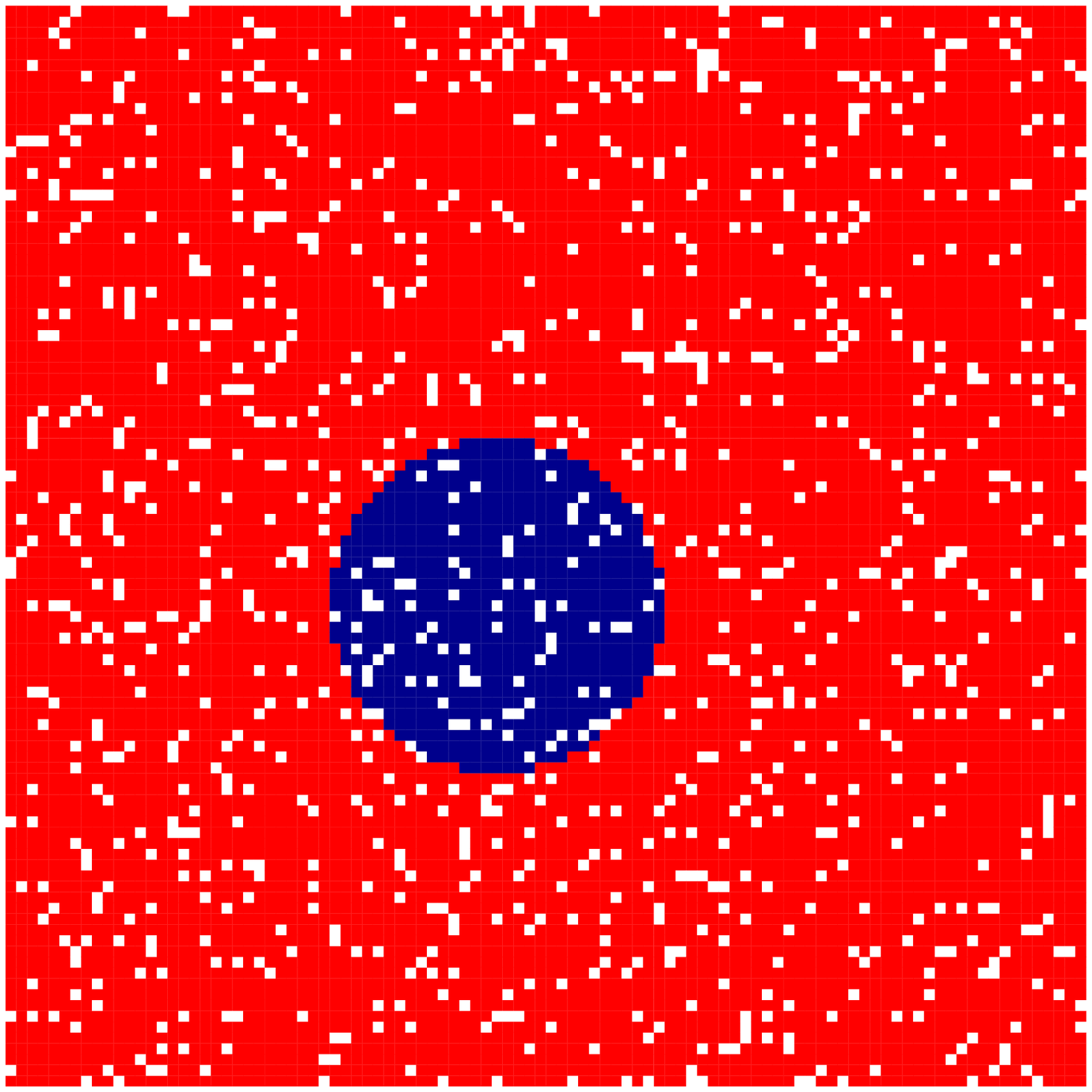}
\end{minipage}
\begin{minipage}{0.33\textwidth}
\includegraphics[width=\textwidth]{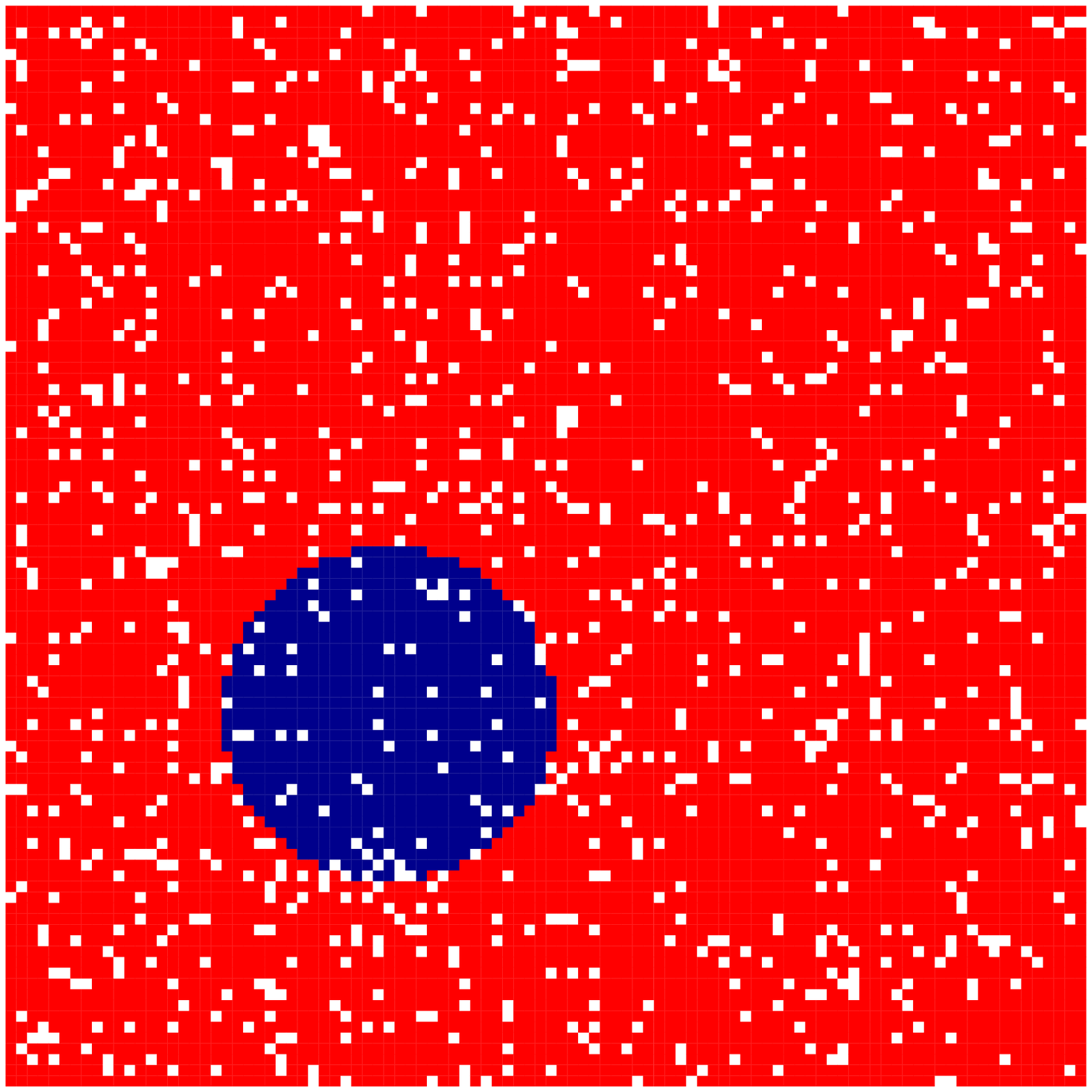}
\includegraphics[width=\textwidth]{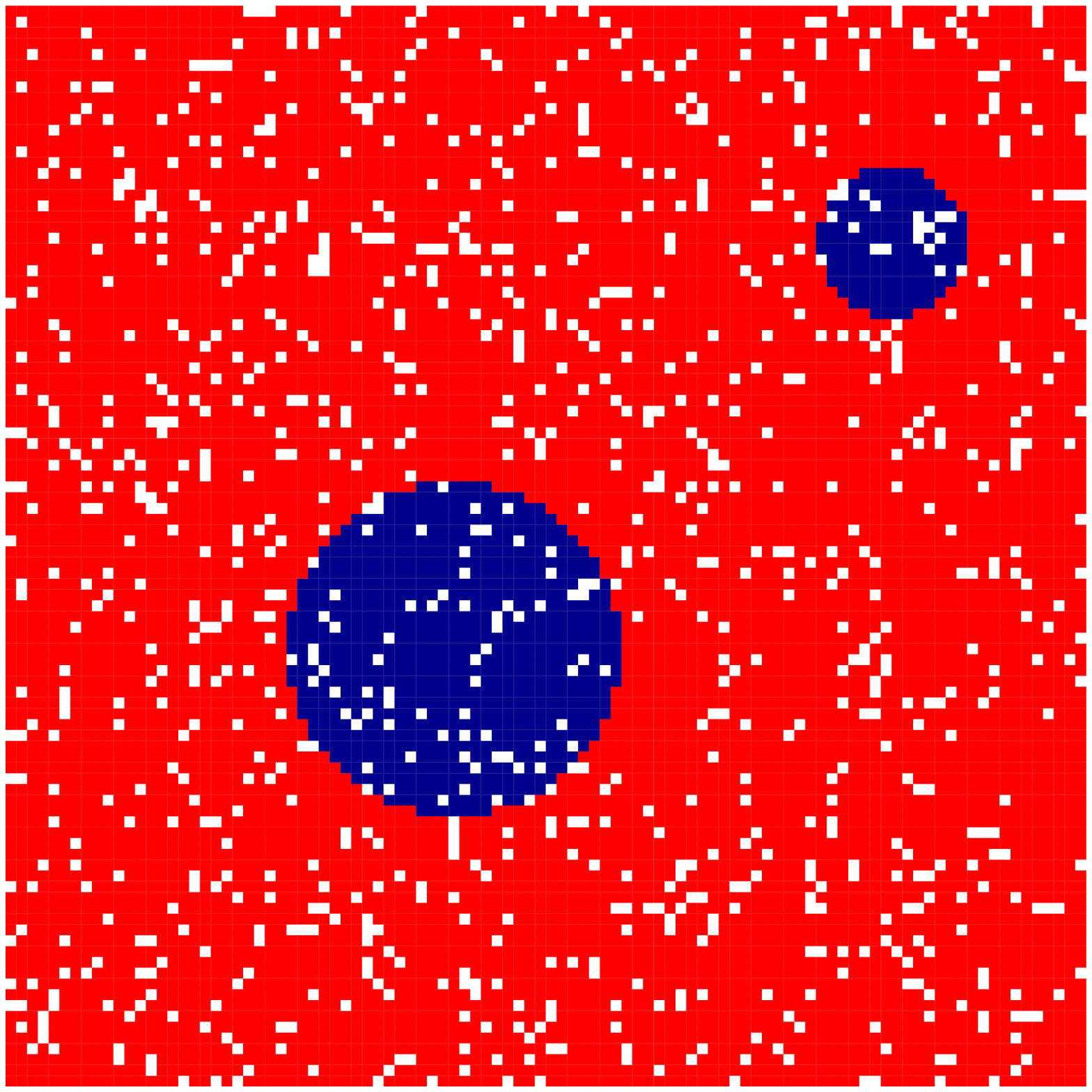}
\includegraphics[width=\textwidth]{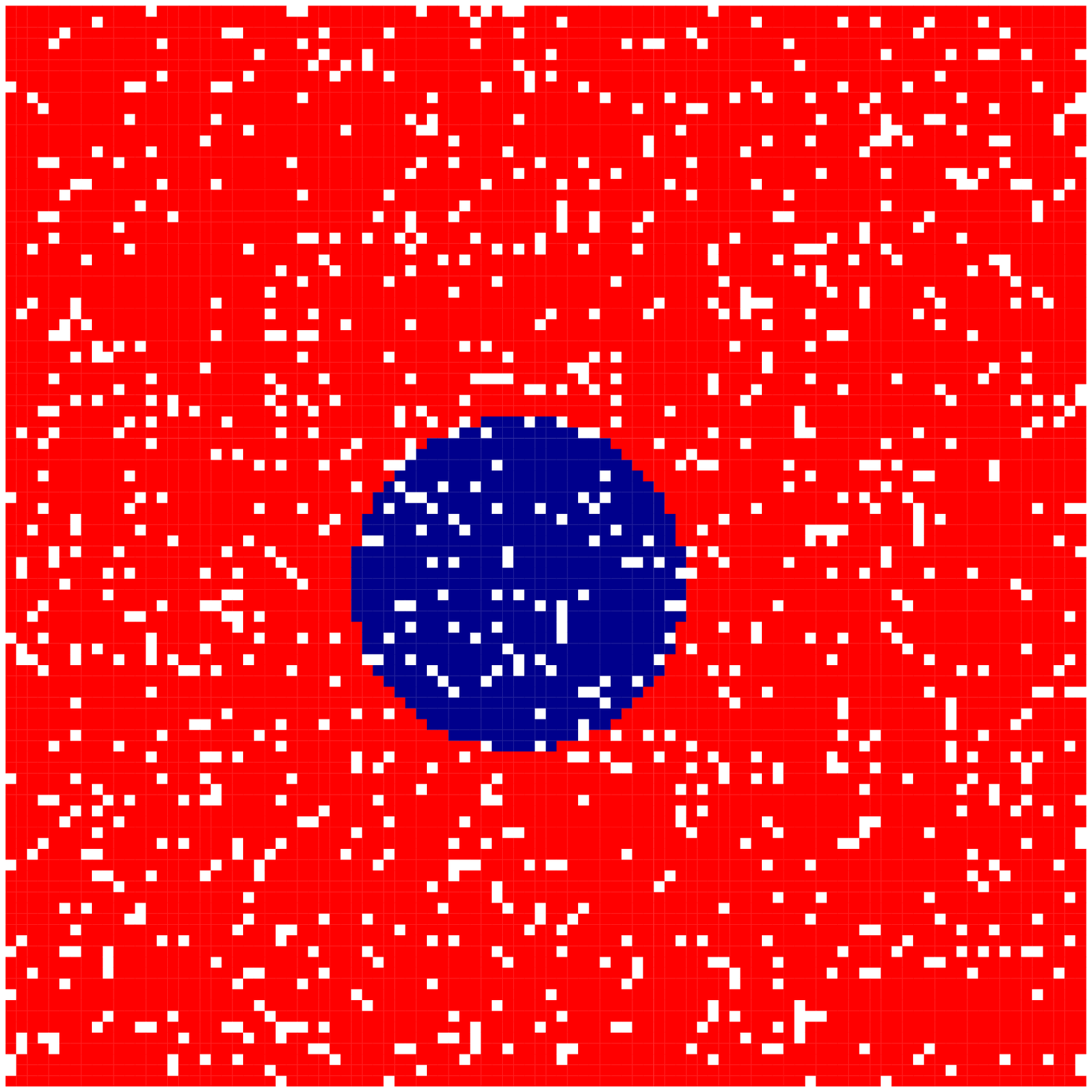}
\end{minipage}
\begin{minipage}{0.33\textwidth}
\includegraphics[width=\textwidth]{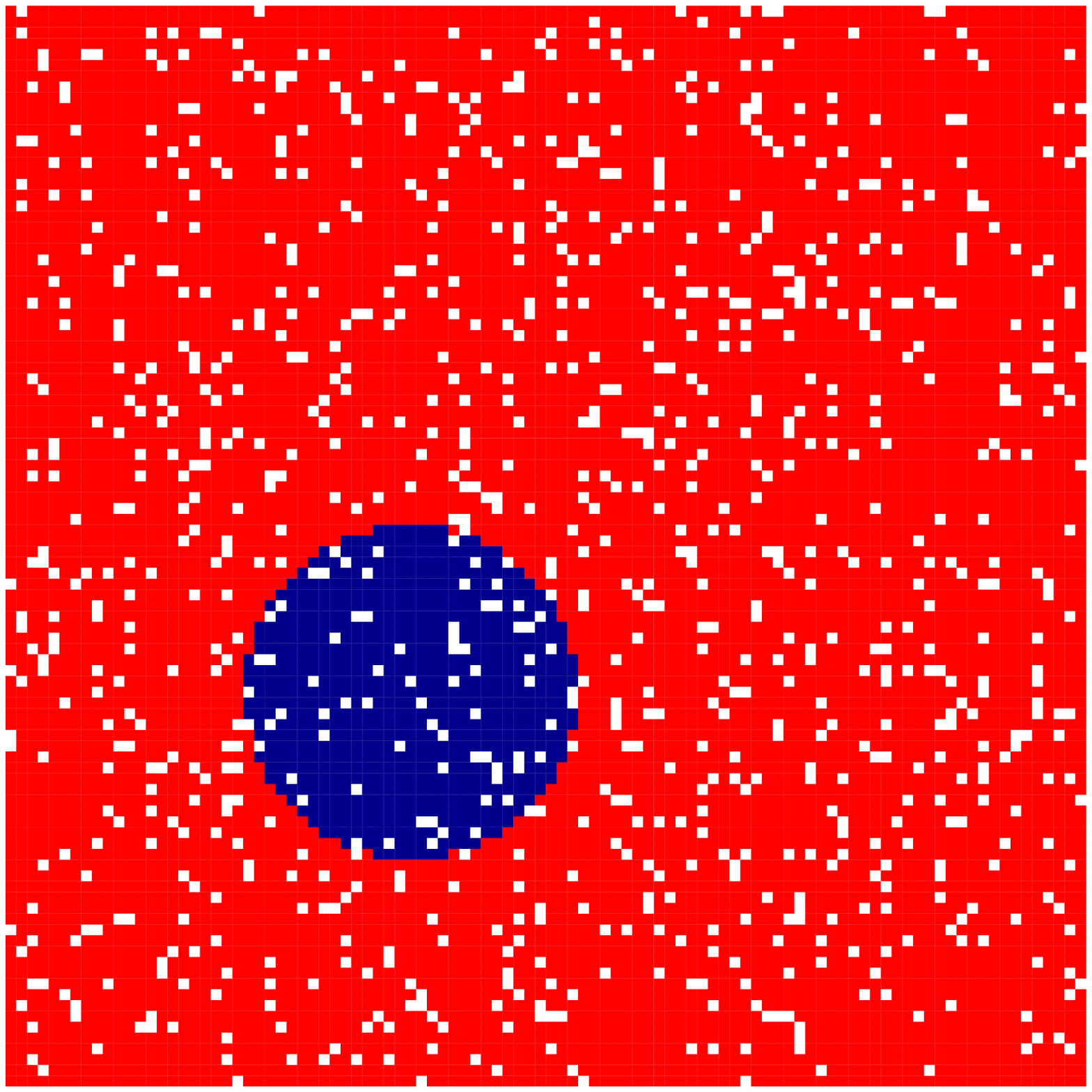}
\includegraphics[width=\textwidth]{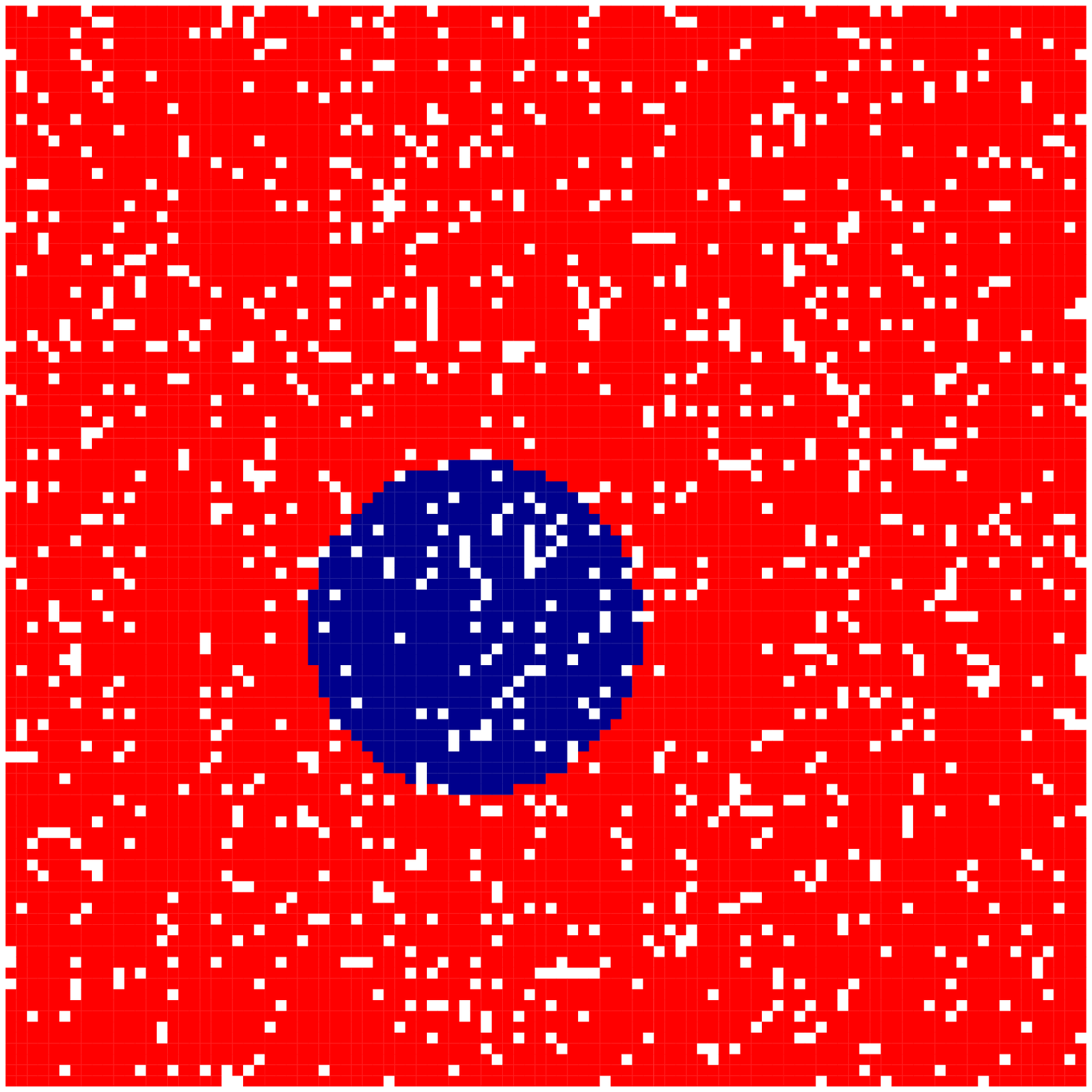}
\includegraphics[width=\textwidth]{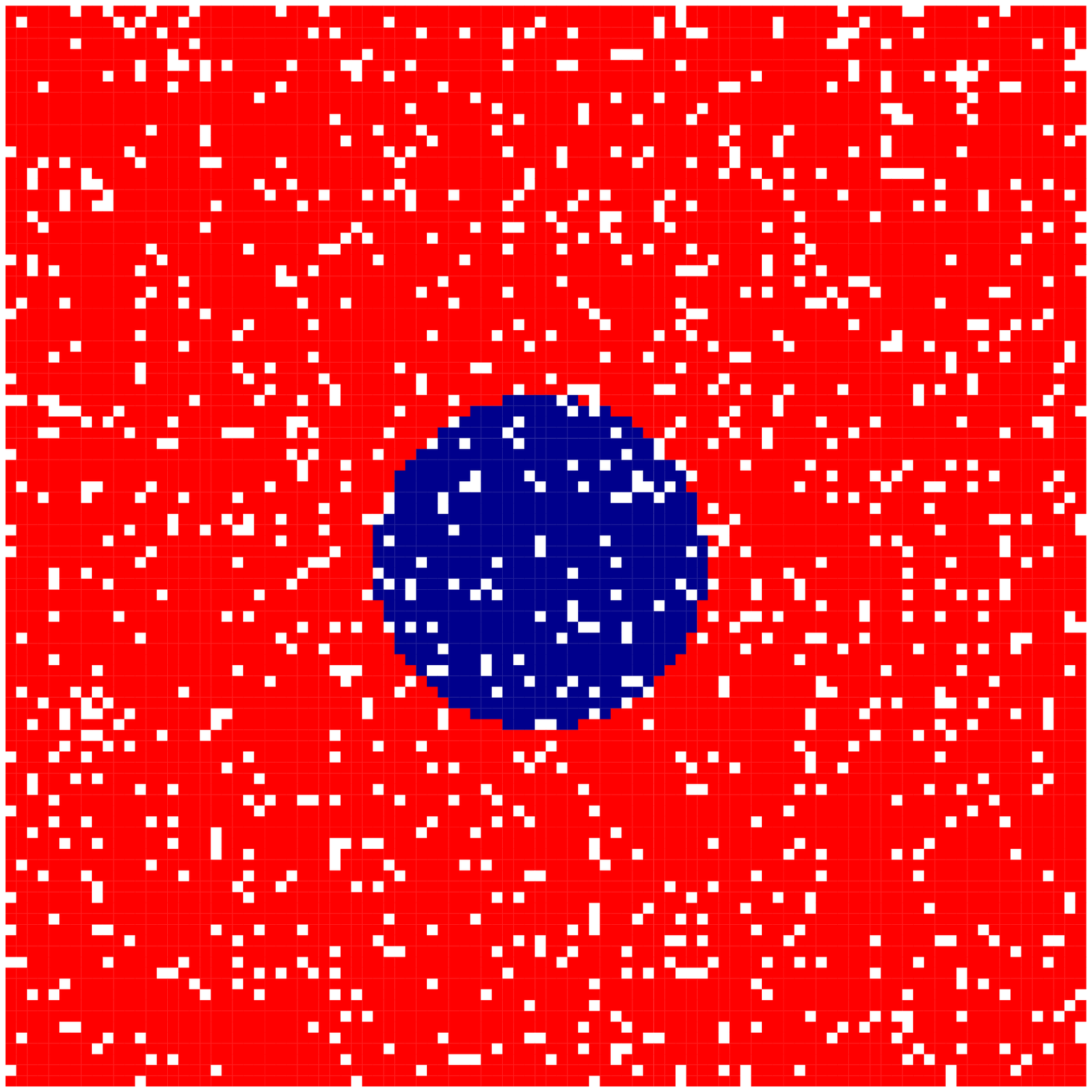}
\end{minipage}
\caption{Moving object in a noisy environment, images are $100\times100$ pixels. The fifth image is corrupted with an additional object which should be removed.}
\end{figure}

The animation consists of $9$ frames of $100\times100$ pixels. It represents a slowly moving blue ball on a red background. We removed the color of $10\%$ of the pixels at random and added an additional ball in the fifth image to represent a corrupted frame. To convert the animation into a graph problem we connect neighboring pixels in each frame and with the corresponding pixels in the previous and next frame, resulting in a $100 \times 100 \times 9$ grid graph on a total of $n = 90000$ nodes as in Figure \ref{fig:grid3}. 

\begin{figure}[H]
\centering
\begin{minipage}{0.5\textwidth}
\includegraphics[width=\textwidth]{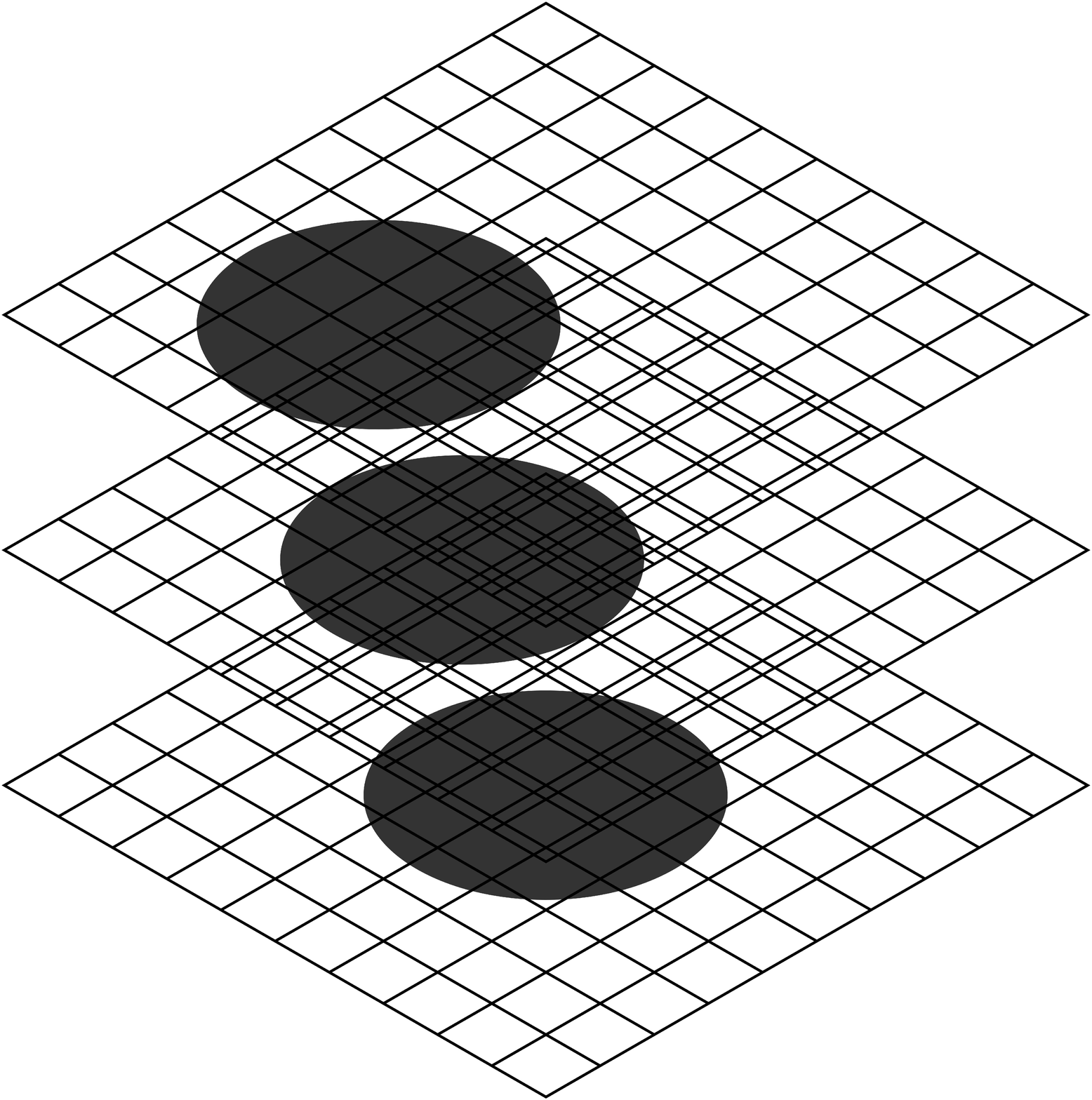}
\end{minipage}
\caption{Schematic representation of the construction of the grid graph. Each pixel is also connected to the pixel at the same location in the frame before and after the current frame (these lines are omitted from the above representation for clarity).}
\label{fig:grid3}
\end{figure}

We can explicitly compute the eigenvalues as $\lambda_i + \mu_j + \nu_k$ (cf. Theorem 3.5 in \cite{mohar1991}), where
\begin{align*}
\lambda_i &=\ 4\sin^2\left(\frac{\pi i}{200}\right), \quad i=0, \ldots, 99, \\
\mu_j &=\ 4\sin^2\left(\frac{\pi j}{200}\right), \quad j=0, \ldots, 99, \\
\nu_k &=\ 4\sin^2\left(\frac{\pi k}{18}\right), \quad k=0, \ldots, 8.
\end{align*}
The corresponding eigenvectors are given by the tensor products $w^{(k)} \otimes v^{(j)} \otimes u^{(i)}$ of eigenvectors of the path graph with sizes $100$ ($u$), $100$ ($v$) and $9$ ($w$) as in Equation (\ref{eq:eigenvector}). Using the noisy images, we estimate the location of the ball as show in Figures \ref{fig:ballp} and \ref{fig:ballest}. We observe that the object is located in all images and that the additional object in frame $5$ adds some noise in the probability estimates, but is ignored when we truncate at probability $0.5$.

\begin{figure}[H]
\begin{minipage}{0.33\textwidth}
\includegraphics[width=\textwidth]{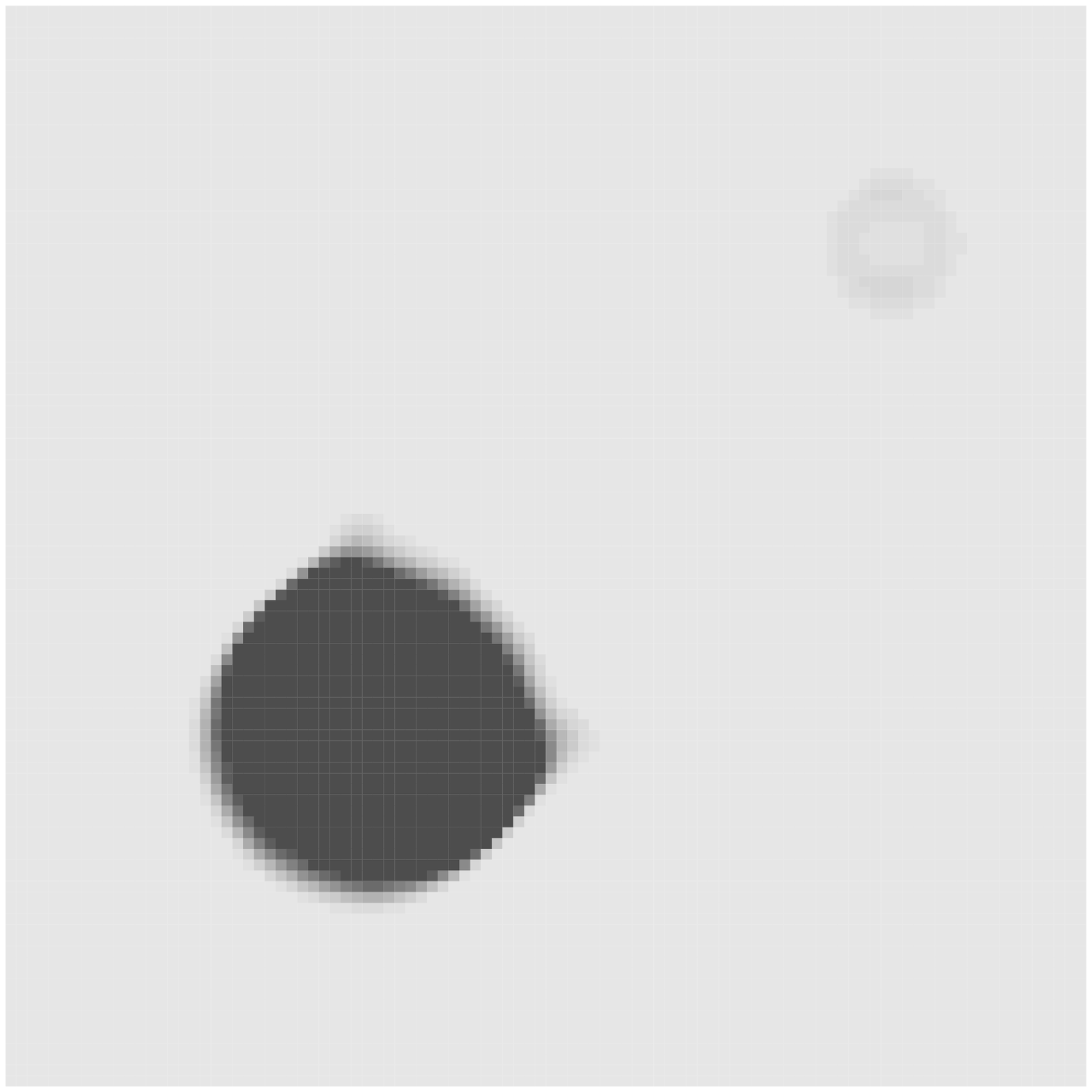}
\includegraphics[width=\textwidth]{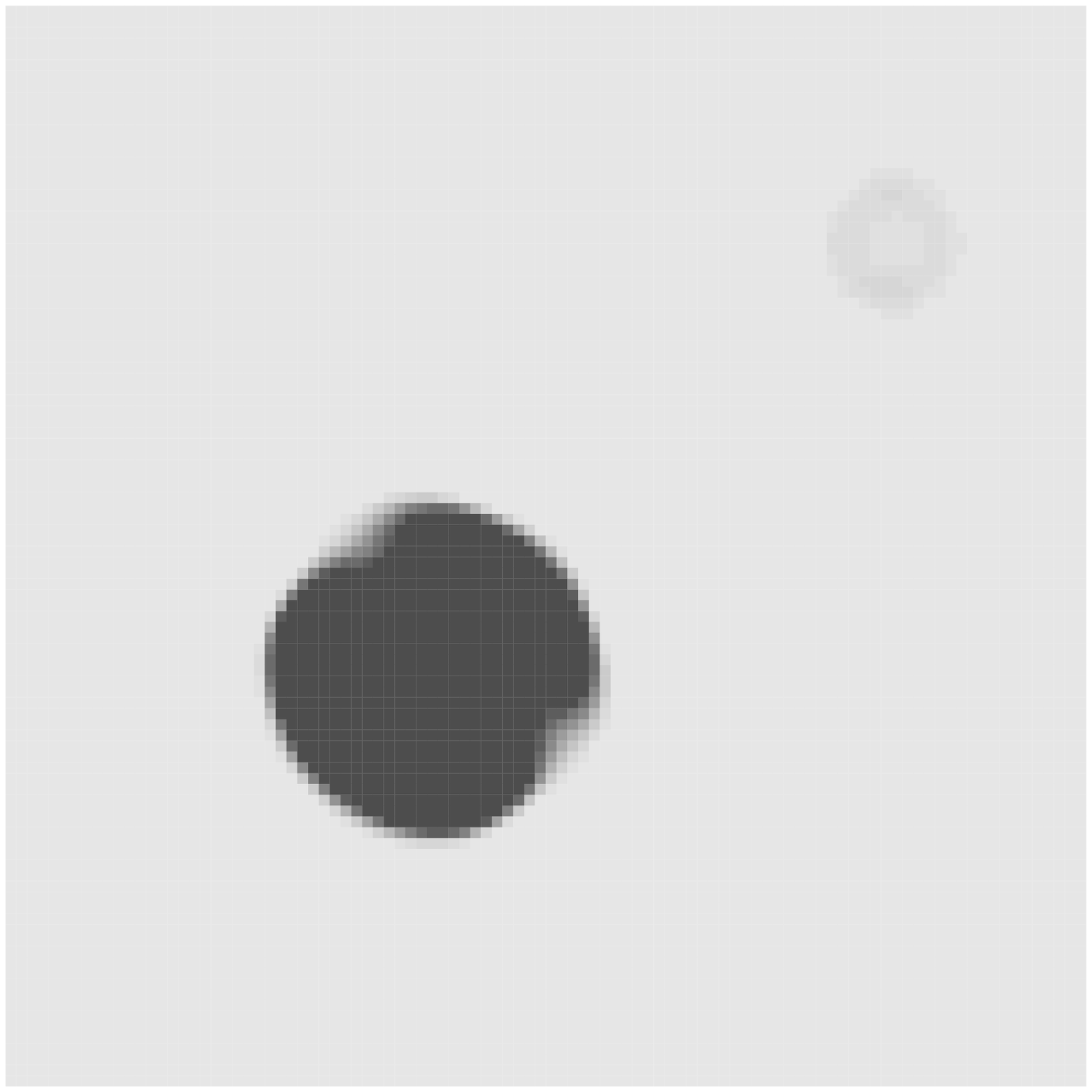}
\includegraphics[width=\textwidth]{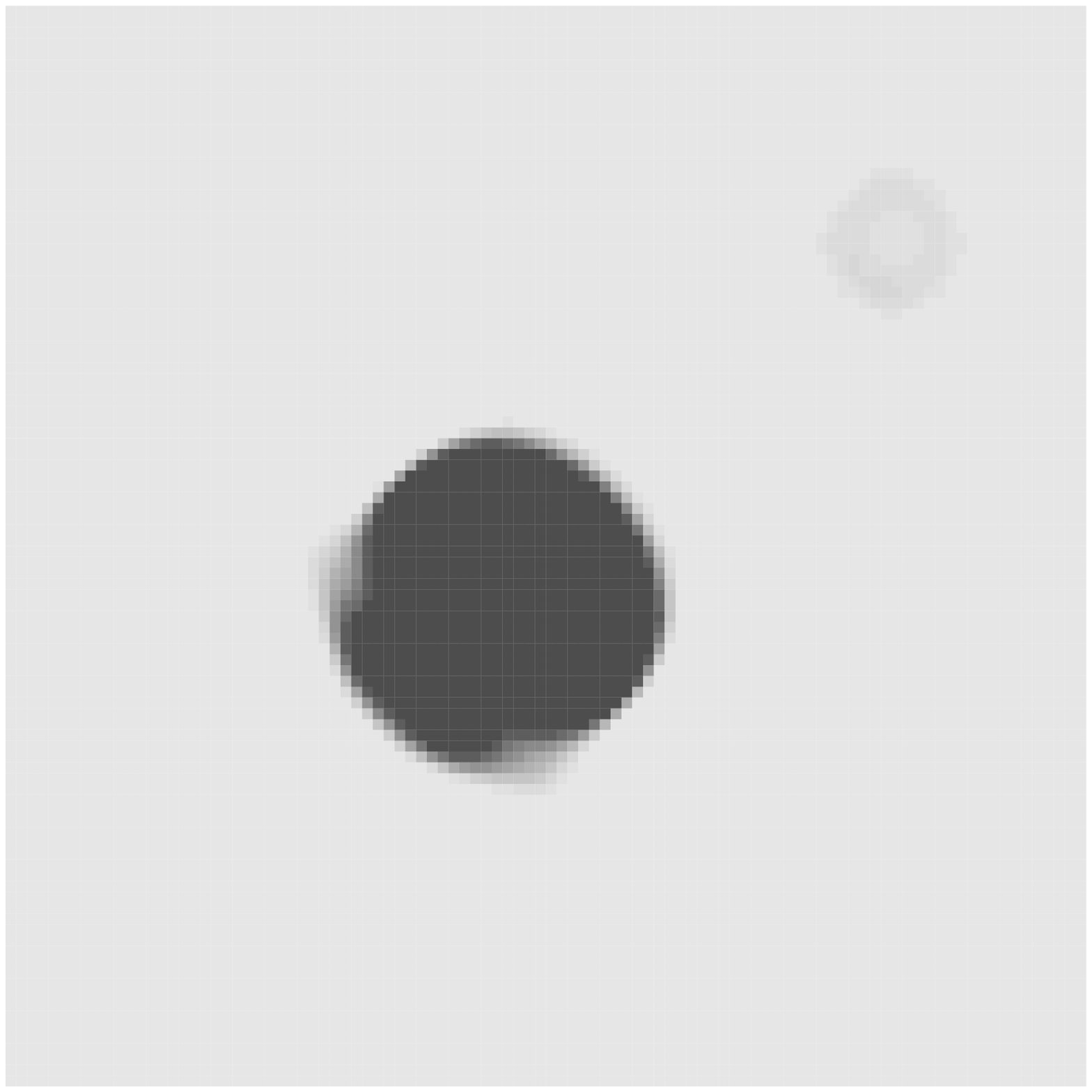}
\end{minipage}
\begin{minipage}{0.33\textwidth}
\includegraphics[width=\textwidth]{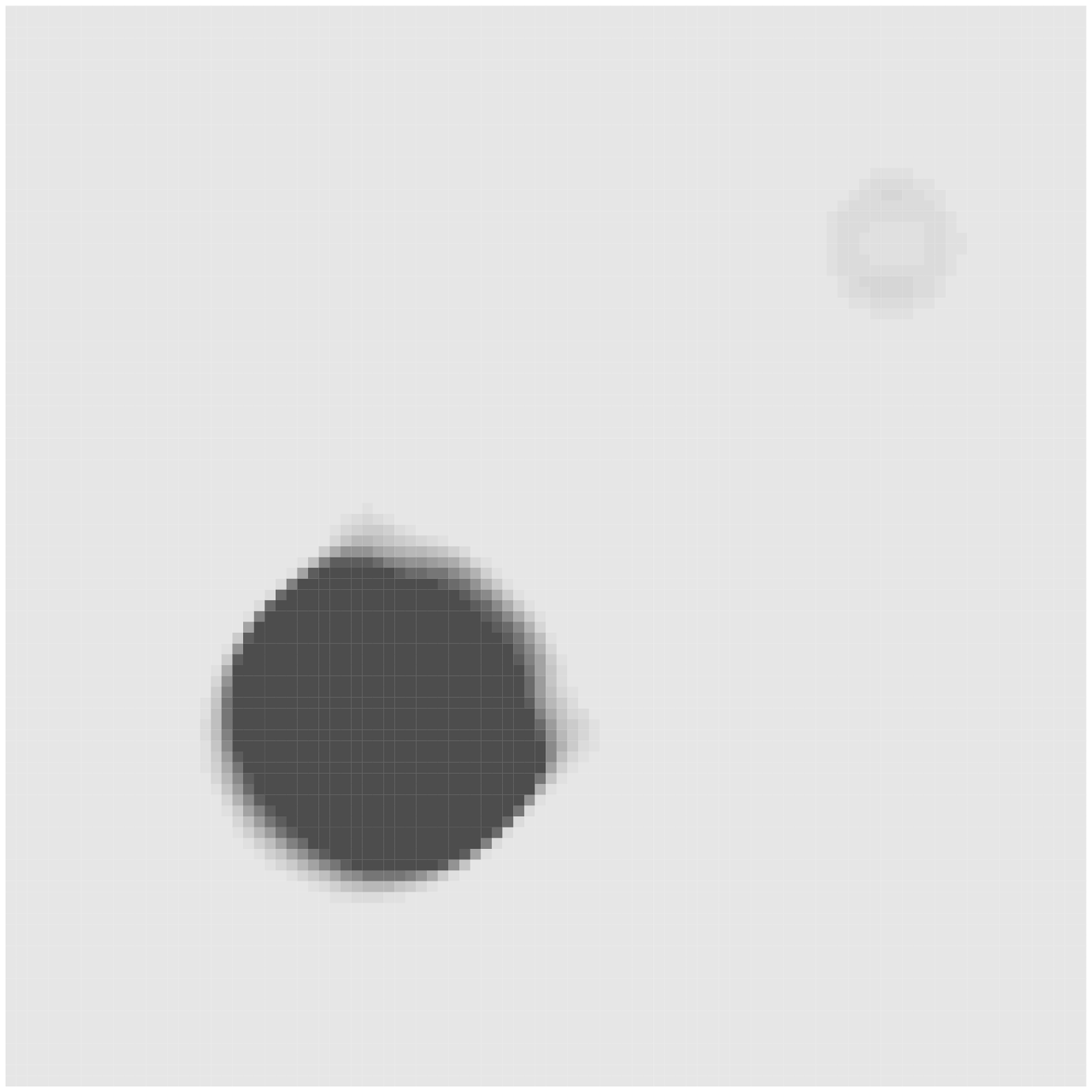}
\includegraphics[width=\textwidth]{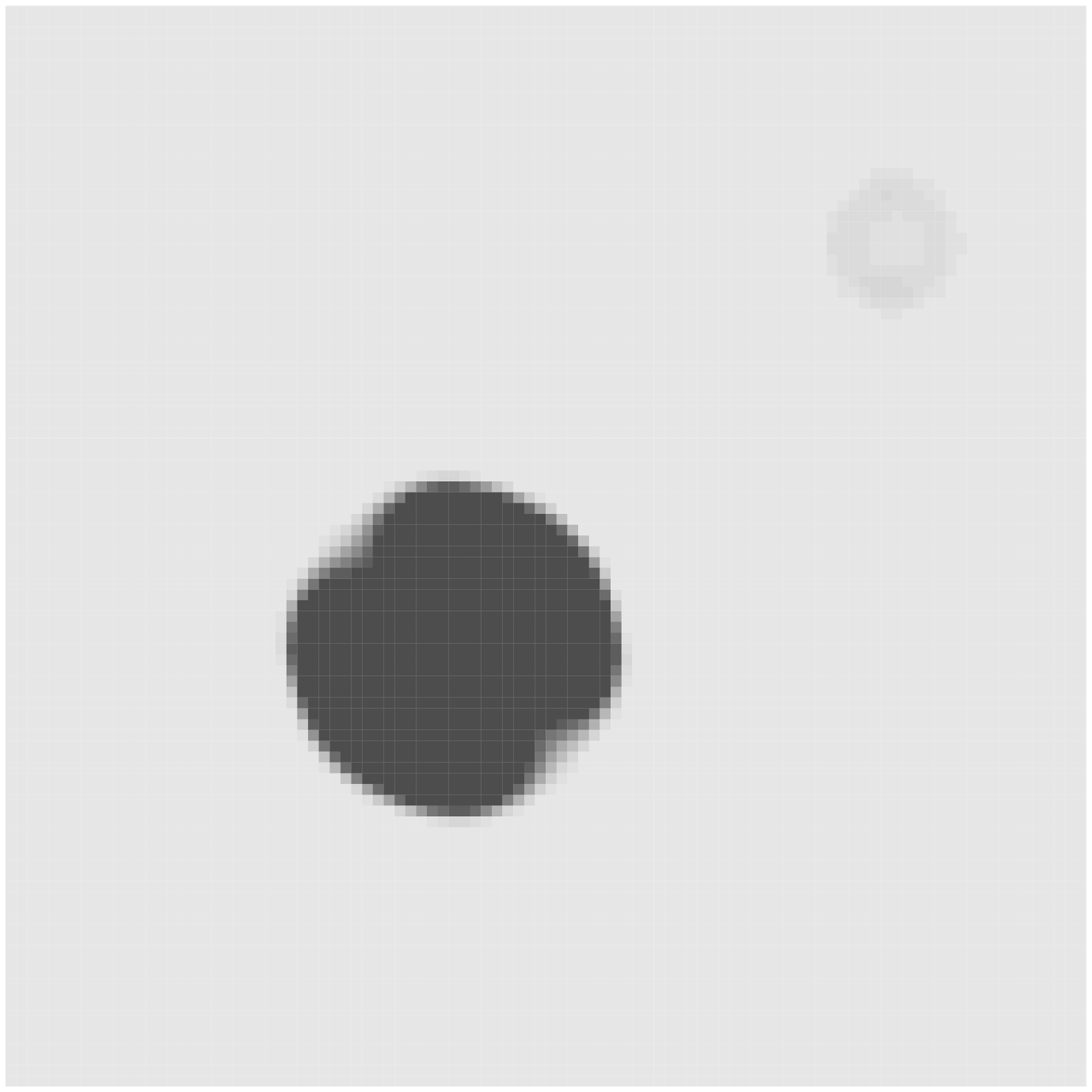}
\includegraphics[width=\textwidth]{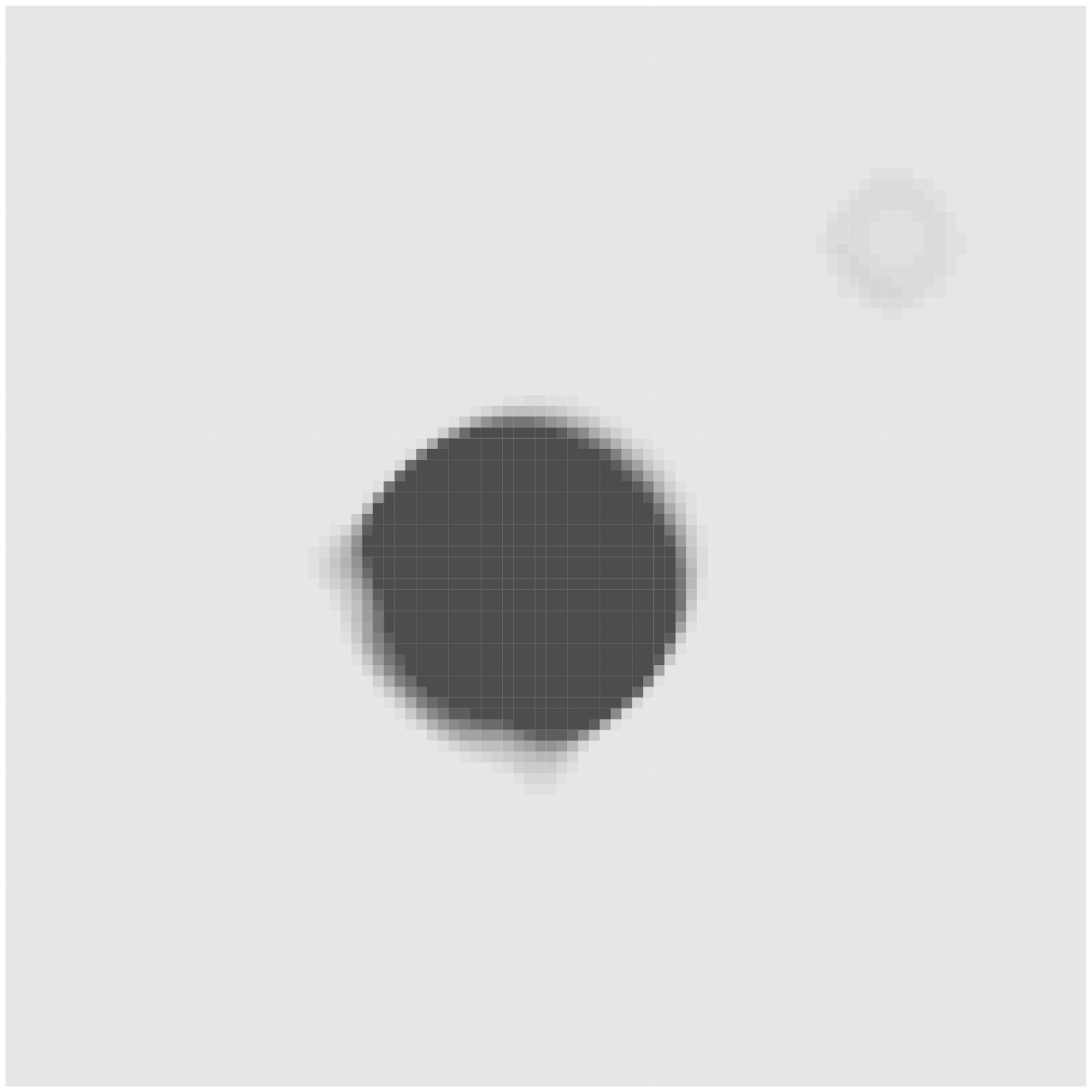}
\end{minipage}
\begin{minipage}{0.33\textwidth}
\includegraphics[width=\textwidth]{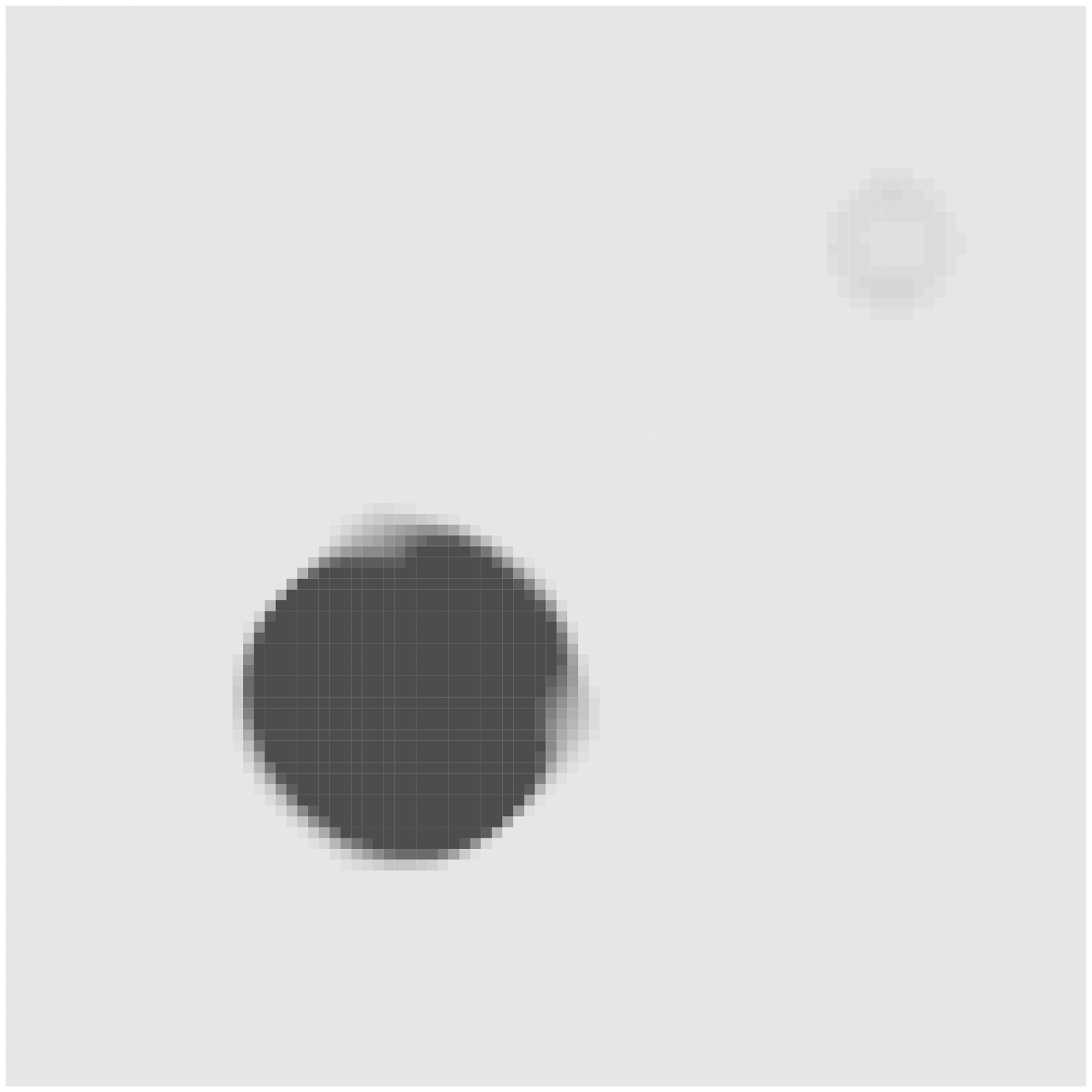}
\includegraphics[width=\textwidth]{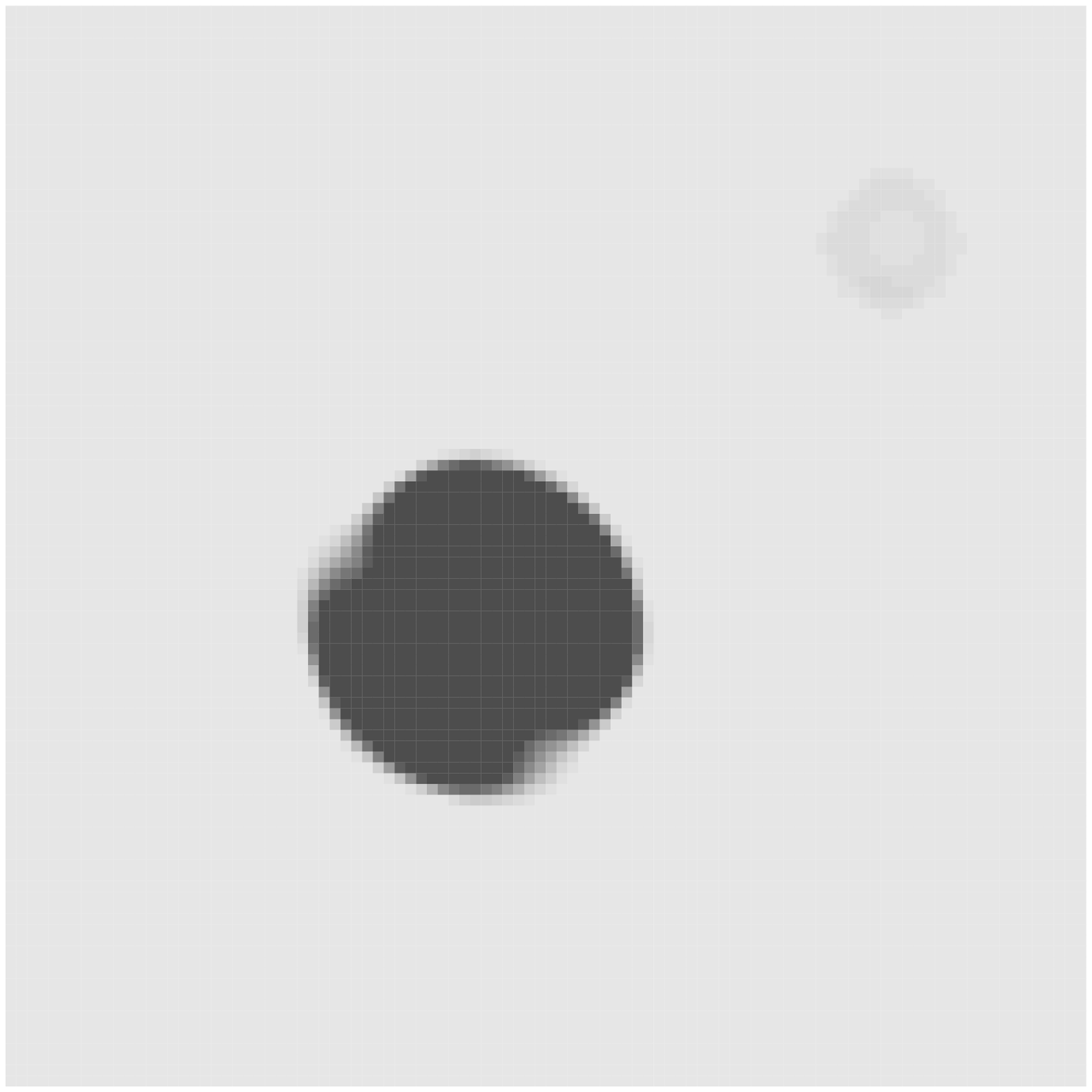}
\includegraphics[width=\textwidth]{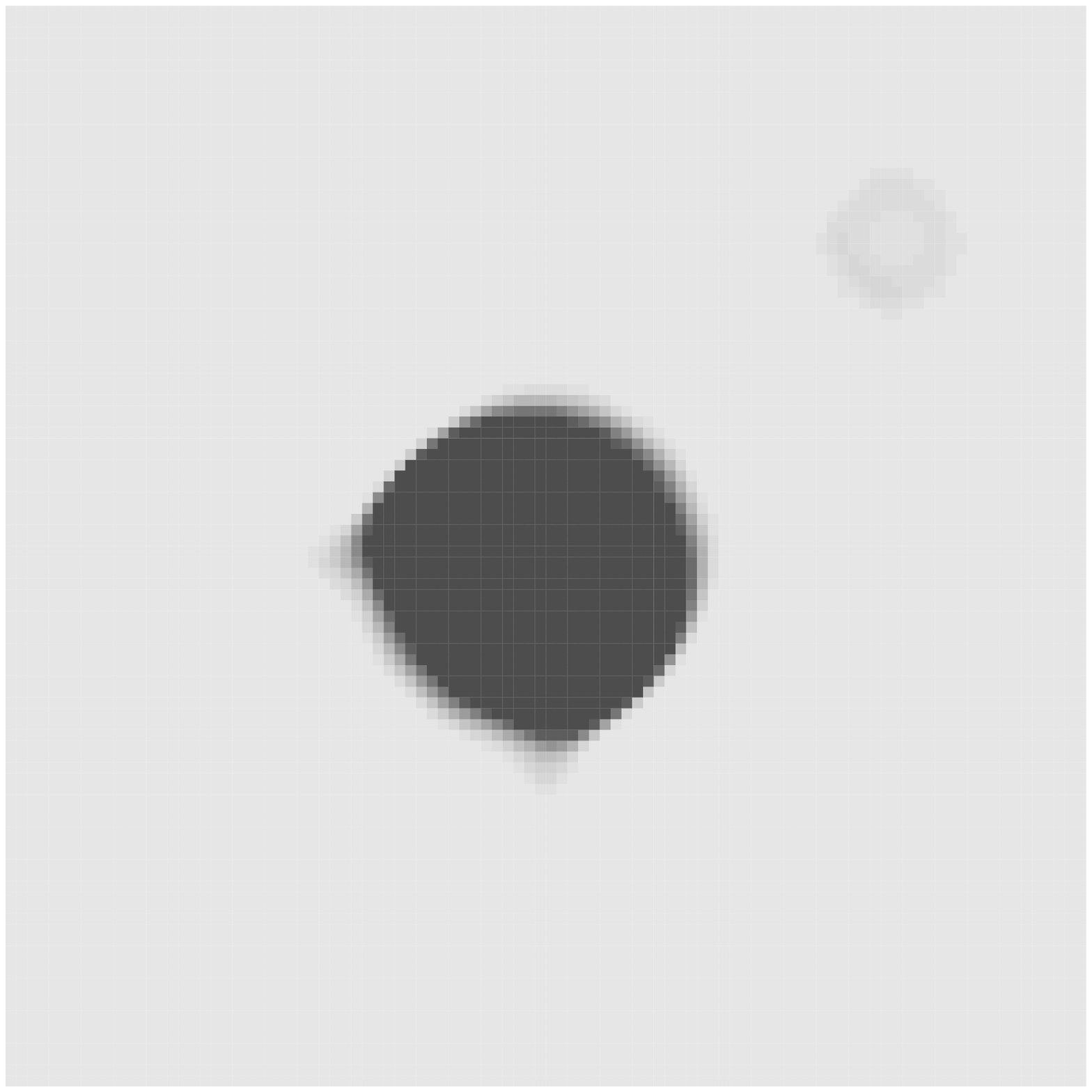}
\end{minipage}
\caption{Estimated probabilities of location of object in a noisy environment. The gray scale represents a probability of being a black pixels where dark is close to $1$ and light is close to $0$.}\label{fig:ballp}
\end{figure}

\begin{figure}[H]
\begin{minipage}{0.33\textwidth}
\includegraphics[width=\textwidth]{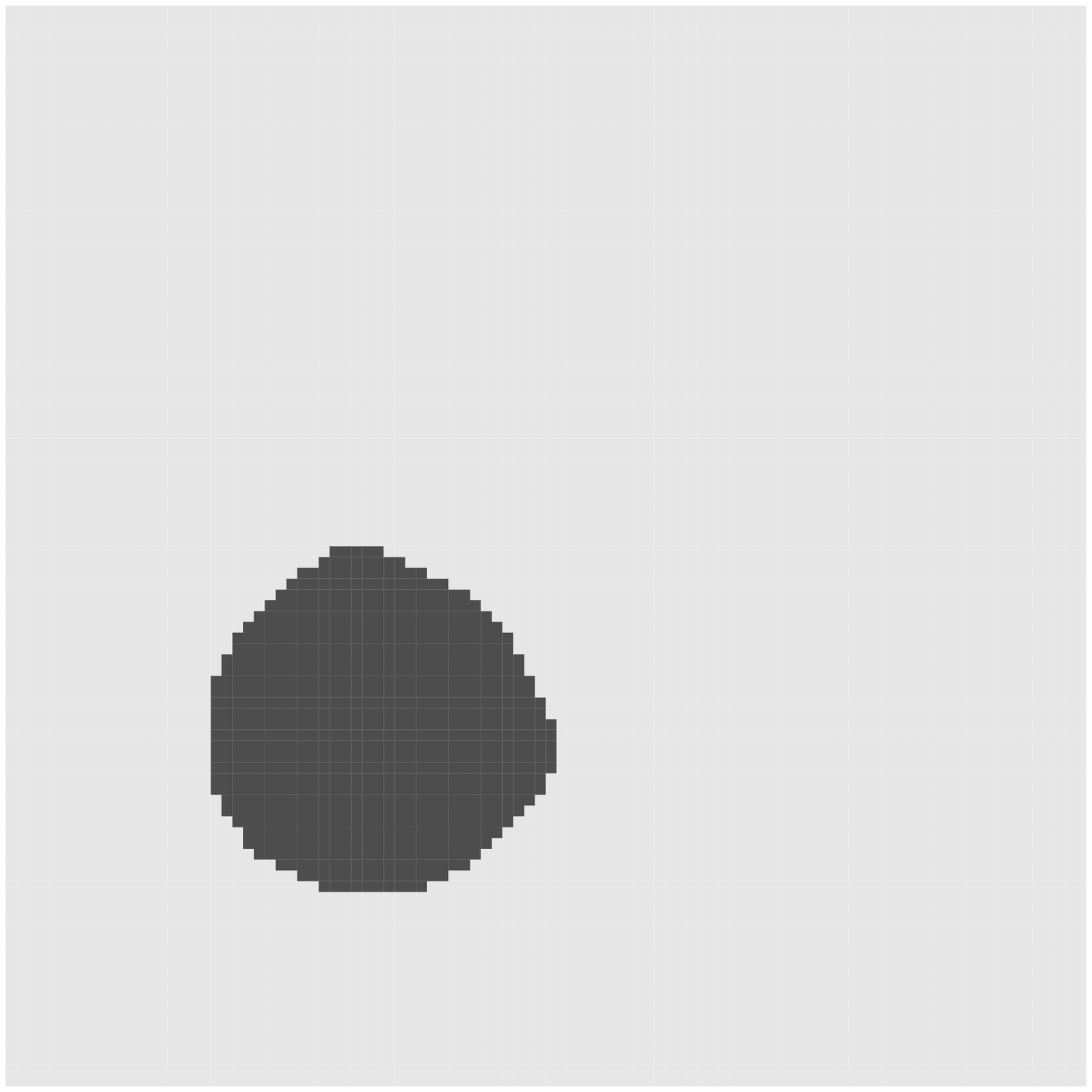}
\includegraphics[width=\textwidth]{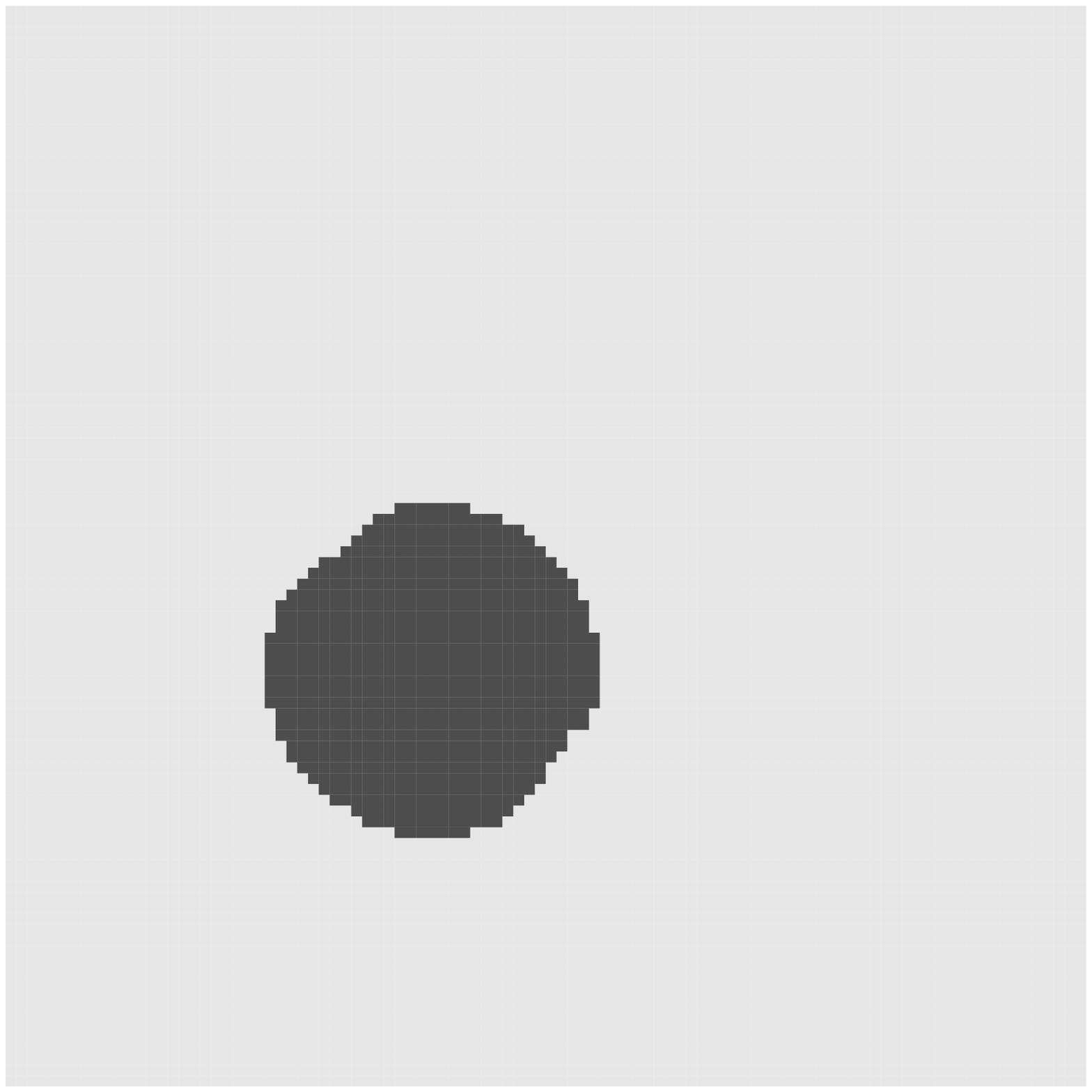}
\includegraphics[width=\textwidth]{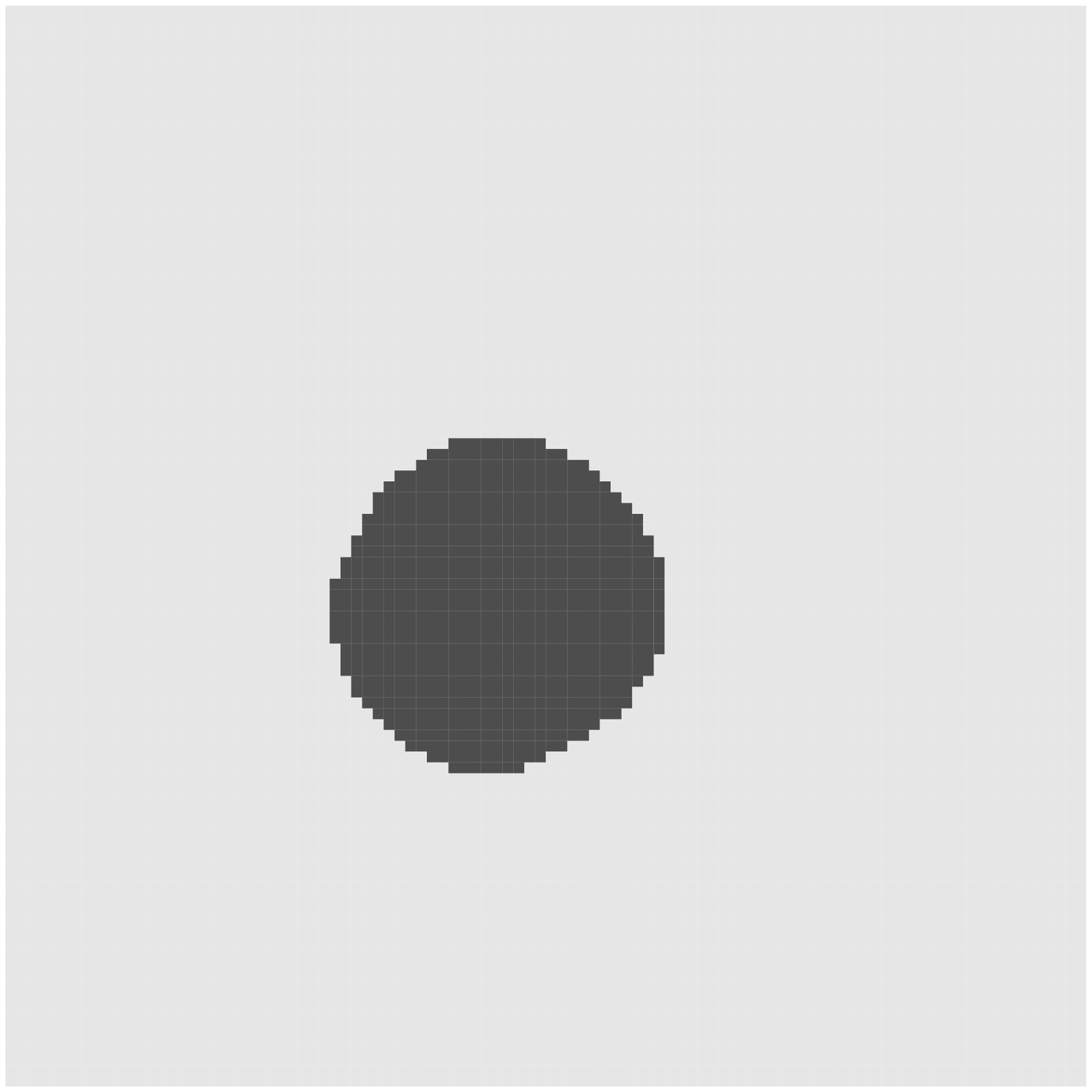}
\end{minipage}
\begin{minipage}{0.33\textwidth}
\includegraphics[width=\textwidth]{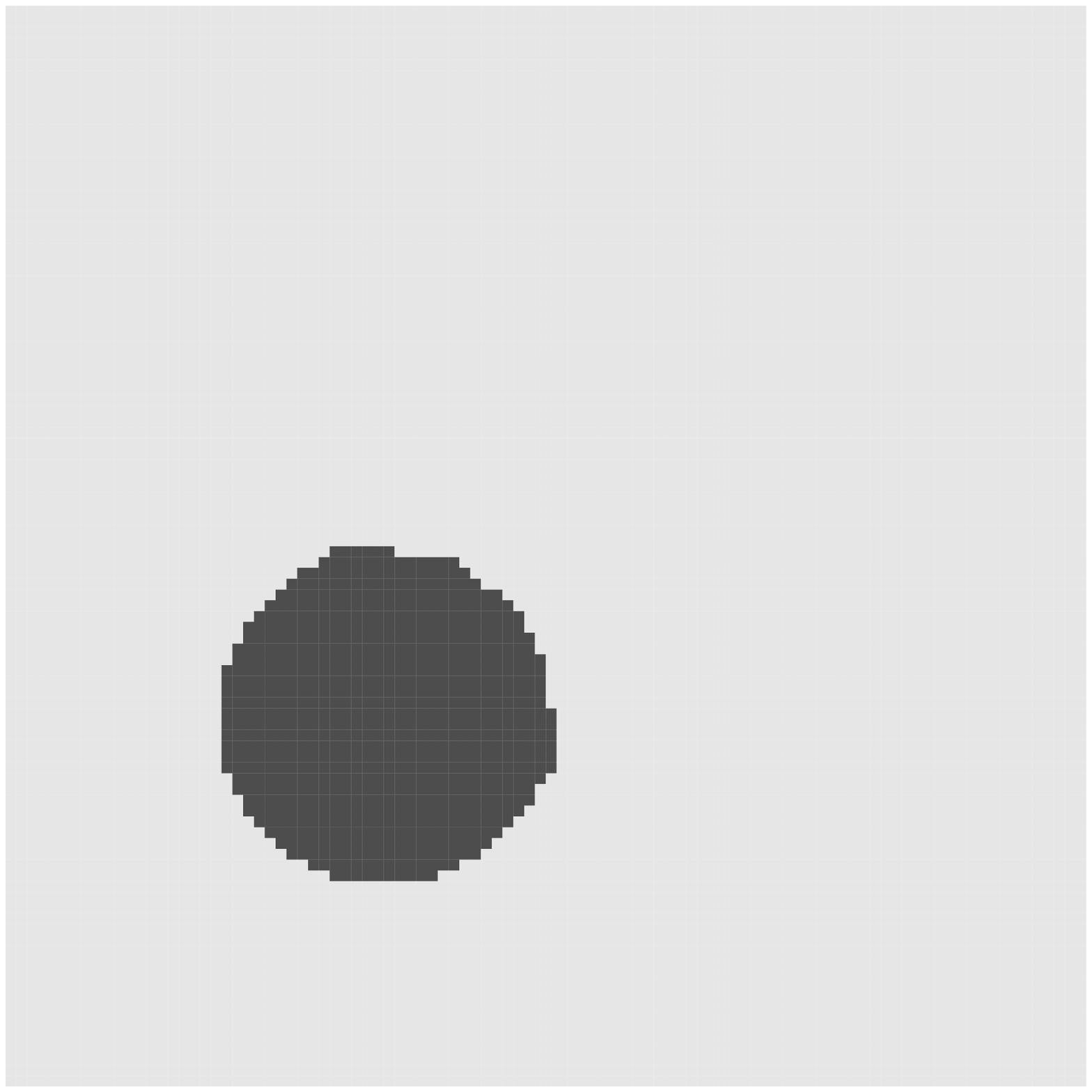}
\includegraphics[width=\textwidth]{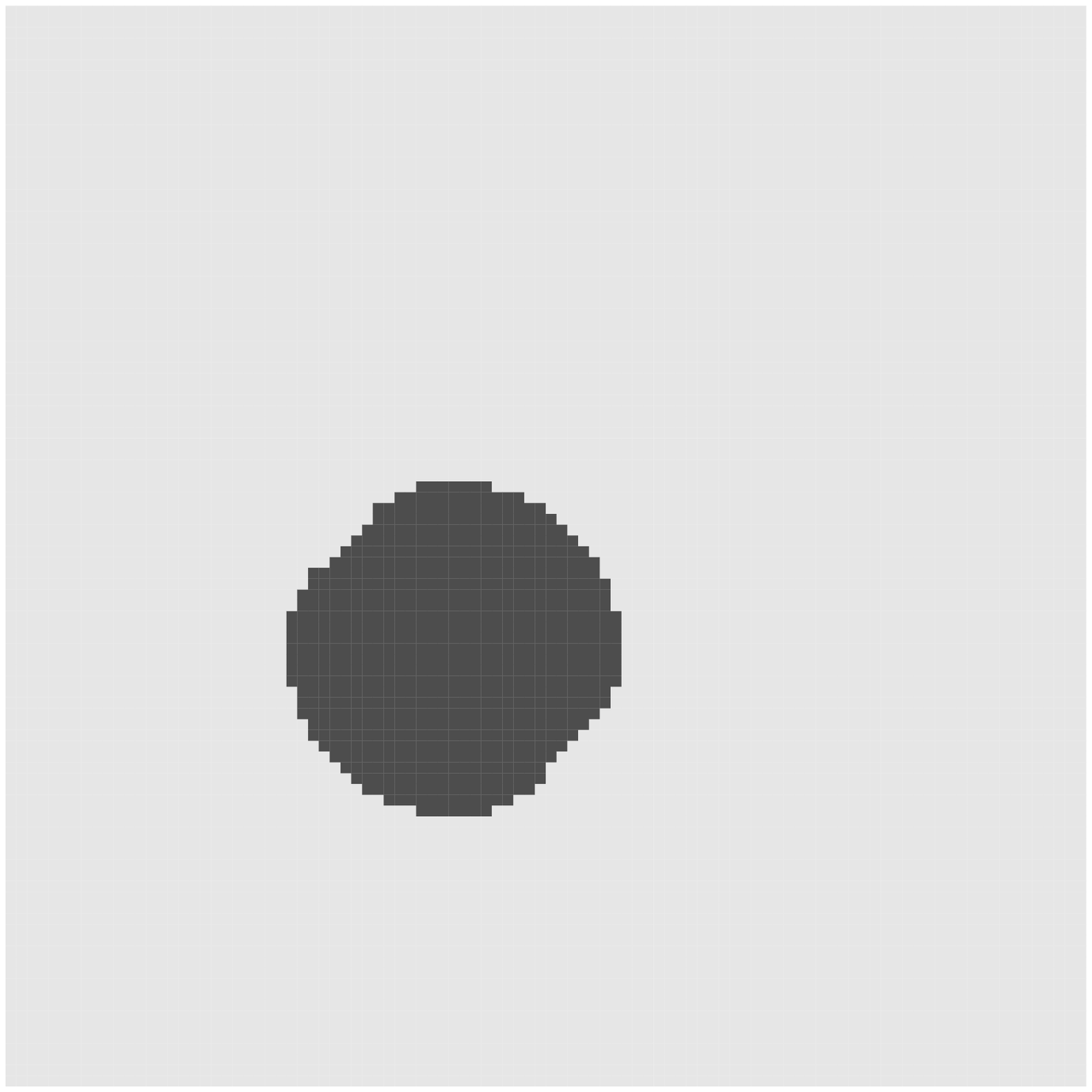}
\includegraphics[width=\textwidth]{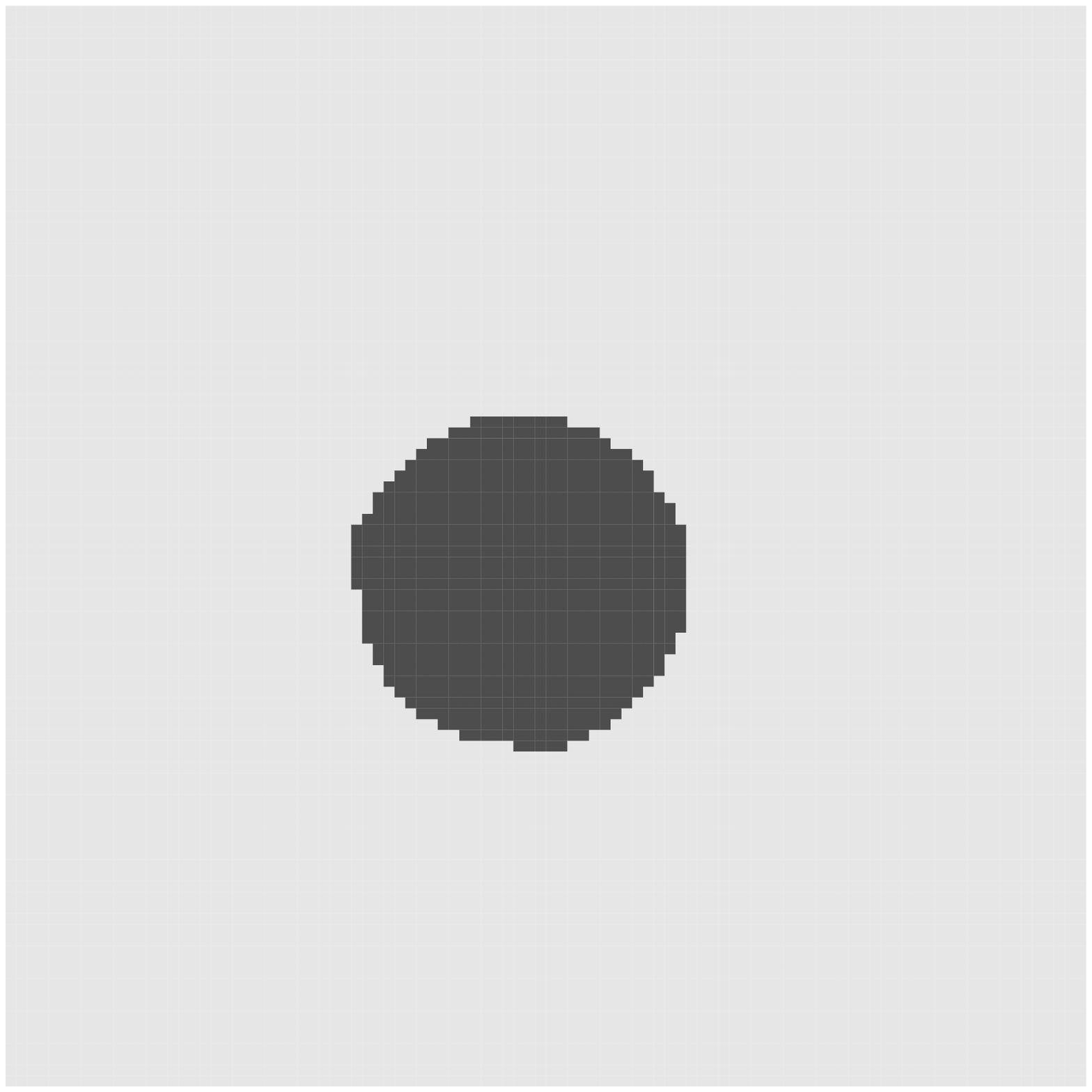}
\end{minipage}
\begin{minipage}{0.33\textwidth}
\includegraphics[width=\textwidth]{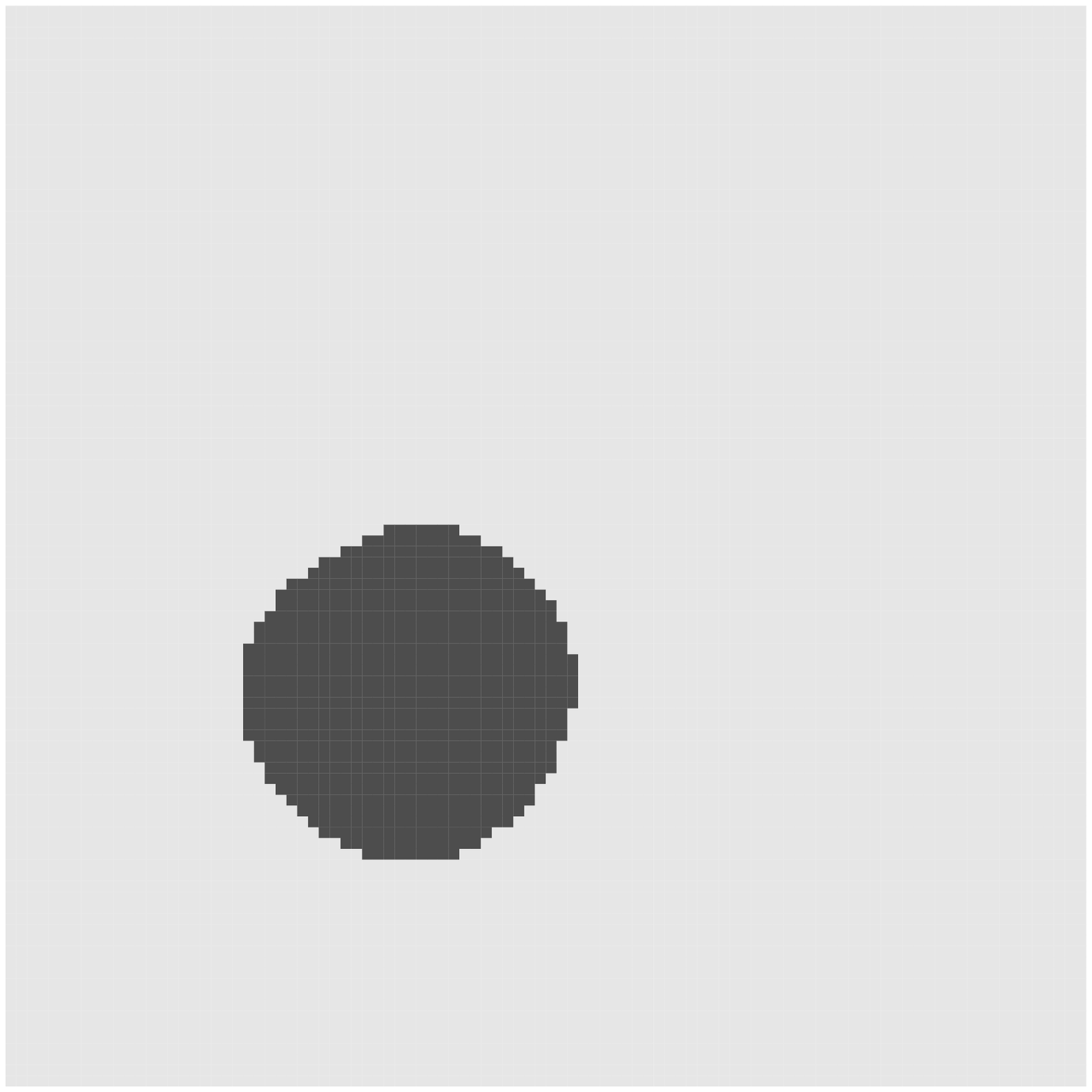}
\includegraphics[width=\textwidth]{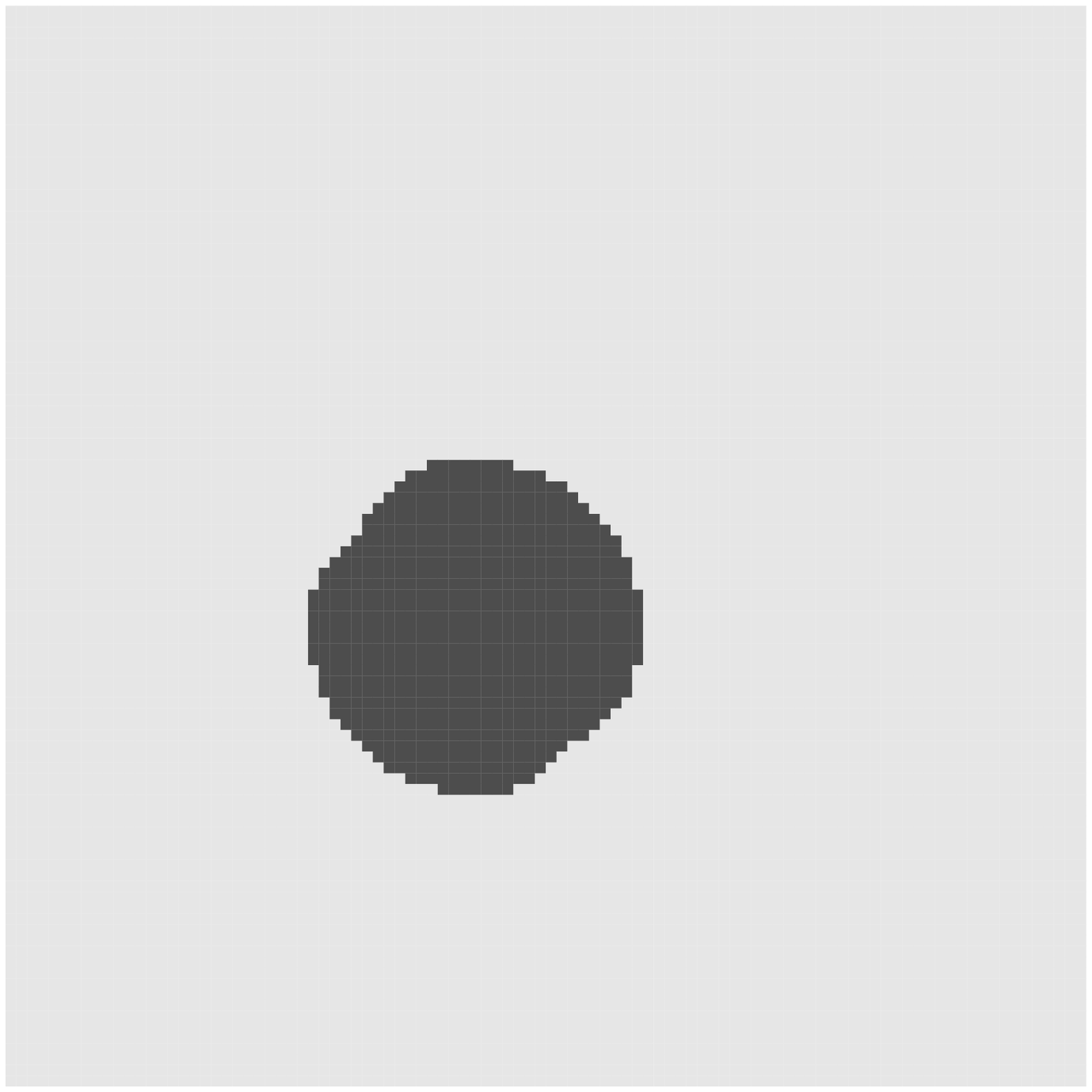}
\includegraphics[width=\textwidth]{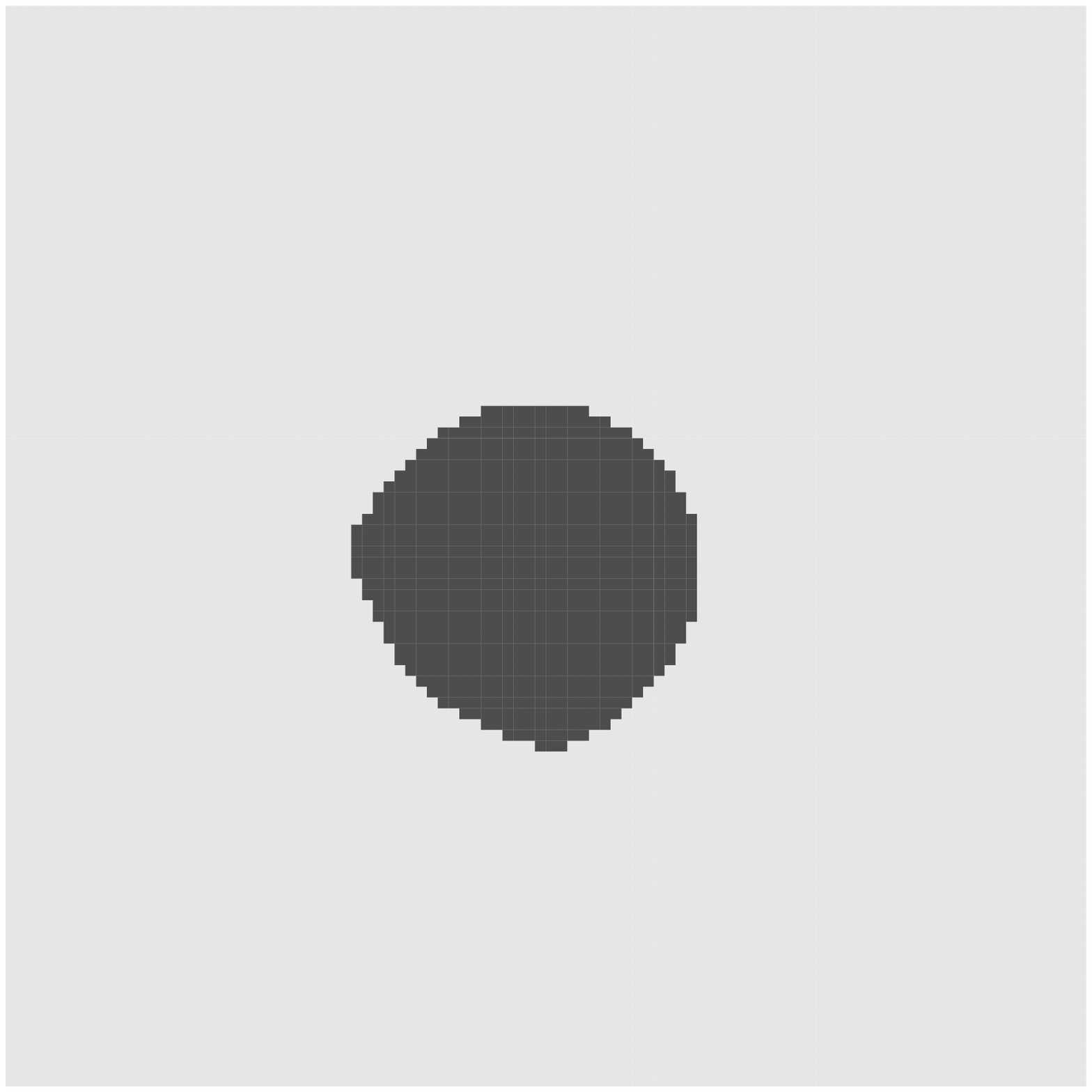}
\end{minipage}
\caption{Estimated location of object in a noisy environment. We truncated the probabilities from Figure \ref{fig:ballp} at $0.5$ to decide whether the pixels belong to the object or not.}\label{fig:ballest}
\end{figure}

An  advantage of the Bayes procedure that we use is that we obtain credible intervals as indicators of uncertainty in our prediction. The width of these intervals per pixel are shown in Figure \ref{fig:ballq}. We observe that the uncertainty around the boundary of the object is relatively high, whereas it is relatively low inside and outside of the object.

\begin{figure}[H]
\begin{minipage}{0.33\textwidth}
\includegraphics[width=\textwidth]{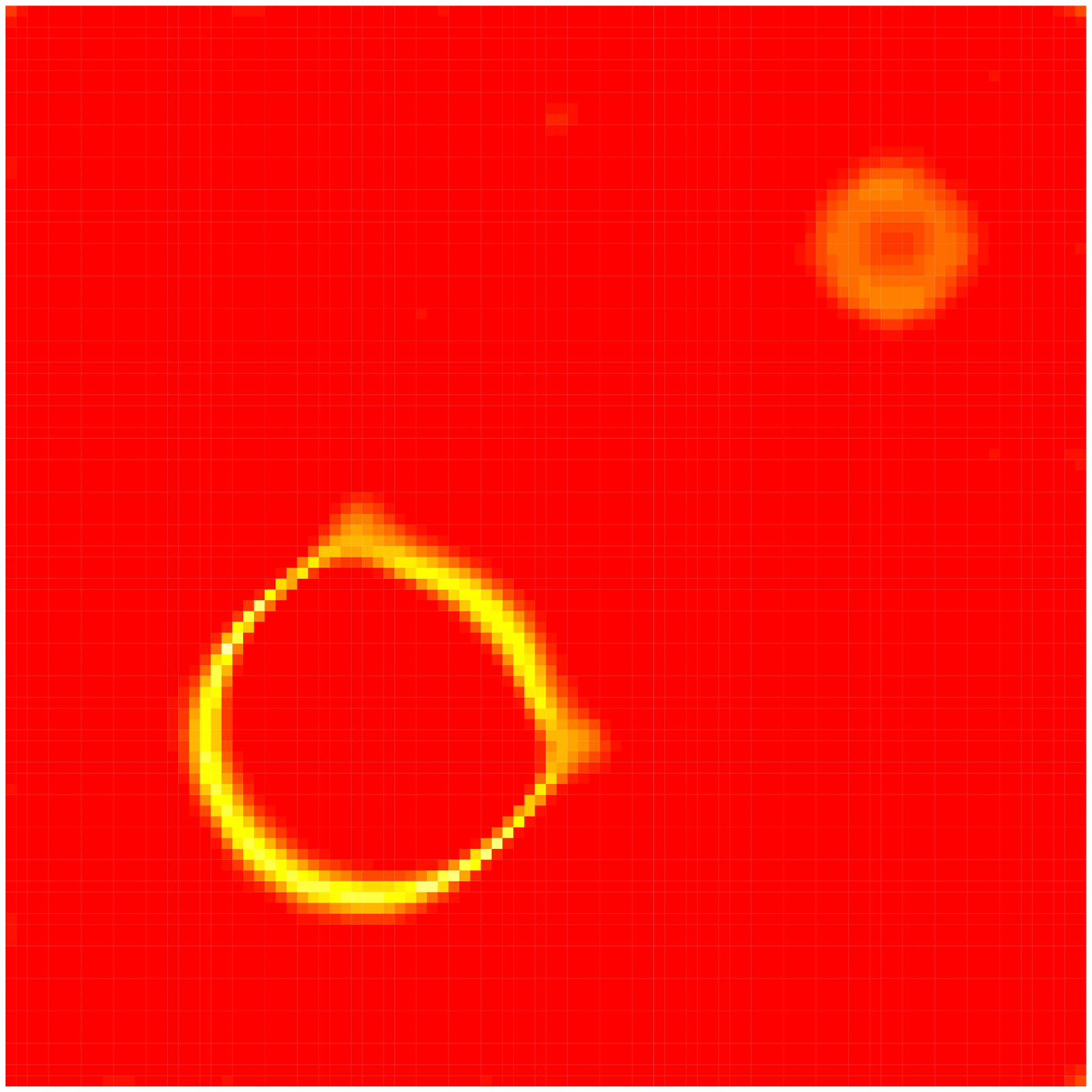}
\includegraphics[width=\textwidth]{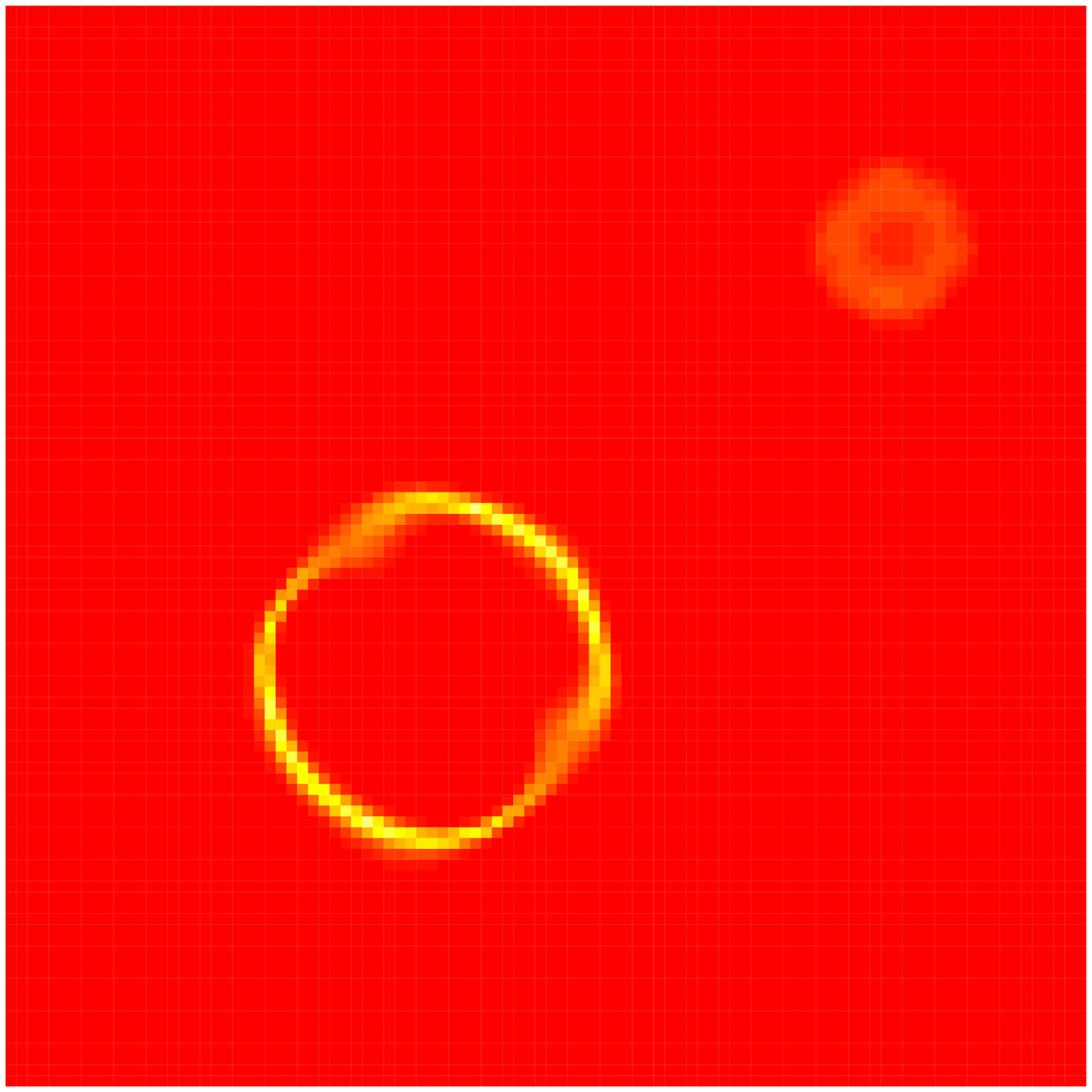}
\includegraphics[width=\textwidth]{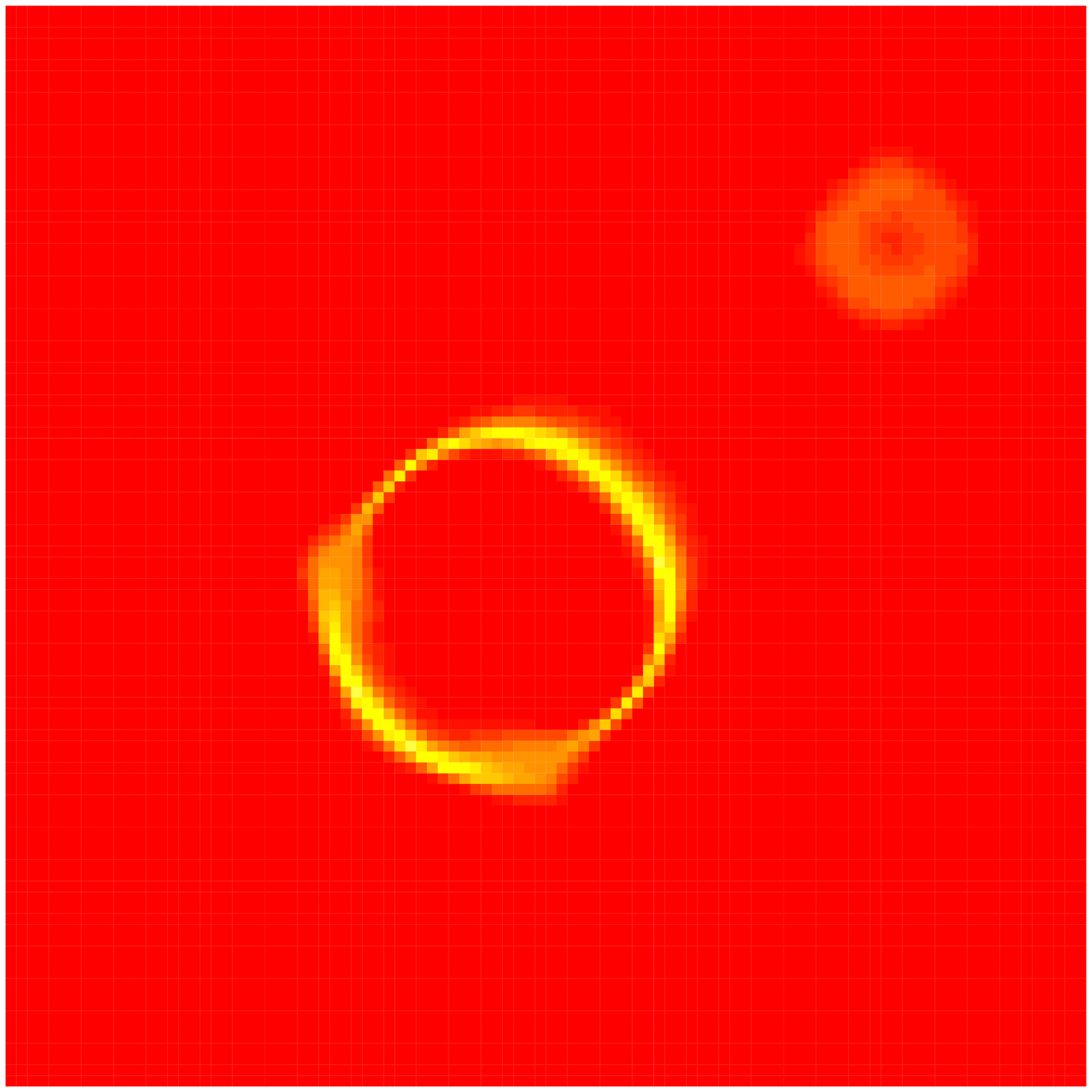}
\end{minipage}
\begin{minipage}{0.33\textwidth}
\includegraphics[width=\textwidth]{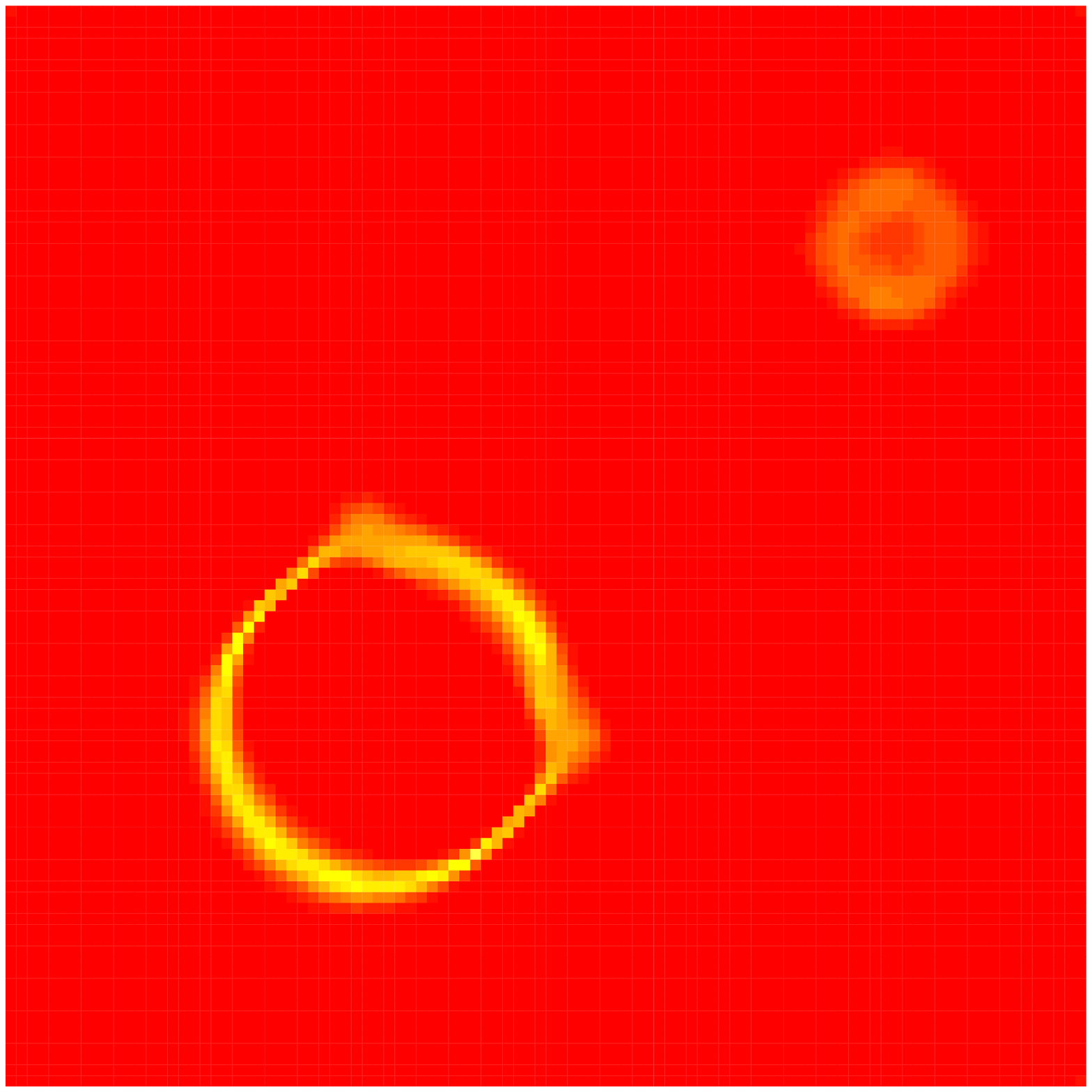}
\includegraphics[width=\textwidth]{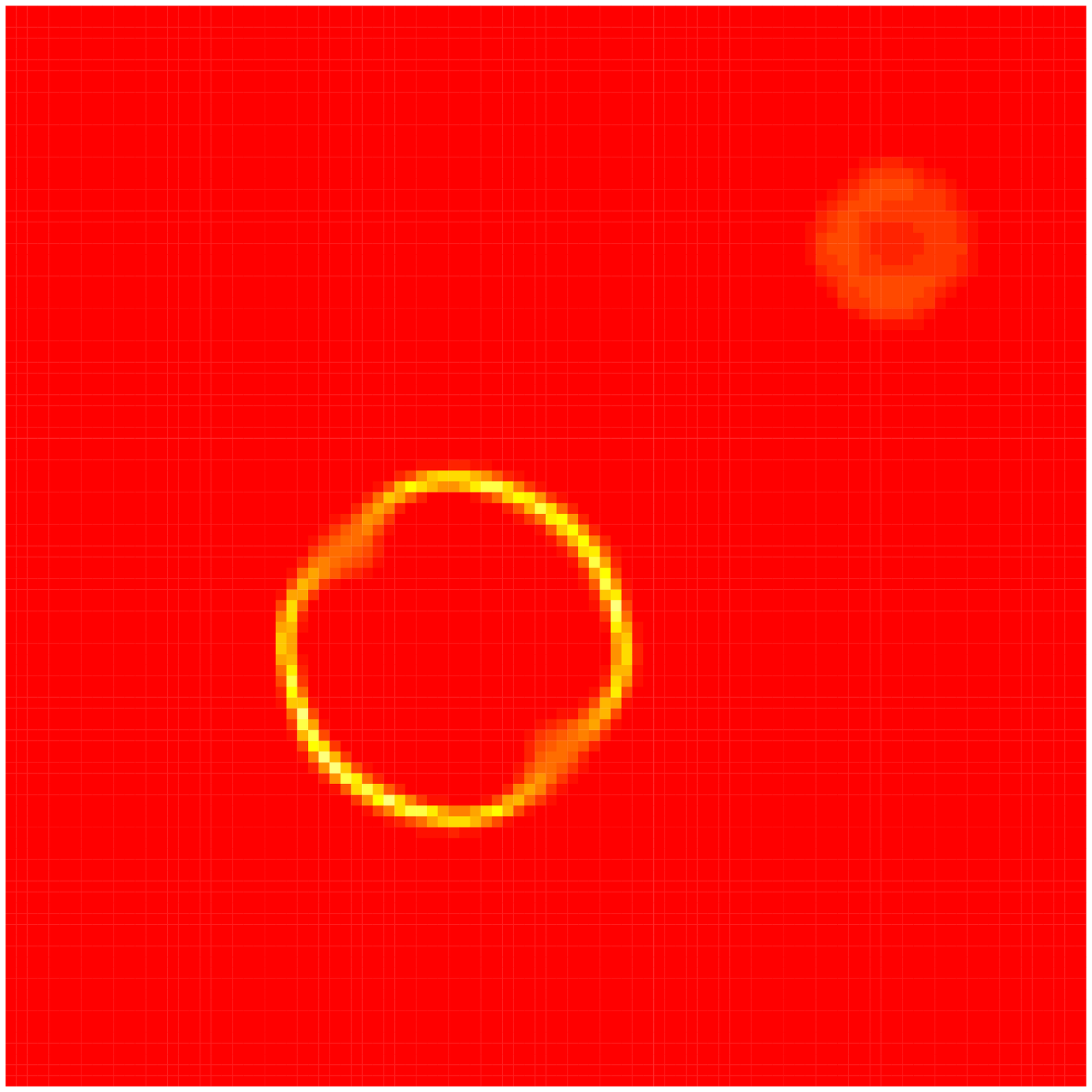}
\includegraphics[width=\textwidth]{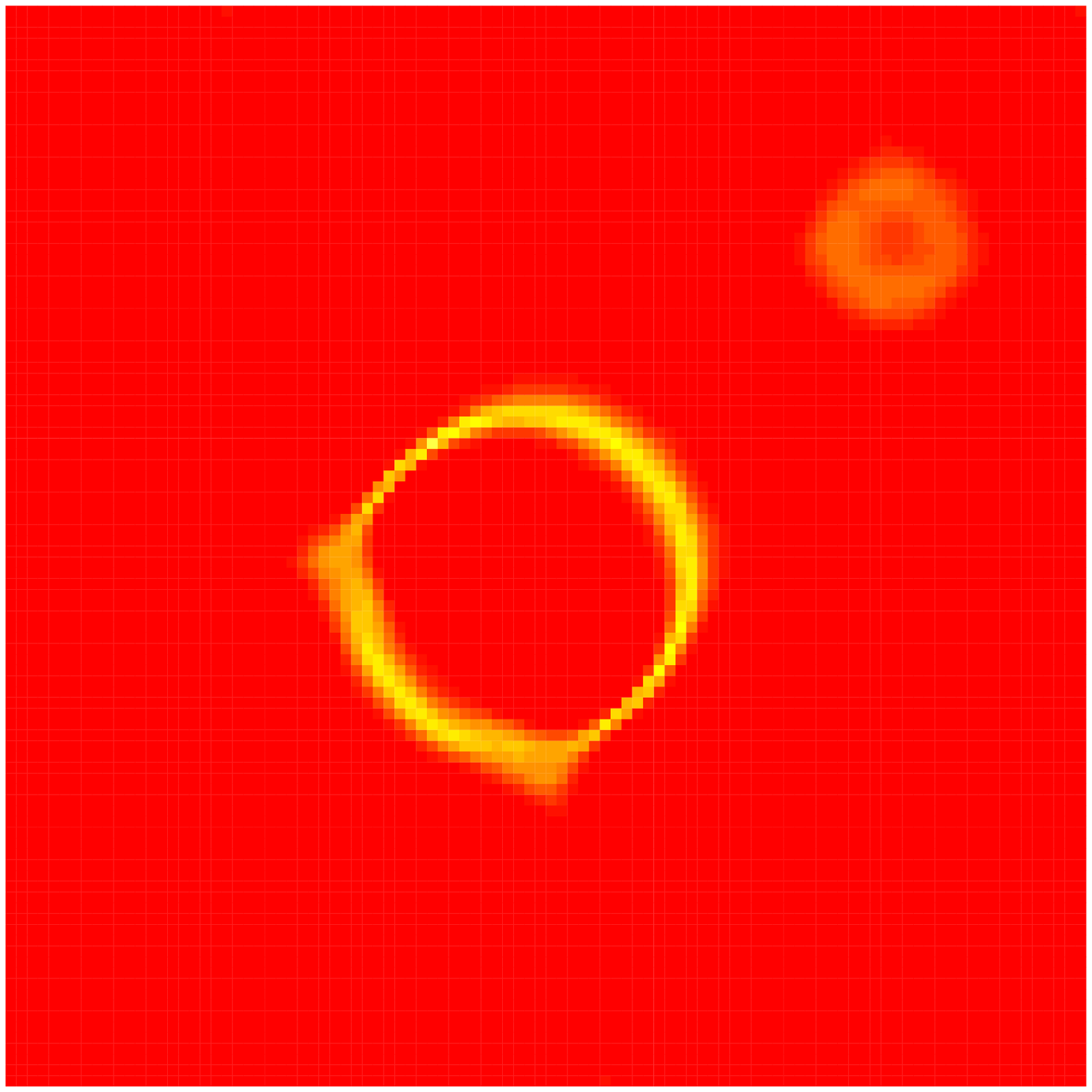}
\end{minipage}
\begin{minipage}{0.33\textwidth}
\includegraphics[width=\textwidth]{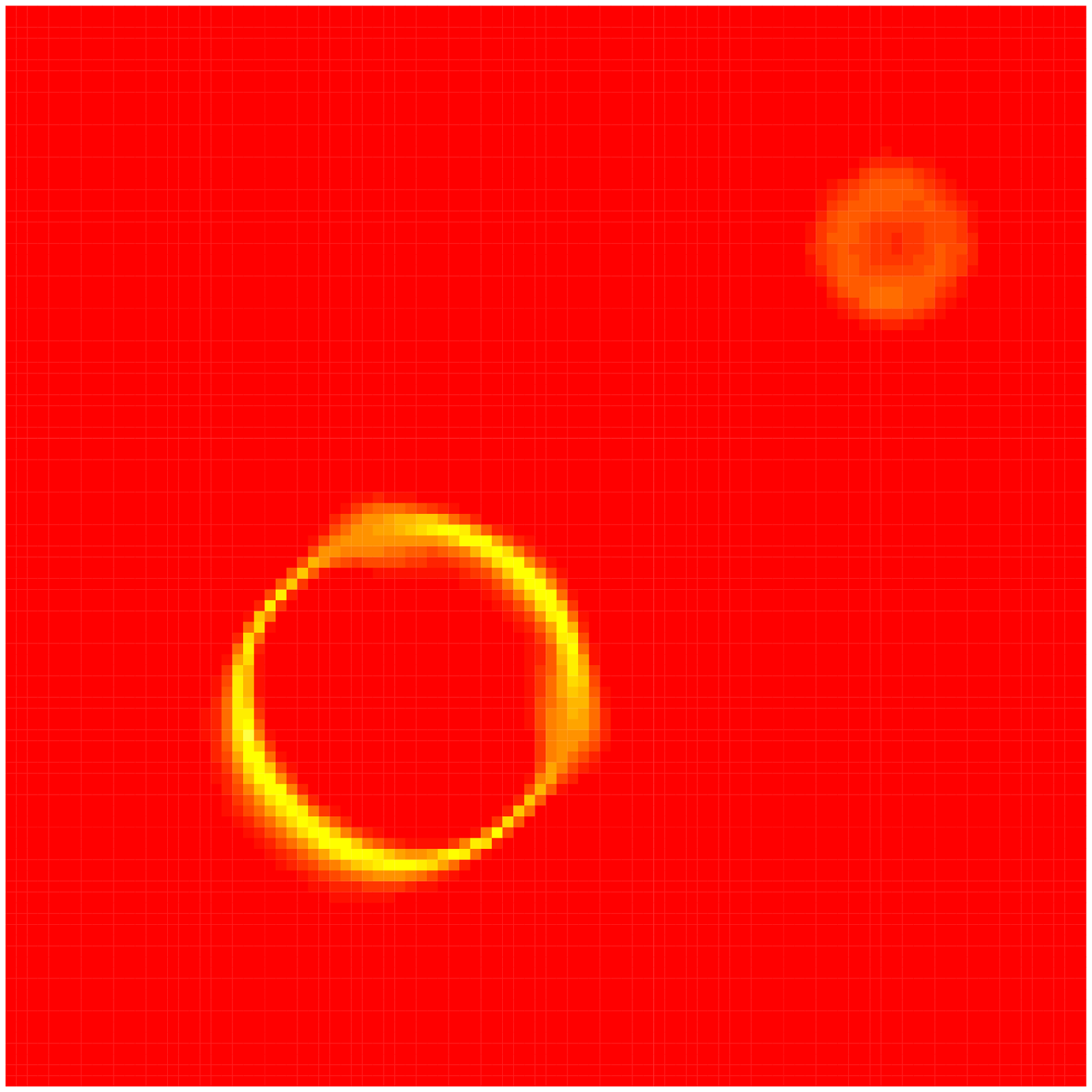}
\includegraphics[width=\textwidth]{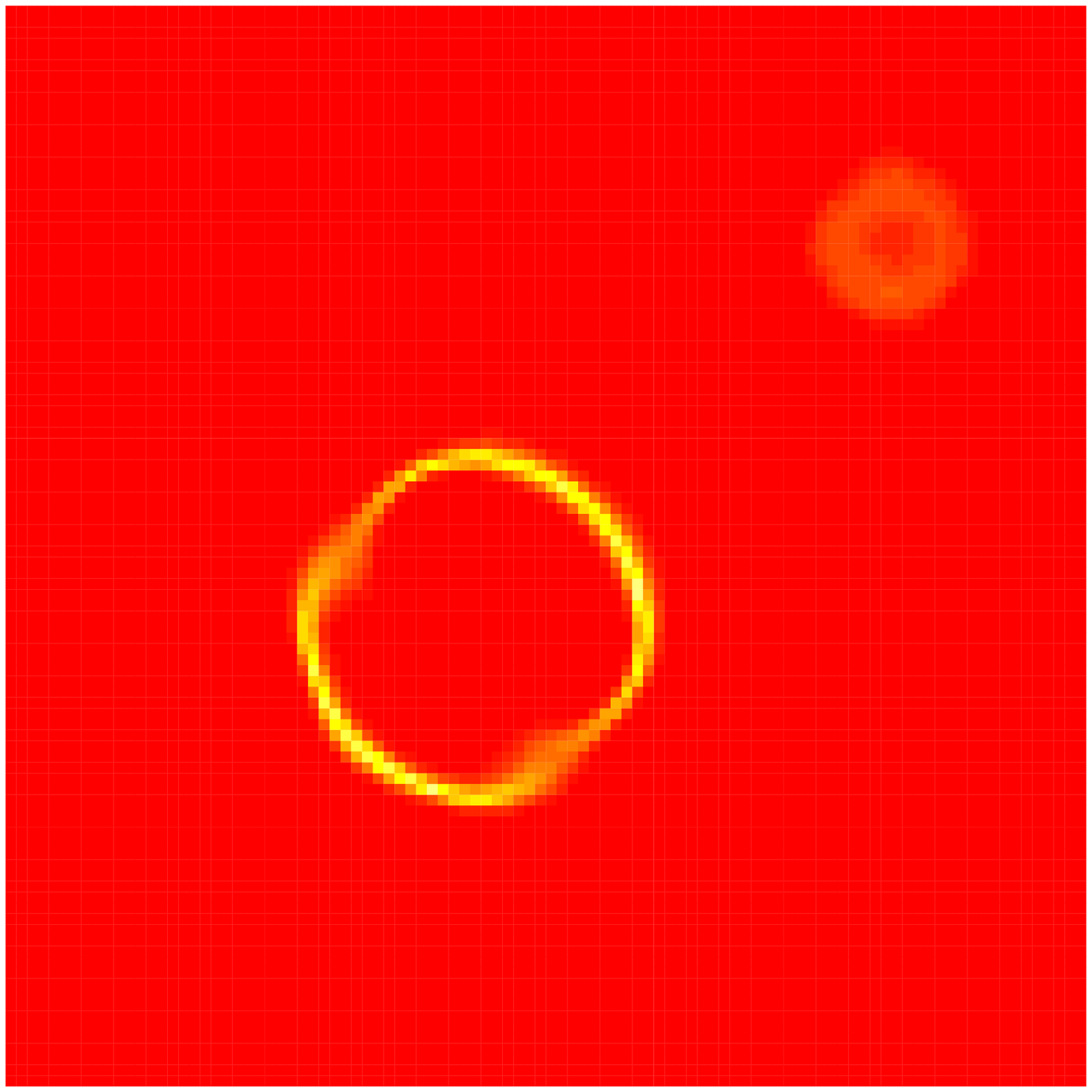}
\includegraphics[width=\textwidth]{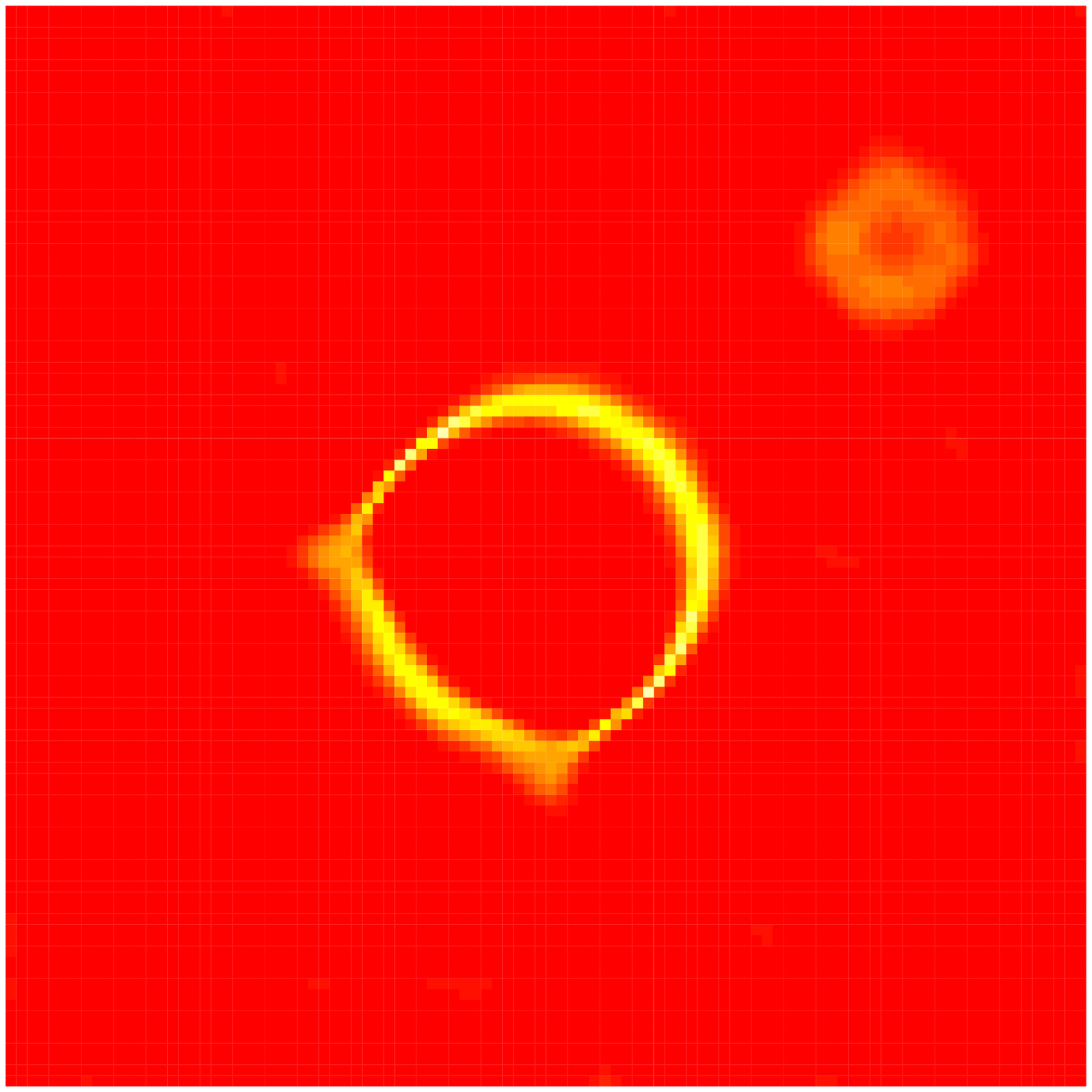}
\end{minipage}
\caption{Uncertainty in of the predicted probability for each pixel computed as the width of the $95\%$ credible interval. The heat map is from red (low uncertainty, width of credible interval close to $0$) to white (high uncertainty, width of credible interval close to $0.25$). }\label{fig:ballq}
\end{figure}

\section{Concluding remarks}

We have described an implementation of a nonparametric Bayesian approach to solving binary classification problems on graphs. We have considered a hierarchical Bayesian approach with a randomly scaled Gaussian series prior as in \cite{jarno}, but  with a random truncation point. We have implemented the procedure using a reversible jump Markov chain Monte Carlo algorithm. 

Our numerical experiments suggest that good results are obtained using Algorithm \ref{alg:mcmc} using hyperparameters $a=b=0$ and $\gamma = 20/n$. We find that in the examples we studies, the random truncation point results in a superior performance compared to the method proposed in \cite{jarno}
in terms of computational effort, while the prediction performance remains 
comparable.  We have also demonstrated that our proposed method is scalable to large graphs.

\end{document}